\documentclass[10pt,reqno]{amsart}

\usepackage{amsmath,amssymb,amsfonts,amsthm}
\usepackage{amssymb,url,color,booktabs,nccmath}
\definecolor{red}{rgb}{0,0,0}
\usepackage{microtype}
\usepackage{hyperref}
\hypersetup{colorlinks=true,linkcolor=blue,citecolor=blue,urlcolor=blue}
\allowdisplaybreaks[2]
\usepackage[hmargin=0.95in,vmargin=1in]{geometry}
\newcommand{\T}{\mathbb T}
\newcommand{\N}{\mathbb N}

\newcommand{\Z}{\mathbb Z}

\newcommand{\F}{\mathfrak F}
\newcommand{\cN}{\mathcal N}
\newcommand{\Tr}{\mathrm{Tr}}

\newcommand{\e}{\mathrm e}
\newcommand{\ii}{\mathrm i}
\renewcommand{\d}{\,\mathrm d}
\newcommand{\dG}{\mathrm d\Gamma}

\newcommand{\gH}{\mathfrak{H}}
\newcommand{\gF}{\mathfrak{F}}

\newcommand{\Fock}{\mathfrak F}
\newcommand{\hs}{\otimes_s}
\newcommand{\1}{\mathbf{1}}

\newcommand{\rmk}[1]{#1}
\newcommand{\vbeta}{\rmk{v_\beta}}
\newcommand{\hvb}{\rmk{\widehat v_\beta}}
\newcommand{\Wbeta}{\rmk{W_{\beta}^{(2)}}}
\newcommand{\vbstar}{\rmk{v_{\beta,*}}}

\numberwithin{equation}{section}

\theoremstyle{plain}
\newtheorem{theorem}{Theorem}[section]
\newtheorem{lemma}[theorem]{Lemma}
\newtheorem{remark}[theorem]{Remark}

\newtheorem{proposition}[theorem]{Proposition}

\newtheorem{corollary}[theorem]{Corollary}

\makeatletter
\renewcommand{\paragraph}{\@startsection{paragraph}{4}{\z@}%
  {0.8\baselineskip \@plus 0.2\baselineskip \@minus 0.1\baselineskip}%
  {-0.5em}%
  {\normalfont\bfseries}}
\makeatother

\title[Derivation of Gibbs measure from Gibbs state with the fractional Bessel interaction in Two Dimensions]{\rmk{Derivation of Gibbs measure from Gibbs state with the fractional Bessel interaction in Two Dimensions}}
\author{PHAN TH\`ANH NAM, RONGCHAN ZHU, AND XIANGCHAN ZHU}
\date{}

\begin{document}
\begin{abstract}
	We derive the classical Gibbs measure on $\T^2$ associated with the fractional Bessel interaction potential $\hvb(k)=\langle k\rangle^{-\beta}$ from a renormalized grand-canonical quantum Bose gas with the same interaction. Our result covers the whole range $\frac32<\beta\le2$, where $\hvb(k)$ is  not summable and the quantum model cannot be written in the usual density-square form, as the associated self-energy diverges.  We therefore need to renormalize the zero mode by a centered number-fluctuation term and then develop a detailed analysis for the high-frequency remainders. All this allows us to implement a low-frequency localization and obtain the convergence of the quantum relative free energy to the classical fractional-Bessel free energy, as well as the convergence of the reduced density matrices to the limiting Gibbs measure.
	
	
\end{abstract}
\maketitle
\thispagestyle{plain}

\setcounter{tocdepth}{2}
\tableofcontents
\bigskip

\section{Introduction}\label{sec:introduction}

\paragraph{Background and motivation.}
The derivation of nonlinear Gibbs measures from many-body quantum Gibbs states provides a precise bridge between bosonic quantum statistical mechanics and classical random fields; see, for instance, \cite{LewNamRou-15,LewNamRou-21,FKSS22,FKSS23,NZZ25,JouRou25,CKRTG26, LNZ26}. In this program one starts from a grand-canonical quantum Bose gas at temperature of order $\lambda^{-1}$, lets $\lambda\downarrow0$, and asks whether the interacting quantum Gibbs state converges to a classical Gibbs measure built on the Gaussian free field. {\rmk{The present paper addresses this question for the two-dimensional fractional-Bessel interaction on the torus, in the full range $\frac32<\beta\le2$.}} Our goal is not only to put the renormalized quantum model in a workable form, but to complete the quantum-to-classical derivation at the level of free energy and reduced density matrices.

{\rmk{The periodic fractional Bessel potential $v_\beta:\T^2\to \mathbb{C}$ is given by its Fourier coefficients 
		\begin{align}\label{eq:Bessel-potential}
			\hvb(k)=\langle k\rangle^{-\beta}=(1+|k|^2)^{-\beta/2},
			\qquad k\in\Z^2,
		\end{align}
		with a fixed exponent $\frac32<\beta\le2$. Put differently, $\vbeta(x-y)$ is the kernel of $(1-\Delta)^{-\beta/2}$ on $L^2(\T^2)$. This family contains the Yukawa case  $\beta=2$, which is as singular as the Coulomb potential, and also allows even stronger singularities. In particular, the potential remains non-summable in two dimensions throughout the whole range considered here. 
		
		On the one hand, $\hvb(k)\ge0$ and the corresponding classical Gibbs measure can be constructed directly (see Lemma~\ref{lem:class}) or within the variational approach of \cite{BG20}; for Hartree-type fields this program has been implemented in \cite{Bri22,OOT24}. Classical fields of this kind can also be constructed through stochastic quantization, starting from \cite{PW81} and developed in \cite{JLM85,DD03,Hai14,Kup16,RZZ17a,TW18,MW17a,MW17,CC18,GH19,GH21}; see also \cite{NZZ25} for a recent quantum-to-classical derivation in which stochastic-quantization methods play an essential role. On the other hand, this non-summability means that the standard density-square rewriting of the quantum interaction produces an infinite self-energy. The present model therefore lies beyond the reach of the previous derivations from quantum Gibbs states with fixed pair interactions in the nonlocal setting, which do not cover the regime $\sum_k \hvb(k)=\infty$. The point is precisely to treat a kernel that is already singular at the quantum level: previous derivations for fixed pair interactions concern bounded or otherwise sufficiently regular potentials, whereas here the ultraviolet singularity survives in the quantum Hamiltonian itself.}}

\paragraph{The quantum model and the derivation problem.}
We work on the bosonic Fock space $\F(L^2(\T^2))$ and consider the renormalized Gibbs state
\[
\Gamma^{\rm re}_\lambda=(Z^{\rm re}_\lambda)^{-1}e^{-\mathbb H^{\rm re}_\lambda},
\qquad
\mathbb H^{\rm re}_\lambda=\lambda\dG(h)+\Wbeta+c_\lambda\cN+C_\lambda,
\qquad h=-\Delta+1.
\]
{\rmk{Here $\Wbeta$ is the genuine two-body fractional-Bessel interaction in second quantization, i.e. on the \(n\)-particle sector \(\gH_n\), it acts as}}
\[
\lambda^2\sum_{1\le i<j\le n}\vbeta(x_i-x_j).
\]
Here  $\cN$ is the number operator and $c_\lambda$, $C_\lambda$ are constants, which are chosen in Section \ref{sec:ren-int} below. 

The limiting object is the classical Gibbs measure associated with $\vbeta$
\[
\d\nu(u)=z_{\mathrm{cl}}^{-1}e^{-D(u)}\,\d\mu_0(u),
\]
where $\mu_0$ is the Gaussian free field with covariance $h^{-1}$ and {\rmk{$D(u)$ is the renormalized classical interaction associated with $\vbeta$, obtained as the infinite-dimensional limit of the finite-dimensional Hartree functionals}}
\begin{align}\label{def:D}
	D(u):=\frac12\int_{\T^2\times \T^2} \vbeta(x-y):\!|u(y)|^2\!::\!|u(x)|^2\!:\d x\d y,
\end{align}
where $:\!|u(x)|^2\!:$ means the Wick product and is formally given by $|u(x)|^2-\langle |u(x)|^2\rangle_{\mu_0}$ and can be rigorously defined via suitable smooth approximation. 
The derivation problem is to show that the renormalized quantum free energy converges to $-\log z_{\mathrm{cl}}$ and that the reduced density matrices of $\Gamma^{\rm re}_\lambda$ converge to the moments of $\nu$.

The precise statement of the paper is given in Theorem~\ref{thm:main-result} below. It combines the convergence of the relative free energy with the Hilbert--Schmidt convergence of the rescaled reduced density matrices to the classical correlation operators. The scalar moment convergence against finite Fourier test functions, stated later in Theorem~\ref{thm:density-matrix-convergence-final}, is recovered as a consequence of this operator formulation.

\paragraph{Why the problem is delicate.}
There are several layers of difficulty.

First, the usual density-square rewriting already breaks down at the operator level. If one introduces the density modes
\(\rho_k=\dG(e_k),
\)
and formally rewrites the interaction in terms of $\rho_k\rho_{-k}$, the canonical commutation relations generate a self-energy proportional to
\(
(\sum_{k\in\Z^2}\hvb(k))\cN,
\)
{\rmk{which diverges because $\hvb(k)=\langle k\rangle^{-\beta}$ is not summable in dimension two throughout the whole range $\beta\le2$. The renormalized Hamiltonian must therefore be formulated directly in terms of the genuine two-body operator, with only the zero mode extracted and centered.}}

Second, the free-energy lower bound is not automatic. Once the zero mode is centered, the interaction is no longer manifestly nonnegative when compared with the original two-body Yukawa operator. Hence the lower bound needed for Theorem~\ref{thm:free-energy-convergence-final} cannot be recovered from positivity alone; it has to be reconstructed through a detailed comparison with the finite-dimensional classical problem.

Third, the high-frequency analysis is genuinely multi-scale. After the first cutoff $P=1_{\{h\le \Lambda^2\}}$, the error does not reduce to a single positive tail term. One has to separate a zero-mode remainder from the nonzero-mode contribution, and then introduce a second cutoff to split the latter into a shell part and a tail part. The shell term is neither positive nor termwise negligible and has to be reorganized into fluctuation observables. The tail part is controlled only after exploiting positivity of the full Yukawa operator together with quantitative bounds on high-frequency particle-number fluctuations.

Fourth, even the preliminary a priori theory is already borderline. The source-dependent estimates in Section~\ref{sec:var} cannot be obtained by first rewriting the interaction in a density-square form, because that is exactly where the divergent self-energy appears. One has to work directly with the genuine quartic Yukawa operator under perturbed quasi-free states, and the potentially dangerous direct Wick contribution is controlled only after exploiting the cancellation $\Tr(\gamma_0M_k)=0$ for $k\neq0$ with $M_k$ being multiplication by $e_k$. Thus the singular nature of the fixed Bessel interaction enters not only in the final lower bound, but already in the basic estimates that make the rest of the argument possible.

\paragraph{Relation to previous works.}
This paper fits into the general program initiated in \cite{LewNamRou-15} and developed further in \cite{LewNamRou-21,FKSS22,FKSS23,NZZ25,JouRou25,CKRTG26}. The recent works \cite{JouRou25,CKRTG26} derive $\Phi^4_2$-type field theories from two-dimensional quantum Gibbs states in homogeneous and trapped settings, respectively, through shrinking-range or inhomogeneous limits. By contrast, the present Bessel problem keeps a fixed nonlocal interaction whose Fourier coefficients are not summable. This is different from the summable-interaction setting because the interaction cannot be written as the formal density-square representation before the quantum-to-classical comparison: the divergent quantity $\sum_k \hvb(k)$ is precisely the obstruction. On the classical side we rely on  the usual Wick product and direct analysis for the construction of the Gibbs measure; related nonlinear fields may also be accessed through the variational analysis \cite{BG20, Bri22, OOT24} or stochastic quantization \cite{PW81,JLM85,DD03,Hai14,Kup16,RZZ17a,TW18,MW17a,MW17,CC18,GH19,GH21}. On the quantum side the main new input is a decomposition and control of the high-frequency part that is sharp enough to recover both the free-energy limit and the convergence of reduced density matrices.

\begin{remark}
	We expect that the method developed here can be extended to the inhomogeneous setting with an additional one-body potential $V$, in the spirit of \cite{LewNamRou-21, CKRTG26}. In such a situation, the main high-frequency analysis should remain essentially the same, while the low-frequency comparison would have to be adapted to the presence of $V$. In particular, one should then be able to replace the Yukawa potential by the Coulomb potential. We do not pursue this extension here.
\end{remark}

\begin{remark}[Fractional Bessel kernels --- implemented range]
	\label{rem:frac-bessel-implemented}
	{\rmk{
			For $\beta\le \frac32$, genuinely new arguments seem necessary already for the shell contribution in Section~\ref{sec:shell-vanishing-rewritten}, and also for the tail contribution in Section~\ref{sec:tail-vanishing-checked}. More precisely, the restriction $\beta>\frac32$ already appears in the weighted shell summability estimate in Section~\ref{sec:shell-vanishing-rewritten}, and reappears in Section~\ref{sec:tail-vanishing-checked} in both the weighted mixed-channel estimate and the control of the high--high off-diagonal block.
			For $\beta\le 1$ the present Schur estimate for the high--high block is no longer sufficient, while for $\beta\le \frac12$ the commutator summations from Sections~\ref{sec:high0}--\ref{sec:shell-vanishing-rewritten} already fail.}}
\end{remark}


\bigskip
\noindent
{\bf Acknowledgments.}   P.T.N. was  supported by the European Research Council via the ERC Consolidator Grant RAMBAS (Project No. 10104424).  R.Z. and X.Z. are grateful to the financial supports by National Key R\&D Program of China (No. 2022YFA1006300) and the financial supports of the NSFC (No. 12426205).
R.Z. is grateful to the financial supports of the NSFC (No. 12271030). X.Z. is grateful to the financial supports in part by  the NSFC (No. 12595281, 12288201) and the support by key Lab of Random
Complex Structures and Data Science, Chinese Academy of
Science and the financial supports  by the Deutsche Forschungsgemeinschaft (DFG, German Research Foundation) – Project-ID 317210226--SFB 1283.

\section{Setting and Main result}\label{sec:setting}

In this section, we introduce the renormalized quantum model and state the main result concerning the rigorous connection between the quantum Gibbs state and the classical Gibbs measure. We  always consider the periodic fractional Bessel potential $\vbeta$ defined in \eqref{eq:Bessel-potential}. We use $\langle k\rangle:=(1+|k|^2)^{1/2}$ for $k\in\Z^2$.

\subsection{Many-body framework} Following the standard Fock-space formalism used in \cite[Section~2.1]{LewNamRou-21} and
\cite[Section~2.1]{NZZ25}, we first recall the bosonic grand-canonical framework and the
free quasi-free Gibbs state. Let
\[
\gH:=L^2(\T^2),\qquad \T^2=[0,2\pi]^2,
\]
and let
\[
\Fock:=\bigoplus_{n=0}^{\infty}\gH^{\otimes_s n},
\qquad
\gH^{\otimes_s 0}:=\mathbb C,
\qquad
\gH_n:=\gH^{\otimes_s n}=L^2_{\mathrm{sym}}\!\big((\T^2)^n\big).
\]
Here \(\otimes_s\) denotes the symmetric tensor product. For a self-adjoint one-body
operator \(A\) on \(\gH\), we use the usual second quantization
\[
\dG(A)
=
0\oplus\bigoplus_{n\ge1}\sum_{j=1}^{n}A_j,
\]
where \(A_j\) acts as \(A\) on the \(j\)-th variable and as the identity on the others.
In particular,
\[
\cN:=\dG(1)
\]
is the particle-number operator.

For every $f\in L^2(\T^2)$, we define the annihilation operator $a(f)$ and the creation operator $a^*(f)$ on Fock space \(\Fock\) by
\begin{align*}
	(a(f) \psi)(x_1,...,x_{n-1}) &= \sqrt{n} \int_{\T^2} \overline{f(x)}  \psi(x_1,...,x_{n-1},x) \d x,\\
	(a^*(f) \psi)(x_1,...,x_{n+1}) &= \frac1{\sqrt{n+1}} \sum_{j=1}^{n+1} f(x_j) \psi(x_1,...,x_{j-1}, x_{j+1}, ..., x_{n+1}),
\end{align*}
and extend to \(\Fock\) by linearity. These operators are related to the operators $a_x,a^*_x$  by
$$
a(f)= \int_{\T^2} \overline{f(x)} a_x \d x , \quad a^{*}(f) = \int_{\T^2} {f(x)} a_x^* \d x,\quad \forall f\in L^2(\T^2).
$$
They satisfy the canonical commutation relations
\[
[a(f),a(g)]=0,
\qquad
[a^\ast(f),a^\ast(g)]=0,
\qquad
[a(f),a^\ast(g)]=\langle f,g\rangle_{\gH}.
\]
With the normalized Fourier basis
\[
e_k(x)=(2\pi)^{-1}\e^{\ii k\cdot x},
\qquad
k\in\Z^2,
\]
we write
\[
a_k:=a(e_k),\qquad a_k^\ast:=a^\ast(e_k),
\]
so that
\[
[a_k,a_\ell]=0,\qquad [a_k^\ast,a_\ell^\ast]=0,\qquad [a_k,a_\ell^\ast]=\delta_{k,\ell}.
\]

A quantum state on \(\Fock\) is a positive trace-class operator \(\Gamma\) with
\(\Tr_{\Fock}\Gamma=1\). Writing \(\Gamma=\bigoplus_{n\ge0}\Gamma_n\) according to the
particle-number decomposition, its reduced \(k\)-body density matrix is
\[
\Gamma^{(k)}
=
\sum_{n\ge k}\binom{n}{k}\Tr_{k+1\to n}\Gamma_n
\qquad (k\ge1).
\]
Equivalently, for every bounded self-adjoint operator \(A_k\) on \(\gH^{\otimes_s k}\),
\[
\Tr_{\gH^{\otimes_s k}}\!\big(A_k\Gamma^{(k)}\big)
=
\Tr_{\Fock}\!\big(\mathbb A_k\Gamma\big),
\]
where \(\mathbb A_k\) denotes the second-quantized lift of \(A_k\). In particular,
\[
\Tr_{\Fock}(\cN\Gamma)=\Tr_{\gH}\!\big(\Gamma^{(1)}\big).
\]
We also use $\mathbb W(\cdot)$ to denote the bosonic second quantization of a two-body operator, characterized by
\[
\Tr_{\mathfrak F}(\mathbb W(T)\Gamma)=\Tr_{\gH^{\otimes_s 2}}(T\Gamma^{(2)})
\]
for every state $\Gamma$ whenever both sides make sense.
Define the one-body operator
\[
h=-\Delta+1,
\qquad
\dG(h)=\sum_{k\in\Z^2}(|k|^2+1)a_k^\ast a_k.
\]
We also write $h(p)=|p|^2+1$ for $p\in \Z^2$.

\paragraph{Free Gibbs state.}
For \(\lambda>0\), define the free quasi-free Gibbs state
\begin{equation}\label{eq:Gamma0}
	\Gamma_0 = Z_0^{-1}\exp(-\lambda \dG(h)),
	\qquad
	Z_0=\Tr\exp(-\lambda \dG(h)).
\end{equation}
We write
\[
\langle A\rangle_0:=\Tr_{\Fock}(A\Gamma_0).
\]
Its one-body density matrix is
\[
\gamma_0=\frac{1}{\e^{\lambda h}-1},
\]
and in momentum space
\[
\langle a_k^\ast a_\ell\rangle_0=n_k\,\delta_{k,\ell},
\qquad
n_k=\frac{1}{\e^{\lambda(|k|^2+1)}-1}.
\]
By translation invariance, the diagonal of \(\gamma_0\) is constant, and
\[
N_0:=\Tr_{\Fock}(\cN\Gamma_0)=\sum_{k\in\Z^2}n_k.
\]
Recall from \cite{LewNamRou-21} 
\[\lambda N_0\le C(1+|\log\lambda|),\qquad \lambda^q\Tr(\cN^q\Gamma_{0})
\le C_q(1+|\log\lambda|)^q,\qquad q\geq1.\] 



\subsection{Renormalized interaction, Hamiltonian, and Gibbs state}\label{sec:ren-int}\label{sec:ren-gibbs}

Now we introduce the renormalized fractional-Bessel interaction on
\(\T^2\). 
{\rmk{We first consider the genuine two-body fractional-Bessel interaction. On the \(n\)-particle sector
		\(\gH_n\), it acts as
		\[
		\lambda^2\sum_{1\le i<j\le n}\vbeta(x_i-x_j),
		\]
		and in Fock space it can be written in position space as
		\[
		\Wbeta
		:=
		\frac{\lambda^2}{2}\iint_{\T^2\times\T^2}
		\vbeta(x-y)\,a_x^\ast a_y^\ast a_y a_x\,\d x\,\d y .
		\]
		Expanding \(a_x=\sum_{p\in\Z^2}e_p(x)a_p\) together with the Fourier series of \(\vbeta\), one
		obtains the equivalent momentum representation}}
\begin{equation}\label{eq:Wmom}
	\Wbeta
	=
	\frac{\lambda^2}{2(2\pi)^2}\sum_{k\in\Z^2}\hvb(k)\sum_{p,q\in\Z^2}
	a_{p+k}^\ast a_{q-k}^\ast a_q a_p.
\end{equation}
For later comparison with density-square representations, define the density Fourier modes by
\begin{equation}\label{eq:rhok}
	\rho_k := \dG(e_k) = \int_{\T^2} e_k(x)\,a_x^\ast a_x\,\d x
	= \frac{1}{2\pi}\sum_{p\in\Z^2} a_{p+k}^\ast a_p.
\end{equation}
A direct CCR computation yields
\begin{equation}\label{eq:keyidentity}
	\sum_{p,q\in\Z^2} a_{p+k}^\ast a_{q-k}^\ast a_q a_p
	= (2\pi)^2\,\rho_k\rho_{-k} - \cN.
\end{equation}

Let us explain why the formal density-square representation is not used. Formally, plugging \eqref{eq:keyidentity} into \eqref{eq:Wmom} gives
\begin{equation}\label{eq:W-density-square}
	\Wbeta
	=
	\frac{\lambda^2}{2}\sum_{k\in\Z^2}\hvb(k)\,\rho_k\rho_{-k}
	-
	\frac{\lambda^2}{2(2\pi)^2}\Big(\sum_{k\in\Z^2}\hvb(k)\Big)\cN.
\end{equation}
For interactions with summable Fourier coefficients, the identity
\eqref{eq:W-density-square} is a convenient starting point for renormalization; see, for
instance, \cite[Section~2.4]{LewNamRou-21}. {\rmk{For the fractional-Bessel kernel in the present range $\frac32<\beta\le2$, however,
		\[
		\sum_{k\in\Z^2}\hvb(k)=\sum_{k\in\Z^2}\langle k\rangle^{-\beta}=+\infty,
		\]}}
so the second term in \eqref{eq:W-density-square} is a divergent self-energy proportional to
\(\cN\). Therefore the non-zero modes must be kept in the genuine quartic form
\eqref{eq:Wmom}, while only the zero mode is centered separately.

We define the renormalized interaction by
\begin{equation}\label{eq:Wre-def}
	W^{\mathrm{re}}
	:=
	\frac{\lambda^2}{2(2\pi)^2}
	\sum_{k\neq0}\hvb(k)\sum_{p,q\in\Z^2}
	a_{p+k}^\ast a_{q-k}^\ast a_q a_p
	\;+\;
	\frac{\lambda^2}{2(2\pi)^2}(\cN-N_0)^2.
\end{equation}
Since \(\hvb(0)=1\) and
\[
\sum_{p,q\in\Z^2}a_p^\ast a_q^\ast a_q a_p = \cN(\cN-1)=\cN^2-\cN,
\]
the \(k=0\) mode in \(\Wbeta\) equals
\[
\frac{\lambda^2}{2(2\pi)^2}(\cN^2-\cN).
\]
Therefore,
\begin{equation}\label{def1:Wre}
	W^{\mathrm{re}}
	=
	\Wbeta
	+
	\frac{\lambda^2}{2(2\pi)^2}\Big[(\cN-N_0)^2-(\cN^2-\cN)\Big]
	=
	\Wbeta + c_\lambda \cN + C_\lambda,
\end{equation}
with
\[
c_\lambda=\frac{\lambda^2}{2(2\pi)^2}(1-2N_0),
\qquad
C_\lambda=\frac{\lambda^2}{2(2\pi)^2}N_0^2.
\]
By Proposition~\ref{prop:def-Gibbs}, the renormalized Hamiltonian defines a well-posed
grand-canonical Gibbs state:
\begin{equation}\label{eq:Hre}
	\mathbb{H}_\lambda^{\mathrm{re}}:=\lambda \dG(h)+W^{\mathrm{re}},
	\qquad
	\Gamma_\lambda^{\mathrm{re}}:=(Z_\lambda^{\mathrm{re}})^{-1}\exp(-\mathbb{H}_\lambda^{\mathrm{re}}),
	\qquad
	Z_\lambda^{\mathrm{re}}:=\Tr\exp(-\mathbb{H}_\lambda^{\mathrm{re}}).
\end{equation}
From now on, for notational simplicity, we write
\[
\Gamma_\lambda:=\Gamma_\lambda^{\mathrm{re}}.
\]
We also denote by $\gamma_\lambda$ the one-body density matrix of $\Gamma_\lambda$.

\subsection{\rmk{Classical fractional-Bessel Gibbs measure and main result}}\label{sec:main-results}

{\rmk{Fix $\frac32<\beta\le2$. Let $\mu_0$ be the Gaussian free field on $\T^2$ with covariance $h^{-1}$. By the variational construction in \cite{OOT24}, the renormalized classical interaction $D$ associated with $\vbeta$ from \eqref{def:D} is well defined as the $L^1(\mu_0)$-limit of the finite-dimensional Hartree functionals, and the corresponding classical Gibbs measure is}}
\[
\d\nu(u)=z_{\mathrm{cl}}^{-1}e^{-D(u)}\,\d\mu_0(u),
\qquad
z_{\mathrm{cl}}=\int e^{-D(u)}\,\d\mu_0(u).
\]

\begin{theorem}[Main result]
	\label{thm:main-result}
	{\rmk{As $\lambda\downarrow0$, the following hold for the fixed fractional-Bessel interaction with $\frac32<\beta\le2$.}}
	\begin{enumerate}
		\item The relative free energy converges:
		\[
		-\log\frac{Z_\lambda^{\mathrm{re}}}{Z_0}
		\longrightarrow
		-\log z_{\mathrm{cl}}.
		\]
		\item For every fixed $k\ge1$,
		\[
		k!\lambda^k\Gamma_\lambda^{(k)}
		\longrightarrow
		\int |u^{\otimes k}\rangle\langle u^{\otimes k}|\,\d\nu(u)
		\qquad\text{in }\mathfrak S^2(\gH^{\otimes_s k}).
		\]
	\end{enumerate}
\end{theorem}

The free-energy convergence is proved in Theorem~\ref{thm:free-energy-convergence-final}, and the Hilbert--Schmidt convergence of the reduced density matrices is proved in Theorem~\ref{thm:HS-convergence-main}.

\paragraph{Strategy of proof.}
The proof is variational throughout.  In Section~\ref{sec:var} we establish source-dependent a priori bounds for perturbed Gibbs states. This is already a genuinely Bessel-specific step: the nonzero modes are estimated directly from the quartic formula, and the cancellation $\Tr(\gamma_0M_k)=0$ removes the potentially divergent direct term. We then localize the Gibbs state to low frequencies and use an exact decomposition of the interaction into a projected part plus high-frequency remainders. The projected interaction is compared in Subsection~\ref{sec:low-freq} with the finite-dimensional Hartree functional $D_P$ by means of a quantitative de Finetti theorem and a comparison between the free lower symbol and the Gaussian measure on $PH$.

This first reduction produces explicit remainder terms, and the rest of the paper is devoted to showing that they are negligible. In Section~\ref{sec:dec} we introduce the second cutoff that separates the nonzero-mode error into shell and tail pieces. Section~\ref{sec:high0} controls the zero-mode high-frequency remainder by combining a second-order correlation inequality with an exact block decomposition of the double commutator, so that every surviving quartic monomial still contains a low-frequency leg and can be estimated by asymmetric quartic Hilbert--Schmidt bounds. Section~\ref{sec:shell-vanishing-rewritten} proves that the shell contribution vanishes by polarizing the localized fluctuation observables and applying the same commutator strategy mode by mode. Section~\ref{sec:tail-vanishing-checked} handles the tail contribution differently: after discarding the positive $Q_1Q_1$ block, the remaining mixed part is kept as a genuine off-diagonal two-body operator, whose norm is paired with the small far-tail mass. These three mechanisms are what eventually recover the sharp lower bound and allow the quantum and classical variational principles to match.

\section{A priori estimate}

This section collects the basic a priori estimates for the renormalized Gibbs state. In
Subsection~\ref{sec:var} we recall the variational characterization of the free energy and prove a
uniform upper bound on $-\log(Z_\lambda^{\mathrm{re}}/Z_0)$. For later applications of the
second-order correlation inequality, we also formulate the corresponding bounds uniformly for a
family of perturbed Gibbs states obtained by adding a bounded one-body source $tB$. As
$\hvb(k)$ is not summable, the argument differs from
\cite{LewNamRou-21,NZZ25}: after the quasi-free entropy comparison in
Lemma~\ref{lem:relent-signedA}, the nonzero-mode contribution in
Lemma~\ref{lem:shell-perturbed-free} must be estimated directly from the genuine quartic formula,
and the direct term is controlled through the cancellation $\Tr(M_k\gamma_0)=0$ for $k\neq0$ with $M_k$ being multiplication by $e_k$,
which allows us to write $\Tr(\gamma_{0,t}M_k)=\Tr((\gamma_{0,t}-\gamma_0)M_k)$. In
Subsection~\ref{sec:numki} we derive the particle-number, entropy, and kinetic-energy estimates
that will be used later in Sections~\ref{sec:high0}, \ref{sec:shell-vanishing-rewritten}, and
\ref{sec:tail-vanishing-checked}.

\subsection{Variational formulation and a-priori free-energy upper bound}\label{sec:var}

The relevant free-energy functional is the relative entropy with respect to the free
quasi-free state \(\Gamma_0\) plus the renormalized interaction energy. By the Gibbs
variational principle,
\begin{equation}\label{eq:variational}
	-\log\frac{Z_\lambda^{\mathrm{re}}}{Z_0}
	=
	\min_{\Gamma\ge0,\ \Tr\Gamma=1}
	\Big\{\mathcal{H}(\Gamma,\Gamma_0)+\Tr(W^{\mathrm{re}}\Gamma)\Big\},
\end{equation}
where
\[
\mathcal{H}(\Gamma,\Gamma_0)=\Tr[\Gamma(\log\Gamma-\log\Gamma_0)]\ge0.
\]

\begin{lemma}
	\label{lem:relent-signedA}
	Let $h>0$ be self-adjoint on the one-body space and let $A=A^\ast$ satisfy
	\begin{equation}
		\label{eq:A-upper-bound}
		A \le c\,h \qquad\text{for some }c\in(0,1).
	\end{equation}
	For $\lambda>0$ define the quasi-free Gibbs states
	\[
	\Gamma_s := Z_s^{-1}\exp\!\big(-\lambda\, \dG(h-sA)\big),\qquad s\in[0,1],
	\]
	with one-body density matrices
	\[
	\gamma_s := (e^{\lambda(h-sA)}-1)^{-1}.
	\]
	Then $\Gamma_s$ is well-defined for all $s\in[0,1]$ and the relative entropy satisfies
	\begin{equation}
		\label{eq:relent-bound}
		\mathcal{H}(\Gamma_1,\Gamma_0)
		\;\le\;
		C(c)\,\Tr\!\big(h^{-1}Ah^{-1}A\big).
	\end{equation}
\end{lemma}
\begin{proof}
	Since
	\[
	\log\Gamma_s=-\lambda\, \dG(h-sA)-\log Z_s,
	\qquad
	\partial_s \log Z_s=\lambda\,\Tr(A\gamma_s),
	\]
	a direct computation yields the exact identity
	\begin{equation}
		\label{eq:relent-path}
		\mathcal{H}(\Gamma_1,\Gamma_0)
		=
		\lambda\int_0^1 \Tr\!\big(A(\gamma_1-\gamma_s)\big)\,\d s .
	\end{equation}
	Fix $s\in[0,1)$. Set
	\[
	h_s:=h-sA,
	\qquad
	A_s:=(1-s)A,
	\]
	so that
	\[
	h_s-A_s=h-A,
	\qquad
	\gamma_s=(e^{\lambda h_s}-1)^{-1},
	\qquad
	\gamma_1=(e^{\lambda(h_s-A_s)}-1)^{-1}.
	\]
	Hence
	\begin{equation}
		\label{eq:integrand-rewrite}
		\Tr\!\big(A(\gamma_1-\gamma_s)\big)
		=
		\frac{1}{1-s}\,
		\Tr\!\Big(A_s\Big(\frac{1}{e^{\lambda(h_s-A_s)}-1}-\frac{1}{e^{\lambda h_s}-1}\Big)\Big).
	\end{equation}
	From \eqref{eq:A-upper-bound} we have
	\[
	h_s=h-sA \ge (1-sc)h>0,
	\qquad
	A_s=(1-s)A \le (1-s)c\,h \le \frac{(1-s)c}{1-sc}\,h_s.
	\]
	Define
	\[
	c_s:=\frac{(1-s)c}{1-sc}\in(0,1),
	\qquad
	\text{so that } A_s\le c_s h_s.
	\]
	Applying \cite[Lemma 6.3]{LewNamRou-21}  to the pair $(\lambda h_s,\lambda A_s)$ gives
	\begin{equation}
		\label{eq:lemma63-applied}
		0\le
		\lambda\Tr\!\Big(A_s\Big(\frac{1}{e^{\lambda(h_s-A_s)}-1}-\frac{1}{e^{\lambda h_s}-1}\Big)\Big)
		\le
		\frac{1}{1-c_s}\,\Tr\!\big(h_s^{-1}A_s\,h_s^{-1}A_s\big).
	\end{equation}
	Combining \eqref{eq:integrand-rewrite} and \eqref{eq:lemma63-applied} yields, for $s\in[0,1)$,
	\begin{equation}
		\label{eq:integrand-bound}
		\lambda\Tr\!\big(A(\gamma_1-\gamma_s)\big)
		\le
		\frac{1}{1-s}\cdot \frac{1}{1-c_s}\,\Tr\!\big(h_s^{-1}A_s\,h_s^{-1}A_s\big).
	\end{equation}
	First note that
	\[
	1-c_s = 1-\frac{(1-s)c}{1-sc}=\frac{1-c}{1-sc}
	\qquad\Longrightarrow\qquad
	\frac{1}{1-c_s}=\frac{1-sc}{1-c}.
	\]
	Next, since $h_s\ge (1-sc)h$, we have $h_s^{-1}\le (1-sc)^{-1}h^{-1}$ and hence
	\begin{align*}
	&\Tr\!\big(h_s^{-1}A_s\,h_s^{-1}A_s\big)
	=
	\Tr\!\big(h_s^{-1/2}A_s\,h_s^{-1}A_sh_s^{-1/2}\big)
	\\
	&\le \frac1{1-sc}\Tr\!\big(h_s^{-1/2}A_s\,h^{-1}A_sh_s^{-1/2}\big)
	=
	\frac1{1-sc}\Tr\!\big(h^{-1/2}A_s\,h_s^{-1}A_sh^{-1/2}\big)
	\\
	&\le \frac1{(1-sc)^2}\Tr\!\big(h^{-1/2}A_s\,h^{-1}A_sh^{-1/2}\big)
	=
	\frac{(1-s)^2}{(1-sc)^2}\,\Tr\!\big(h^{-1}Ah^{-1}A\big).
	\end{align*}
	Insert these two relations into \eqref{eq:integrand-bound} to obtain
	\[
	\lambda\Tr\!\big(A(\gamma_1-\gamma_s)\big)
	\le
	\frac{1}{1-s}\cdot \frac{1-sc}{1-c}\cdot
	\frac{(1-s)^2}{(1-sc)^2}\,\Tr\!\big(h^{-1}Ah^{-1}A\big)
	=
	\frac{1-s}{1-c}\cdot\frac{1}{1-sc}\,\Tr\!\big(h^{-1}Ah^{-1}A\big).
	\]
	Plugging this into \eqref{eq:relent-path} and integrating gives
	\[
	\mathcal{H}(\Gamma_1,\Gamma_0)
	\le
	\frac{1}{1-c}\Big(\int_0^1 \frac{1-s}{1-sc}\,\d s\Big)\Tr\!\big(h^{-1}Ah^{-1}A\big).
	\]
	which proves \eqref{eq:relent-bound}.
\end{proof}

\begin{lemma}
	\label{lem:shell-perturbed-free}
	Let $f=f^*$ be a bounded one-body operator with
	\[
	\|f\|\le 1,
	\qquad
	a_f:=\bigl\|h^{-1/2}fh^{-1/2}\bigr\|_{\mathrm{HS}}\le c_0,
	\]
	for some sufficiently small universal constant $c_0>0$.  Define
	\[
	h_t:=h-\frac t2 f,
	\qquad
	\Gamma_{0,t}
	:=
	\frac{e^{-\lambda \dG(h_t)}}{\Tr(e^{-\lambda \dG(h_t)})},
	\qquad |t|\le 1,
	\]
	and denote by $\gamma_{0,t}$ the one-body density matrix of $\Gamma_{0,t}$.
	Then uniformly for $|t|\le 1$,
	\begin{align}
		\lambda^2
		\Tr\!\Big(
		(\cN-N_0)^2\Gamma_{0,t}
		\Big)
		&\le C,
		\label{eq:shell-free-N2-local}
		\\
		\Tr\!\big(W^{\rm re}\Gamma_{0,t}\big)
		&\le C.
		\label{eq:shell-free-W-local}
	\end{align}
\end{lemma}

\begin{proof}
	We divide the proof into two steps.

	\smallskip
	\noindent\textbf{Step 1. Centered number fluctuations for the perturbed free state.}
	Using \cite[Theorem 5.11]{LewNamRou-21}, we obtain
	\[
	\lambda^2
	\Tr_{\mathfrak F}\!\Big((\cN-\Tr(\gamma_{0,t}))^2\Gamma_{0,t}\Big)
	\lesssim \Tr(h_t^{-2})\lesssim 1.
	\]
By \cite[Theorem 6.1]{LewNamRou-21}, together with
	Lemma~\ref{lem:relent-signedA} we obtain
\begin{align*}
\lambda\bigl|\Tr(\gamma_{0,t})-N_0\bigr|
&\le
\Big(\Tr \bigl|\sqrt{\lambda h}(\gamma_{0,t}-\gamma_0)\sqrt{\lambda h}\bigr|^2\Big)^{1/2}
(\Tr h^{-2})^{1/2}
\le C a_f.
\end{align*}
Thus we find
	\[
	\lambda^2
	\Tr\!\Big((\cN-N_0)^2\Gamma_{0,t}\Big)
	\le
	2\lambda^2
	\Tr\!\Big((\cN-\Tr(\gamma_{0,t}))^2\Gamma_{0,t}\Big)
	+
	2\lambda^2\bigl(\Tr(\gamma_{0,t})-N_0\bigr)^2
	\le C.
	\]
	This proves \eqref{eq:shell-free-N2-local}.
	
	\smallskip
	\noindent\textbf{Step 2. Expectation of the renormalized interaction.}
	We split
	\[
	W^{\rm re}
	=
	\frac{\lambda^2}{2(2\pi)^2}(\cN-N_0)^2
	+
	W_{\neq0}^{\rm re}.
	\]
	The zero-mode part is already controlled by \eqref{eq:shell-free-N2-local}. For the
	nonzero-mode part, let
	\[
	M_k:=M_{e_k}
	\qquad (k\in\Z^2)
	\]
	denote multiplication by the normalized Fourier mode $e_k$. Then Wick's theorem yields
	\begin{equation}\label{eq:wick}
	\Tr_{\mathfrak F}\!\bigl(W_{\neq0}^{\rm re}\Gamma_{0,t}\bigr)
	=
	\frac{\lambda^2}{2}
	\sum_{k\neq 0}\hvb(k)
	\Big(
	|\Tr(\gamma_{0,t}M_k)|^2
	+
	\Tr(\gamma_{0,t}M_k\gamma_{0,t}M_{-k})
	\Big).
	\end{equation}
	Indeed, recalling that
	\[
	W_{\neq0}^{\rm re}
	=
	\frac{\lambda^2}{2(2\pi)^2}
	\sum_{k\neq 0}\hvb(k)\sum_{p,q\in\mathbb Z^2}
	a_{p+k}^*a_{q-k}^*a_q a_p,
	\]
	it is enough to compute the expectation of each quartic monomial in the quasi-free state
	$\Gamma_{0,t}$.
	
	For a gauge-invariant quasi-free state with one-body density matrix $\gamma_{0,t}$, one has
	\[
	\Tr_{\mathfrak F}\!\bigl(a^*(f)a(g)\Gamma_{0,t}\bigr)
	=
	\langle g,\gamma_{0,t}f\rangle,
	\]
	and Wick's theorem gives
	\begin{align*}
		\Tr_{\mathfrak F}\!\bigl(a_{p+k}^*a_{q-k}^*a_q a_p\,\Gamma_{0,t}\bigr)
		&=
		\Tr_{\mathfrak F}\!\bigl(a_{p+k}^*a_p\,\Gamma_{0,t}\bigr)
		\Tr_{\mathfrak F}\!\bigl(a_{q-k}^*a_q\,\Gamma_{0,t}\bigr) \\
		&\quad+
		\Tr_{\mathfrak F}\!\bigl(a_{p+k}^*a_q\,\Gamma_{0,t}\bigr)
		\Tr_{\mathfrak F}\!\bigl(a_{q-k}^*a_p\,\Gamma_{0,t}\bigr).
	\end{align*}
	Thus \eqref{eq:wick} follows by a direct computation.
	
%
%
	We first treat the exchange term. Since $h_t\ge \frac12 h$ for $|t|\le1$, we have
	\[
	0\le \lambda\gamma_{0,t}=\frac{\lambda}{e^{\lambda h_t}-1}\le h_t^{-1}\le 2h^{-1}.
	\]
	Hence
	\begin{align*}
		&\lambda^2\Tr(\gamma_{0,t}M_k\gamma_{0,t}M_{-k})
		=
		\lambda^2\Tr\bigl(\gamma_{0,t}^{1/2}M_k\gamma_{0,t}M_{-k}\gamma_{0,t}^{1/2}\bigr)
		\\
		&\le
		C\lambda\Tr\bigl(\gamma_{0,t}^{1/2}M_k h^{-1}M_{-k}\gamma_{0,t}^{1/2}\bigr)
		=
		C\lambda\Tr\bigl(h^{-1/2}M_{-k}\gamma_{0,t}M_kh^{-1/2}\bigr)
		\\
		&\le
		C\Tr\bigl(h^{-1}M_{-k}h^{-1}M_k\bigr)
		=
		C\sum_{p\in\Z^2}\frac1{(|p+k|^2+1)(|p|^2+1)}
		\le
		C\frac{\log(2+|k|)}{1+|k|^2},
	\end{align*}
where the last step follows from Lemma~\ref{lem:sum}.
	
	For the direct term, we use $\Tr(\gamma_0M_k)=0$ for $k\neq0$ and write
	\[
	\Tr(\gamma_{0,t}M_k)=\Tr\bigl((\gamma_{0,t}-\gamma_0)M_k\bigr).
	\]
	By \cite[Theorem 6.1]{LewNamRou-21}, together with
	Lemma~\ref{lem:relent-signedA},
	\[
	\lambda\,
	\bigl\|\sqrt h\,(\gamma_{0,t}-\gamma_0)\,\sqrt h\bigr\|_{\mathrm{HS}}
	\le C a_f.
	\]
	Hence
	\[
	\lambda |\Tr(\gamma_{0,t}M_k)|
	\le
	\lambda\,
	\bigl\|\sqrt h\,(\gamma_{0,t}-\gamma_0)\,\sqrt h\bigr\|_{\mathrm{HS}}
	\bigl\|h^{-1/2}M_kh^{-1/2}\bigr\|_{\mathrm{HS}}
	\le
	C a_f\Bigl(\frac{\log(2+|k|)}{1+|k|^2}\Bigr)^{1/2},
	\]
	where we used
	\[
	\|h^{-1/2}M_k h^{-1/2}\|_{\mathrm{HS}}^2
	=
	\sum_{q\in\mathbb Z^2}
	\frac{1}{(|q|^2+1)(|q+k|^2+1)}
	\le
	C\frac{\log(2+|k|)}{1+|k|^2}.
	\]
	{\color{red}Since $\hvb(k)=\langle k\rangle^{-\beta}$ with $\beta>0$, both the direct and exchange
	contributions are summed against
	\[
	\hvb(k)\frac{\log(2+|k|)}{1+|k|^2},
	\]
	which is summable on $\Z^2$. Hence the two series converge uniformly in $|t|\le1$, and we conclude that}
	\[
	\Tr_{\mathfrak F}\!\bigl(W_{\neq0}^{\rm re}\Gamma_{0,t}\bigr)\le C.
	\]
	This proves \eqref{eq:shell-free-W-local}.
\end{proof}

Let $f$ be as in Lemma \ref{lem:shell-perturbed-free}. Set
	\[
	B:=\frac{\lambda}{2}\dG(f),
	\qquad
	\Gamma_{\lambda,t}
	:=
	\frac{e^{-\mathbb{H}_\lambda^{\rm re}+tB}}{\Tr(e^{-\mathbb{H}_\lambda^{\rm re}+tB})},
	\qquad
	\Gamma_{0,t}
	:=
	\frac{e^{-\lambda \dG(h)+tB}}{\Tr(e^{-\lambda \dG(h)+tB})}
	=
	\frac{e^{-\lambda \dG(h_t)}}{\Tr(e^{-\lambda \dG(h_t)})},
	\qquad |t|\le 1.
	\]

\begin{lemma}
	\label{lem:exch2d}\label{lem:yukawa-lemma11-safe}
Then uniformly for \(|t|\le1\),
\begin{equation}\label{eq:yukawa-free-energy-upper}
	-\log\frac{\Tr(e^{-\mathbb{H}_\lambda^{\mathrm{re}}+tB})}{\Tr(e^{-\lambda \dG(h)+tB})}\le C.
\end{equation}
Consequently,
\begin{equation}\label{eq:yukawa-free-energy-functional}
	\mathcal{H}(\Gamma_{\lambda,t},\Gamma_{0,t})
	+\Tr\!\big(W^{\mathrm{re}}\Gamma_{\lambda,t}\big)
	\le C.
\end{equation}
\end{lemma}

\begin{proof}
	By the Gibbs variational principle relative to $\Gamma_{0,t}$,
	\[
	-\log\frac{\Tr(e^{-\mathbb{H}_\lambda^{\mathrm{re}}+tB})}{\Tr(e^{-\lambda \dG(h)+tB})}
	=
	\mathcal{H}(\Gamma_{\lambda,t},\Gamma_{0,t})
	+
	\Tr\!\big(W^{\mathrm{re}}\Gamma_{\lambda,t}\big).
	\]
	Taking the trial state $\Gamma_{0,t}$ yields
	\[
	\mathcal{H}(\Gamma_{\lambda,t},\Gamma_{0,t})
	+
	\Tr\!\big(W^{\mathrm{re}}\Gamma_{\lambda,t}\big)
	\le
	\Tr\!\big(W^{\mathrm{re}}\Gamma_{0,t}\big).
	\]
	By Lemma~\ref{lem:shell-perturbed-free},
	\[
	\Tr\!\big(W^{\mathrm{re}}\Gamma_{0,t}\big)\le C
	\qquad\text{uniformly for }|t|\le1.
	\]
	This proves \eqref{eq:yukawa-free-energy-functional}, and
	\eqref{eq:yukawa-free-energy-upper} follows immediately.
\end{proof}

\begin{remark}

	\label{rem:missing-lower-bound-lemma11}
	At this stage,  the lower bound
	\[
	0\le -\log\frac{Z_\lambda^{\mathrm{re}}}{Z_0}
	\]
	does \emph{not} follow directly from positivity of the interaction, because with the present
	renormalization one keeps the \(k\neq0\) part in genuine quartic form in order to avoid the divergent
	self-energy term \(\big(\sum_k \hvb(k)\big)\cN\). Equivalently, in the sector-wise realization one has
	\[
	\mathbb{H}_\lambda^{\mathrm{re}}
	=
	\lambda \dG(h)+\Wbeta+c_\lambda \cN+C_\lambda,
	\]
	and for small \(\lambda\) the coefficient \(c_\lambda\) is negative. Thus \(W^{\mathrm{re}}\) is not manifestly
	nonnegative. Later, the required lower bound is recovered through the comparison with the finite-dimensional classical problem and the vanishing of the high-frequency remainders; see  Theorem~\ref{thm:lower-bound-classical-final}.
\end{remark}

\subsection{A priori bounds on the number of particles and kinetic operator}\label{sec:numki}

The goal of this subsection is to establish the a priori estimates that will be used in the
high-frequency analysis: logarithmic bounds on the moments of $\cN$, a weighted
Hilbert--Schmidt estimate for $\gamma_\lambda-\gamma_0$, and a uniform kinetic-energy
bound.

\begin{proposition}[Logarithmic moments of the number operator for the perturbed Gibbs states]
	\label{prop:number-moments-log}\label{prop:number-moments-log-safe}
	Let $B$,
	$\Gamma_{\lambda,t}$ and $\Gamma_{0,t}$ be as in Lemma \ref{lem:yukawa-lemma11-safe}. Then for every fixed $q>0$
	there exists $C_q>0$ such that, for all sufficiently small $\lambda\in(0,1/2)$,
	\[
	\sup_{|t|\le1}\lambda^q\Tr(\cN^q\Gamma_{\lambda,t})
	\le C_q(1+|\log\lambda|)^q.
	\]
	In particular,
	\[
	\sup_{|t|\le1}\lambda\Tr(\cN\Gamma_{\lambda,t})
	\le C(1+|\log\lambda|).
	\]
\end{proposition}

\begin{proof}
	Set
	\(
	L_\lambda:=1+|\log\lambda|.
	\)
	Choose once and for all
	\(
	s_0\in(0,\frac14].
	\)
	For every $|t|\le1$ and $|s|\le s_0$, define
	\[
	Z_{\lambda,t}(s):=\Tr_{\mathfrak F}\!\big(e^{-\mathbb{H}_\lambda^{\mathrm{re}}+tB+s\lambda \cN}\big),
	\qquad
	\Gamma_{\lambda,t,s}:=\frac{e^{-\mathbb{H}_\lambda^{\mathrm{re}}+tB+s\lambda \cN}}{Z_{\lambda,t}(s)},
	\]
	and
	\[
	\Gamma_{0,t,s}
	:=
	\frac{e^{-\lambda\dG(h)+tB+s\lambda\cN}}{\Tr(e^{-\lambda\dG(h)+tB+s\lambda\cN})}
	=
	\frac{e^{-\lambda\dG(h_{t,s})}}{\Tr(e^{-\lambda\dG(h_{t,s})})},
	\]
	where
	\(
	h_{t,s}:=h-s-\frac t2 f.
	\)
	Since $\|f\|\le1$ and $h\ge1$, we have
	\[
	h_{t,s}\ge h-\Bigl(s_0+\frac12\Bigr)\ge \frac14 h>0,
	\]
	so the free state $\Gamma_{0,t,s}$ is well defined.

	Set
	\(
	A_{t,s}:=s+\frac t2 f.
	\)
	Then
	\[
	h_{t,s}=h-A_{t,s},
	\qquad
	A_{t,s}\le \Bigl(s_0+\frac12\Bigr)h\le \frac34 h.
	\]
	Moreover,
	\[
	\Tr\!\big(h^{-1}A_{t,s}h^{-1}A_{t,s}\big)\le C,
	\]
	uniformly for $|t|\le1$ and $|s|\le s_0$, because $h^{-1}\in\mathfrak S^2$ and
	$\|h^{-1/2}fh^{-1/2}\|_{\mathrm{HS}}\le c_0$. The proof of Lemma~\ref{lem:shell-perturbed-free}
	uses only the lower bound $h_{t,s}\gtrsim h$ and the control of
	$\Tr(h^{-1}A_{t,s}h^{-1}A_{t,s})$, so it applies verbatim with $h_t$ replaced by $h_{t,s}$.
	Therefore
	\begin{equation}\label{eq:free-W-ts}
	\Tr\!\big(W^{\mathrm{re}}\Gamma_{0,t,s}\big)\le C
	\qquad\text{uniformly for }|t|\le1,\ |s|\le s_0.
	\end{equation}

	\smallskip
	\noindent\textbf{Step 1. Uniform first and second particle-number bounds.}
	By the Gibbs variational principle relative to $\Gamma_{0,t,s}$,
	\[
	\mathcal H(\Gamma_{\lambda,t,s},\Gamma_{0,t,s})
	+
	\Tr\!\big(W^{\mathrm{re}}\Gamma_{\lambda,t,s}\big)
	\le
	\Tr\!\big(W^{\mathrm{re}}\Gamma_{0,t,s}\big).
	\]
	Hence \eqref{eq:free-W-ts} implies
	\[
	\Tr\!\big(W^{\mathrm{re}}\Gamma_{\lambda,t,s}\big)\le C
	\qquad\text{uniformly for }|t|\le1,\ |s|\le s_0.
	\]
	Using
	\[
	W^{\mathrm{re}}=\Wbeta+c_\lambda\cN+C_\lambda,
	\]
	together with
	\[
	\Wbeta
	\ge
	\frac{\vbstar\lambda^2}{2}\,\cN(\cN-1)
	\ge
	\frac{\vbstar\lambda^2}{4}\,\cN^2-C\lambda^2,
	\]
from Lemma \ref{lem:periodic-yukawa-positive},	we infer
	\[
	\frac{\vbstar\lambda^2}{4}\Tr(\cN^2\Gamma_{\lambda,t,s})
	+
	c_\lambda\Tr(\cN\Gamma_{\lambda,t,s})
	+
	C_\lambda
	\le C.
	\]
	Using
	\[
	c_\lambda\Tr(\cN\Gamma_{\lambda,t,s})
	=
	\alpha_0\lambda^2\Tr(\cN\Gamma_{\lambda,t,s})
	-
	2\alpha_0(\lambda N_0)\bigl(\lambda\Tr(\cN\Gamma_{\lambda,t,s})\bigr),
	\]
	with $\alpha_0=\frac12(2\pi)^{-2}$, we discard the positive term
	$\alpha_0\lambda^2\Tr(\cN\Gamma_{\lambda,t,s})$ and infer
	\[
	\frac{\vbstar\lambda^2}{4}\Tr(\cN^2\Gamma_{\lambda,t,s})
	-
	2\alpha_0(\lambda N_0)\bigl(\lambda\Tr(\cN\Gamma_{\lambda,t,s})\bigr)
	+
	\alpha_0(\lambda N_0)^2
	\le C.
	\]
	By Cauchy--Schwarz,
	\[
	\lambda\Tr(\cN\Gamma_{\lambda,t,s})
	\le
	\bigl(\lambda^2\Tr(\cN^2\Gamma_{\lambda,t,s})\bigr)^{1/2},
	\]
	and another application of Young's inequality yields
	\[
	2\alpha_0(\lambda N_0)\bigl(\lambda\Tr(\cN\Gamma_{\lambda,t,s})\bigr)
	\le
	\frac{\vbstar}{8}\lambda^2\Tr(\cN^2\Gamma_{\lambda,t,s})
	+
	C(\lambda N_0)^2.
	\]
Thus
	\[
	\sup_{\substack{|t|\le1\\ |s|\le s_0}}
	\lambda^2\Tr(\cN^2\Gamma_{\lambda,t,s})
	\le
	C L_\lambda^2,
	\]
	and therefore
	\begin{equation}\label{eq:perturbed-first-moment-ts}
	\sup_{\substack{|t|\le1\\ |s|\le s_0}}
	\lambda\Tr(\cN\Gamma_{\lambda,t,s})
	\le
	C L_\lambda.
	\end{equation}

	\smallskip
	\noindent\textbf{Step 2. Exponential moments and all fixed moments.}
	Since both $\mathbb{H}_\lambda^{\mathrm{re}}$ and $B$ commute with $\cN$, the map
	$s\mapsto Z_{\lambda,t}(s)$ is differentiable and
	\[
	\frac{\d}{\d s}\log Z_{\lambda,t}(s)
	=
	\lambda\Tr(\cN\Gamma_{\lambda,t,s}).
	\]
	By \eqref{eq:perturbed-first-moment-ts}, for every $|t|\le1$ and $|s|\le s_0$,
	\[
	\left|
	\log\frac{Z_{\lambda,t}(s)}{Z_{\lambda,t}(0)}
	\right|=
	\left|\int_0^s\lambda\Tr(\cN\Gamma_{\lambda,t,r})\,\d r\right|
	\le
	C L_\lambda |s|.
	\]
	Define
	\(
	X_\lambda:=\frac{\lambda\cN}{L_\lambda}.
	\)
	Then for every $\tau\in[0,s_0]$,
	\[
	\log\Tr\!\big(e^{\tau X_\lambda}\Gamma_{\lambda,t}\big)
	=
	\log\frac{Z_{\lambda,t}(\tau/L_\lambda)}{Z_{\lambda,t}(0)}
	\le C\tau,
	\]
	hence
	\[
	\sup_{|t|\le1}\Tr\!\big(e^{\tau X_\lambda}\Gamma_{\lambda,t}\big)\le e^{C\tau}
	\qquad\text{for all }\tau\in[0,s_0].
	\]
	Fix $q>0$ and $\tau_*:=s_0/2$. Using the elementary inequality
	\(
	x^q\le C_{q,\tau_*}e^{\tau_*x}
	\) for $x\geq0$,
	we conclude that
	\[
	\sup_{|t|\le1}\Tr(X_\lambda^q\Gamma_{\lambda,t})
	\le
	C_{q,\tau_*}\sup_{|t|\le1}\Tr\!\big(e^{\tau_*X_\lambda}\Gamma_{\lambda,t}\big)
	\le
	C_q.
	\]
	Equivalently,
	\[
	\sup_{|t|\le1}\lambda^q\Tr(\cN^q\Gamma_{\lambda,t})
	\le C_qL_\lambda^q.
	\]
	This proves the proposition.
\end{proof}

\begin{proposition}[Weighted Hilbert--Schmidt estimate with logarithmic loss for the perturbed Gibbs states]
	\label{prop:one-body-HS-safe}
		Let $B$,
	$\Gamma_{\lambda,t}$ and $\Gamma_{0,t}$ be as in Lemma \ref{lem:yukawa-lemma11-safe}. Then there exists $C>0$ such that,
	for all sufficiently small $\lambda\in(0,1/2)$,
	\begin{equation}\label{eq:relative-entropy-log2}
		\sup_{|t|\le1}\mathcal{H}(\Gamma_{\lambda,t},\Gamma_{0,t})
		\le C(1+|\log\lambda|^2).
	\end{equation}
	Consequently,
	\begin{equation}\label{eq:one-body-HS-safe}
		\sup_{|t|\le1}
		\lambda\,
		\bigl\|
		\sqrt h\,(\Gamma_{\lambda,t}^{(1)}-\Gamma_{0,t}^{(1)})\,\sqrt h
		\bigr\|_{\mathrm{HS}}
		\le C(1+|\log\lambda|^2).
	\end{equation}
	In particular,
	\begin{equation}\label{eq:number-difference-global-safe}
		\sup_{|t|\le1}
		\lambda\,
		\left|
		\Tr\big[(\Gamma_{\lambda,t}^{(1)}-\Gamma_{0,t}^{(1)})\big]
		\right|
		\le
		C(1+|\log\lambda|^2).
	\end{equation}
\end{proposition}

\begin{proof}
	Set again
	\(
	L_\lambda:=1+|\log\lambda|.
	\)
	By Lemma~\ref{lem:yukawa-lemma11-safe},
	\[
	\mathcal{H}(\Gamma_{\lambda,t},\Gamma_{0,t})
	+
	\Tr(W^{\mathrm{re}}\Gamma_{\lambda,t})
	\le C
	\qquad\text{for all }|t|\le1.
	\]
	Using
	\[
	W^{\mathrm{re}}=\Wbeta+c_\lambda\cN+C_\lambda,
	\]
	with
	\[
	c_\lambda
	=
	\frac{\lambda^2}{2(2\pi)^2}(1-2N_0),
	\qquad
	C_\lambda
	=
	\frac{\lambda^2}{2(2\pi)^2}N_0^2,
	\]
	{\rmk{and the pointwise nonnegativity of the fractional-Bessel kernel, we have}}
	\[
	\Wbeta\ge0,
	\qquad
	C_\lambda\ge0.
	\]
	Hence
	\[
	\Tr(W^{\mathrm{re}}\Gamma_{\lambda,t})
	\ge
	c_\lambda\Tr(\cN\Gamma_{\lambda,t}),
	\]
	and therefore
	\[
	\mathcal{H}(\Gamma_{\lambda,t},\Gamma_{0,t})
	\le
	C+|c_\lambda|\,\Tr(\cN\Gamma_{\lambda,t}).
	\]
	Now
	\[
	|c_\lambda|
	\le C\lambda^2(1+N_0)
	\le C\lambda L_\lambda,
	\]
	while Proposition~\ref{prop:number-moments-log} gives
	\[
	\sup_{|t|\le1}\lambda\Tr(\cN\Gamma_{\lambda,t})\le CL_\lambda.
	\]
	Combining the last two bounds yields \eqref{eq:relative-entropy-log2}.

		Here
		\(
		h_t=h-\frac t2 f.
		\)
	Applying \cite[Theorem~6.1]{LewNamRou-21} with $h_t$ in place of $h$ gives
	\[
	\lambda^2
	\bigl\|
	\sqrt{h_t}\,(\Gamma_{\lambda,t}^{(1)}-\Gamma_{0,t}^{(1)})\,\sqrt{h_t}
	\bigr\|_{\mathrm{HS}}^2
	\le
	4\,\mathcal{H}(\Gamma_{\lambda,t},\Gamma_{0,t})
	\bigl(\sqrt2+\sqrt{\mathcal{H}(\Gamma_{\lambda,t},\Gamma_{0,t})}\bigr)^2.
	\]
	Since $h\le 2h_t$, we obtain \eqref{eq:one-body-HS-safe}.
By the Hilbert--Schmidt Cauchy--Schwarz inequality,
	\[
	\lambda\left|
	\Tr\big[(\Gamma_{\lambda,t}^{(1)}-\Gamma_{0,t}^{(1)})\big]
	\right|
	\le
	\lambda\|h^{-1}\|_{\mathrm{HS}}\,
	\bigl\|
	\sqrt h\,(\Gamma_{\lambda,t}^{(1)}-\Gamma_{0,t}^{(1)})\,\sqrt h
	\bigr\|_{\mathrm{HS}}
	\le
	CL_\lambda^2,
	\]
	uniformly in $|t|\le1$. This proves \eqref{eq:number-difference-global-safe}.
\end{proof}

We now return to the unperturbed Gibbs state $\Gamma_\lambda$ itself. The next lemma is stated
without the auxiliary perturbation $tB$ and provides the kinetic estimate that will be used later
in the high-frequency analysis. Unlike \cite[Lemma~5.15]{LewNamRou-21}, the proof does not come from a comparison with the free Hamiltonian. The obstruction is the nonsummable  $\sum_{k\in\mathbb Z^2}\hvb(k)=\infty$, which prevents one from rewriting the interaction so as to dominate $\mathbb H^{\rm re}_\lambda$ by a harmless multiple of $\lambda\dG(h)$. We therefore argue directly from the entropy inequality with the choice $X_t=t\lambda\dG(h)$.

\begin{lemma}[First kinetic moment under logarithmic relative entropy]

	\label{lem:first-kinetic-log2}
	For every fixed $t\in(0,1)$ there exists $C_t>0$ such that
	\[
	\lambda^2\Tr\!\big(\dG(h)\Gamma_\lambda\big)
	\le
	\frac{C_t}{t}
	+
	\frac{C}{t}\,\lambda(1+|\log\lambda|^2).
	\]
	In particular,
	\[
	\lambda^2\Tr\!\big(\dG(h)\Gamma_\lambda\big)\le C.
	\]
\end{lemma}

\begin{proof}
	For every self-adjoint operator $X$ such that $\Tr(e^X\Gamma_0)<\infty$, the entropy inequality gives
	\[
	\Tr(X\Gamma_\lambda)
	\le
	\mathcal{H}(\Gamma_\lambda,\Gamma_0)+\log\Tr(e^X\Gamma_0).
	\]
	Choose
	\[
	X_t:=t\lambda\dG(h),
	\qquad 0<t<1.
	\]
	Then
	\[
	t\lambda\Tr\!\big(\dG(h)\Gamma_\lambda\big)
	\le
	\mathcal{H}(\Gamma_\lambda,\Gamma_0)
	+
	\log\frac{\Tr(e^{-(1-t)\lambda\dG(h)})}{\Tr(e^{-\lambda\dG(h)})}.
	\]
	By Proposition~\ref{prop:one-body-HS-safe}, the entropy term is bounded by
	\[
	\mathcal{H}(\Gamma_\lambda,\Gamma_0)\le C(1+|\log\lambda|^2).
	\]
	It remains to estimate the free partition-function ratio.
	
	Set
	\[
	\Phi(s):=\log\Tr(e^{-s\dG(h)})
	=
	-\sum_{p\in\mathbb Z^2}\log(1-e^{-sh(p)}),
	\qquad s>0.
	\]
Here we used \cite[Appendix A]{LewNamRou-15}.
	Then
	\[
	\Phi'(s)
	=
	-\sum_{p\in\mathbb Z^2}\frac{h(p)}{e^{sh(p)}-1},
	\]
	so that
	\[
	\log\frac{\Tr(e^{-(1-t)\lambda\dG(h)})}{\Tr(e^{-\lambda\dG(h)})}
	=
	-\int_{(1-t)\lambda}^{\lambda}\Phi'(s)\,\d s
	=
	\int_{1-t}^{1}
	\sum_{p\in\mathbb Z^2}\frac{\lambda h(p)}{e^{r\lambda h(p)}-1}\,\d r,
	\]
	after the change of variables $s=r\lambda$.
	
	Fix $t\in(0,1)$ and $r\in[1-t,1]$. Split the sum into the regions $h(p)\le\lambda^{-1}$ and $h(p)>\lambda^{-1}$.
	If $h(p)\le\lambda^{-1}$, then
	\[
	\frac{\lambda h(p)}{e^{r\lambda h(p)}-1}\le C_t,
	\]
	hence
	\[
	\sum_{h(p)\le\lambda^{-1}}\frac{\lambda h(p)}{e^{r\lambda h(p)}-1}
	\le
	C_t\,\#\{p\in\mathbb Z^2:\ h(p)\le\lambda^{-1}\}
	\le C_t\lambda^{-1}.
	\]
	If $h(p)>\lambda^{-1}$, then
	\[
	\frac{\lambda h(p)}{e^{r\lambda h(p)}-1}
	\le
	C_t\lambda h(p)e^{-c_t\lambda h(p)},
	\]
	and comparison with the corresponding Gaussian integral on $\mathbb R^2$ gives
	\[
	\sum_{h(p)>\lambda^{-1}}\frac{\lambda h(p)}{e^{r\lambda h(p)}-1}
	\le
	C_t\sum_{p\in\mathbb Z^2}\lambda h(p)e^{-c_t\lambda h(p)}
	\le C_t\lambda^{-1}.
	\]
	Therefore
	\[
	\log\frac{\Tr(e^{-(1-t)\lambda\dG(h)})}{\Tr(e^{-\lambda\dG(h)})}
	\le C_t\lambda^{-1}.
	\]
	Combining the previous estimates, we arrive at
	\[
	t\lambda\Tr\!\big(\dG(h)\Gamma_\lambda\big)
	\le
	C(1+|\log\lambda|^2)+C_t\lambda^{-1}.
	\]
	Multiplying by $\lambda$ yields
	\[
	t\lambda^2\Tr\!\big(\dG(h)\Gamma_\lambda\big)
	\le
	C\lambda(1+|\log\lambda|^2)+C_t,
	\]
	which proves the first estimate. Choosing, for instance, $t=1/2$ gives the uniform bound in the second display.
\end{proof}

Since $[\cN,\dG(h)]=0$, we also have the following H\"older's inequality, which is useful in the sequel.

\begin{lemma}\label{lem:holder}
	Let
	\(
	\theta\in(0,1).
	\)
	Then
	\begin{align}
		\Tr\!\bigl[\dG(h)^\theta \cN^{2-\theta}\Gamma_\lambda\bigr]
		&\le
		\Tr\!\bigl[\dG(h)\Gamma_\lambda\bigr]^\theta\,
		\Tr\!\Bigl[\cN^{\frac{2-\theta}{1-\theta}}\Gamma_\lambda\Bigr]^{1-\theta}.
		\label{eq:main-estimate}
	\end{align}
Moreover, write $L_\lambda:=1+|\log\lambda|$ and $\mathbb{H}_s:=\dG(h^s)$ for $s\in(0,1)$. Then for $s_1,s_2\in(0,1)$ with $\theta=s_1+s_2\in(0,1)$
\begin{equation}
			\label{eq:shell-energy-number-moment-local}
			\lambda^4
			\Tr\!\Big(
			\mathbb{H}_{s_1}\mathbb{H}_{s_2}\Gamma_\lambda
			\Big)
			\le
			C\lambda^{2-\theta}L_\lambda^2.
		\end{equation}
\end{lemma}
\begin{proof}
	Set
	\[
	A:=\dG(h),\qquad \Gamma:=\Gamma_\lambda,\qquad r:=\frac{2-\theta}{1-\theta}.
	\]
	If either $\Tr(A\Gamma)=\infty$ or $\Tr(\cN^r\Gamma)=\infty$, then the desired inequality is
	trivial. Hence we may assume that both quantities are finite.
	
		We use the particle-number decomposition
		\[
		\mathfrak F=\bigoplus_{n=0}^\infty \gH_n,\qquad \gH_n:=\gH^{\otimes_s n},
		\]
	and let $\Pi_n$ be the orthogonal projection onto $\gH_n$. By the definition of second
	quantization,
	\[
	A=0\oplus \bigoplus_{n\ge 1} A_n,
	\qquad
	A_n:=\sum_{j=1}^n h_j
	\quad\text{on }\gH_n,
	\]
	whereas
	\[
	\cN\big|_{\gH_n}=n\,I_{\gH_n}.
	\]
	Moreover, since $[\cN,\Gamma]=0$, the state $\Gamma$ is block-diagonal with respect to the
	particle-number decomposition:
	\[
	\Gamma=\bigoplus_{n=0}^\infty \Gamma_n,
	\qquad
	\Gamma_n:=\Pi_n\Gamma\Pi_n\ge 0.
	\]
	Therefore
	\[
	\Tr(A^\theta \cN^{2-\theta}\Gamma)
	=
	\sum_{n\ge 0} n^{2-\theta}\Tr(A_n^\theta\Gamma_n).
	\]
	
	Fix $t>0$. For every $x\ge 0$, Young's inequality gives
	\[
	n^{2-\theta}x^\theta
	=
	(tx)^\theta
	\Big(t^{-\frac{\theta}{1-\theta}}n^r\Big)^{1-\theta}
	\le
	\theta\,tx+(1-\theta)t^{-\frac{\theta}{1-\theta}}n^r.
	\]
	Applying the functional calculus to the positive self-adjoint operator $A_n$, we obtain
	the quadratic-form inequality on $\gH_n$,
	\[
	n^{2-\theta}A_n^\theta
	\le
	\theta\,tA_n+(1-\theta)t^{-\frac{\theta}{1-\theta}}n^r I_{\gH_n}.
	\]
	Taking the trace against $\Gamma_n\ge 0$ and summing over $n$, we infer that
	\[
	\Tr(A^\theta \cN^{2-\theta}\Gamma)
	\le
	\theta\,t\,\Tr(A\Gamma)
	+
	(1-\theta)t^{-\frac{\theta}{1-\theta}}\Tr(\cN^r\Gamma).
	\]
	
	Now write
	\[
	a:=\Tr(A\Gamma),
	\qquad
	b:=\Tr(\cN^r\Gamma).
	\]
	If $a=0$ or $b=0$, then the desired inequality is immediate. Hence we may assume that
	$a,b>0$ and choose
	\[
	t=\Big(\frac{b}{a}\Big)^{1-\theta}.
	\]
	Then
	\[
	\theta\,ta+(1-\theta)t^{-\frac{\theta}{1-\theta}}b
	=
	\theta\,a^\theta b^{1-\theta}+(1-\theta)a^\theta b^{1-\theta}
	=
	a^\theta b^{1-\theta}.
	\]
	Hence
	\[
	\Tr(A^\theta \cN^{2-\theta}\Gamma)
	\le
	\Tr(A\Gamma)^\theta\,\Tr(\cN^r\Gamma)^{1-\theta},
	\]
	which is exactly the first result.
	
	On the $n$-particle sector,
	\[
	\mathbb{H}_s=\sum_{j=1}^n h_j^s
	\le
	\Bigl(\sum_{j=1}^n h_j\Bigr)^s n^{1-s}
	=
	\dG(h)^s\cN^{1-s},
	\qquad 0<s<1,
	\]
	so, using $[\dG(h),\cN]=0$,
	\[
	\mathbb{H}_{s_1}\mathbb{H}_{s_2}
	\le
	\dG(h)^\theta\cN^{2-\theta}.
	\]
	Applying the first part of the lemma and then Lemma~\ref{lem:first-kinetic-log2} together with Proposition~\ref{prop:number-moments-log} (with $t=0$ and $q=\frac{2-\theta}{1-\theta}$) yields
	\[
	\Tr\!\big(\mathbb{H}_{s_1}\mathbb{H}_{s_2}\Gamma_\lambda\big)
	\le
	C\,\lambda^{-2\theta}\,\lambda^{-(2-\theta)}L_\lambda^{2-\theta}.
	\]
	Since $L_\lambda\ge 1$ and $2-\theta\le 2$, multiplying by $\lambda^4$ gives \eqref{eq:shell-energy-number-moment-local}.
	
\end{proof}

\section{Decomposition of the high-frequency remainders }\label{sec:dec}

This section isolates the high-frequency errors created by the low-frequency localization and records the preliminary tail inputs needed for the lower bound. We proceed in two steps. First, after introducing the cutoff $P$, we identify the two primary remainders $HF_{\neq 0}$ and $HF_0$ coming from the nonzero and zero Fourier modes. Second, after introducing a second cutoff $\Lambda_1$, we split $HF_{\neq 0}$ into a shell part $\mathcal E_{\rm shell}$ and a genuine tail part $\mathcal E_{\rm tail}$. The outcome of the section is that the lower-bound analysis is reduced to three explicit error terms---$HF_0$, $\mathcal E_{\rm shell}$, and $\mathcal E_{\rm tail}$---which will be handled separately in Sections~\ref{sec:high0}--\ref{sec:tail-vanishing-checked}.

\paragraph{First cutoff}
We begin with the first cutoff $P$, which localizes the interaction to the low-frequency sector and isolates the two primary high-frequency remainders.

Recall that $H:=\gH=L^2(\T^2)$ and fix $h=-\Delta+1$. For $\Lambda\ge1$ define the spectral cutoff
\[
P:=P_\Lambda:=\mathbf 1_{h\le \Lambda^2},\qquad Q:=1-P.
\]
Then $PH$ is finite-dimensional with $K:=\Tr(P)<\infty$.

For $k\in\Z^2$ let $M_{e_k}$ denote the multiplication operator $(M_{e_k}f)(x)=e_k(x)f(x)$ on $H$.
Let $\Gamma_0=Z_0^{-1}e^{-\lambda \dG(h)}$ be the free Gibbs state. Define
\[
N_{0,P}:=\Tr(P\gamma_0),\qquad N_{0,Q}:=\Tr(Q\gamma_0),\qquad \cN_P:=\dG(P),\qquad \cN_Q:=\dG(Q),
\]
where $\gamma_0$ is the $1$-body reduced density matrices of $\Gamma_0$.

Define the two-body operator on $H^{\otimes2}$ corresponding to the $k\neq0$ part by
\begin{equation}\label{eq:Vneq0}
	V_{\neq0}
	:=
\sum_{k\neq0}\hvb(k)\,(M_{e_k}\otimes M_{e_{-k}}).
\end{equation}
Its low-frequency compression is
\[
V_{\neq0,P}:=P^{\otimes2}V_{\neq0}P^{\otimes2}\quad\text{on }(PH)^{\otimes2}.
\]
The finer projected notation $e_{k,P}$, $\rho_{k,P}$, and $W_P^{\mathrm{re}}$ will only be needed in the finite-dimensional comparison and is therefore deferred to Subsection~\ref{sec:low-freq}.

We define the high-frequency error for $k\neq0$:
\begin{equation}\label{eq:HF-final}
	HF_{\neq0}(\Gamma)
	:=
	\frac{\lambda^2}{2}
	\mathrm{Tr}_{H^{\otimes_s2}}
	\big[
	\big(V_{\neq0}-V_{\neq0,P}\big)\Gamma^{(2)}
	\big].
\end{equation}
By construction, $V_{\neq0}-V_{\neq0,P}$ collects precisely those
two-body contributions in which at least one particle lies in the
high-energy subspace $QH$.

For $k=0$ we define
\begin{equation}\label{eq:HF-zero}
	HF_0(\Gamma)
	:=
	\frac{\lambda^2}{2(2\pi)^2}
	\Tr_{\mathfrak F(H)}
	\Big[\Big(
	2(\cN_P-N_{0,P})(\cN_Q-N_{0,Q})
	+
	(\cN_Q-N_{0,Q})^2
	\Big)\Gamma\Big].
\end{equation}


\paragraph{Second cutoff} We introduce an auxiliary cutoff $\Lambda_1 > \Lambda$ and define
\[
P_1 := \1_{\Lambda^2 < h \le \Lambda_1^2},
\qquad
Q_1 := \1_{h > \Lambda_1^2},
\qquad
S := P + P_1 = \1_{h \le \Lambda_1^2},
\]
and
\[
\Lambda_1=\lambda^{-{\alpha_1}}
\qquad\text{with}\qquad
{\alpha_1}>\frac12.
\]
The second cutoff is introduced to separate the intermediate shell $P_1H$ from the genuine far tail
$Q_1H$. The shell part still contains all legs below the scale $\Lambda_1$ and will be treated
in Section~\ref{sec:shell-vanishing-rewritten} by mode-by-mode fluctuation estimates. The reason for
this additional decomposition is that the blocks involving the whole high-frequency projector $Q$ are
not suited to direct moment bounds: because  summing the fractional Bessel
coefficients in those channels gives no useful absolute convergence, so a crude $Q$-sector estimate
would retain  divergence. Splitting off the shell isolates the configurations that still
carry a low-frequency leg and can therefore be controlled by fluctuation observables, while the genuine
far tail $Q_1$ can be handled instead through positivity of the full Bessel operator and moments of the
far-tail number observable. The tail part, by contrast, contains at least one
$Q_1$-leg and can therefore be related to the tail number observable
\[
\cN_{Q_1}:=\dG(Q_1), 	\qquad
N_{0,Q_1}:=\Tr(Q_1\gamma_0).
\]
With this notation in
place, the second decomposition becomes completely explicit.

On the two-particle space we set
\[
\Pi := P^{\otimes 2},
\qquad
\Pi_1 := S^{\otimes 2},
\qquad
R_1 := 1 - \Pi_1.
\]

Recall that
\[
HF_{\neq 0}(\Gamma)
=
\frac{\lambda^2}{2}
\Tr_{H^{\otimes_s 2}}
\!\Big[
\bigl(V_{\neq 0} - \Pi V_{\neq 0} \Pi\bigr)\Gamma^{(2)}
\Big].
\]
We now split this remainder into a shell part and a tail part:
\begin{equation}\label{dec:HFneq0}
	HF_{\neq 0}(\Gamma)
	=
	\mathcal E_{\rm shell}(\Gamma)
	+
	\mathcal E_{\rm tail}(\Gamma),
\end{equation}
where
\[
\mathcal E_{\rm shell}(\Gamma)
:=
\frac{\lambda^2}{2}
\Tr_{H^{\otimes_s 2}}
\!\Big[
\bigl(\Pi_1 V_{\neq 0}\Pi_1 - \Pi V_{\neq 0}\Pi\bigr)\Gamma^{(2)}
\Big],
\]
and
\[
\mathcal E_{\rm tail}(\Gamma)
:=
\frac{\lambda^2}{2}
\Tr_{H^{\otimes_s 2}}
\!\Big[
\bigl(V_{\neq 0} - \Pi_1 V_{\neq 0}\Pi_1\bigr)\Gamma^{(2)}
\Big].
\]
\paragraph{Outputs used later.}
We will use the Gibbs variational principle to analyze  the lower-bound of the free energy, which is reduced to three separate high-frequency
remainders,
\[
HF_0(\Gamma_\lambda),
\qquad
\mathcal E_{\rm shell}(\Gamma_\lambda),
\qquad
\mathcal E_{\rm tail}(\Gamma_\lambda).
\]
They will be treated one by one in the next three sections. Section~\ref{sec:high0} handles the
zero-mode remainder $HF_0$ and proves Theorem~\ref{cor:HF0-vanishes}, namely
\[
|HF_0(\Gamma_\lambda)|
\le
C(1+|\log\lambda|^2)(\lambda^{\alpha/2}+\lambda^{3/4-\alpha})
\qquad
\Bigl(\Lambda=\lambda^{-\alpha},\ 0<\alpha<\frac34\Bigr).
\]
{\rmk{Section~\ref{sec:shell-vanishing-rewritten} handles the shell term and proves Theorem~\ref{prop:shell-vanishing-local}; for every fixed $\frac32<\beta\le2$ one gets
\[
|\mathcal E_{\rm shell}(\Gamma_\lambda)|\longrightarrow0.
\]
 Section~\ref{sec:tail-vanishing-checked} handles the tail term and proves Theorem~\ref{cor:tail-lower-bound}, namely
\[
\mathcal E_{\rm tail}(\Gamma_\lambda)\longrightarrow0,
\]
provided $\alpha>0$ is chosen sufficiently small and $\alpha_1>\frac12$ sufficiently close to $\frac12$.}}

To control the contribution from high momenta $HF_0(\Gamma_\lambda)$ and $\mathcal E_{\rm shell}(\Gamma_\lambda)$, we will repeatedly use the following recent result from \cite[Theorem~2]{DeuNamNap-25}, which improves \cite[Theorem~7.1]{LewNamRou-21}. In Sections~\ref{sec:high0} and \ref{sec:shell-vanishing-rewritten}, the relevant observables are localized fluctuations, and the theorem below converts their second moments into a first-moment input plus a double-commutator term. The tail contribution $\mathcal E_{\rm tail}(\Gamma_\lambda)$ is of a different nature: after the second cutoff, the useful lower bound comes instead from the positivity of the full Bessel operator, so the argument in Section~\ref{sec:tail-vanishing-checked} relies on a separate decomposition.

\begin{theorem}[Second order correlation inequality] \label{thm:correlation-intro} Let $A$ be a self-adjoint operator on a separable Hilbert space such that $\Tr [e^{-sA} ] <+\infty$ holds for all $s> 0$. Let $B$ be a symmetric operator such that $B$ is $A$-relatively bounded with a relative bound strictly smaller than $1$. We also assume that the perturbed Gibbs states
	\begin{equation}\label{eq:Gibbs-t-intro}
		G_t = \frac{\exp(-A+tB)}{\Tr [ \exp(-A+tB) ]} , \quad t\in [-1,1]
	\end{equation}
	satisfy
	\begin{equation}\label{eq:CRI-condition-intro}
		\sup_{t\in [-1,1]}|\Tr (B G_t )| \le a.
	\end{equation}
	Then we have
	\begin{equation}\label{eq:StahlA2-intro}
		\Tr [ B^2 G_0 ] \leq a e^{a} + \frac{1}{4} \Tr[ [B, [A, B]] G_0].
	\end{equation}
\end{theorem}

We now rewrite  $\mathcal E_{\rm tail}(\Gamma)$ in a form adapted to the lower-bound argument.
The tail part is treated through the full positive Bessel operator
\[
V
=
\sum_{k\in\mathbb Z^2}
\hvb(k)\,(M_{e_k}\otimes M_{e_{-k}}),
\qquad
V_{\neq 0}
=
V-\frac{1}{(2\pi)^2}\,1_{H^{\otimes_s 2}}.
\]

\begin{proposition}[Second cutoff decomposition of $HF_{\neq 0}$]
	\label{prop:second-cutoff-decomp}
	For every state $\Gamma$ on $\mathfrak F(H)$, one has
	\[
	\mathcal E_{\rm tail}(\Gamma)
	=
	\frac{\lambda^2}{2}
	\Tr_{H^{\otimes_s 2}}
	\!\Big[
	\bigl(
	R_1 V R_1 + R_1 V \Pi_1 + \Pi_1 V R_1
	\bigr)\Gamma^{(2)}
	\Big]
	-
	\frac{\lambda^2}{2(2\pi)^2}
	\Tr_{H^{\otimes_s 2}}(R_1\Gamma^{(2)}).
	\]
	In particular, since $V\ge 0$,
	\[
	\mathcal E_{\rm tail}(\Gamma)
	\ge
	\frac{\lambda^2}{2}
	\Tr_{H^{\otimes_s 2}}
	\!\Big[
	\bigl(
	R_1 V \Pi_1 + \Pi_1 V R_1
	\bigr)\Gamma^{(2)}
	\Big]
	-
	\frac{\lambda^2}{2(2\pi)^2}
	\Tr_{H^{\otimes_s 2}}(R_1\Gamma^{(2)}).
	\]
\end{proposition}

\begin{proof}
	Using
	\[
	V_{\neq 0}
	=
	V-\frac{1}{(2\pi)^2}\,1_{H^{\otimes_s 2}},
	\]
	we get
	\[
	V_{\neq 0}-\Pi_1V_{\neq 0}\Pi_1
	=
	V-\Pi_1V\Pi_1
	-
	\frac{1}{(2\pi)^2}R_1.
	\]
	Next,
	\[
	V-\Pi_1V\Pi_1
	=
	(\Pi_1+R_1)V(\Pi_1+R_1)-\Pi_1V\Pi_1
	=
	R_1VR_1+R_1V\Pi_1+\Pi_1VR_1.
	\]
	Substituting this into the definition of $\mathcal E_{\rm tail}$ yields the claimed identity.
	Since $V\ge 0$, we have $R_1VR_1\ge 0$, and the lower bound follows immediately.
\end{proof}

Proposition~\ref{prop:second-cutoff-decomp} isolates exactly the part of the tail sector that matters in the lower bound. After the positive contribution $R_1VR_1$ is discarded, the only term with an explicit unfavorable sign is the scalar correction $\Tr_{H^{\otimes_s 2}}(R_1\Gamma^{(2)})$. The mixed $S/Q_1$ channels remain inside the positive Bessel form and will be estimated later, whereas the scalar correction will be reduced to moments of the far-tail number operator $\cN_{Q_1}$.

\begin{remark}[Surviving channels after the second cutoff]
	\label{rem:surviving-channels}
	After discarding the positive contribution $R_1VR_1$, the only channels which remain in the lower
	bound are
	\[
	(P,P_1),\quad (P_1,P),\quad (P_1,P_1),
	\]
	coming from the shell term, and
	\[
	(P,Q_1),\quad (P_1,Q_1),\quad (Q_1,P),\quad (Q_1,P_1),
	\]
	coming from the mixed tail term. In other words, the second cutoff isolates two qualitatively different
	mechanisms: shell interactions, where both legs remain below $\Lambda_1$, and mixed tail channels,
	where exactly one leg lies in $Q_1H$ after the positive $Q_1Q_1$ contribution has been discarded.
	The pure $Q_1Q_1$ channel is entirely contained in $R_1VR_1$ and therefore needs no further
	estimate.
\end{remark}

\section{Control of the zero-mode high-frequency remainders  }\label{sec:high0}
%
The goal of this section is to show that the high-frequency part of the renormalized $k=0$
channel is negligible in the interacting Gibbs state. The theorem below records the form needed later
in the lower-bound argument. To prove the quantitative fluctuation bound, we apply the second-order
correlation inequality from Theorem~\ref{thm:correlation-intro} to the observable
$\frac{\lambda}{2}(\cN_Q-N_{0,Q})$. This reduces the problem to estimating the first moment under
perturbed Gibbs states and the double commutator $[\cN_Q,[W^{\mathrm{re}},\cN_Q]]$. Because
$\hvb$ is not summable, the latter cannot be treated by a naive density-square argument;
instead, we rewrite it as a quartic operator, estimate its coefficients through the asymmetric quartic
Hilbert--Schmidt bound of Lemma~\ref{lem:aH-S}, and then control the resulting expression by the
kinetic operator. This reduction is effective only when a surviving quartic monomial contains at least
one $P$-leg coming from a cutoff crossing. A key point of the block decomposition below is therefore
that the pure $Q\!\to\!Q$ / $Q\!\to\!Q$ block vanishes in the double commutator, while every
surviving block still carries such a $P$-leg; making this structure explicit is precisely what forces
the more delicate decomposition carried out in this section.

\begin{theorem}[Quantitative vanishing of the zero-mode remainder]
	\label{cor:HF0-vanishes}
	Let
	\(
	\Lambda=\lambda^{-\alpha},
	0<\alpha<\frac34.
	\)
	Then
	\begin{equation}\label{eq:HF0-vanishes-rate}
		|HF_0(\Gamma_\lambda)|
		\le
		C(1+|\log\lambda|^2)(\lambda^{\alpha/2}+\lambda^{3/4-\alpha}).
	\end{equation}
	In particular,
	\begin{equation}\label{eq:HF0-vanishes}
		HF_0(\Gamma_\lambda)\longrightarrow 0
		\qquad (\lambda\downarrow0).
	\end{equation}
\end{theorem}

\subsection{Double commutator estimates}

Recall that
\[
HF_0(\Gamma)
=
\frac{\lambda^2}{2(2\pi)^2}
\Tr_{\mathfrak F(H)}
\Big[\Big(
2(\cN_P-N_{0,P})(\cN_Q-N_{0,Q})
+
(\cN_Q-N_{0,Q})^2
\Big)\Gamma\Big].
\]
Because the zero mode has already been centered through $(\cN-N_0)^2$, the only genuinely new
quantity is the high-frequency fluctuation of the number operator. Once this fluctuation is shown to
be small, the vanishing of $HF_0(\Gamma_\lambda)$ follows immediately from the previously obtained
control of the low-frequency number fluctuations.

The main estimate proved in this section is the following.

\begin{theorem}[High-frequency number fluctuations]
	\label{prop:high-frequency-fluctuations}
	Let $\Lambda=\lambda^{-\alpha}$ for $\alpha\in(0,3/4)$.
	Then
	\begin{equation}\label{eq:high-frequency-fluctuation-bound}
		\lambda^2
		\Tr_{\mathfrak F(H)}
		\Big[
		(\cN_Q-N_{0,Q})^2\,\Gamma_\lambda
		\Big]
		\le
		C
		(1+|\log\lambda|^2)(\Lambda^{-1}
		+
		\lambda^{\frac32}\Lambda^2)
		.
	\end{equation}
\end{theorem}

The proof is organized in three steps. We first
analyze the double commutator $[\cN_Q,[W_{\neq0}^{\mathrm{re}},\cN_Q]]$ and identify exactly which
$P/Q$-blocks survive. We then estimate the surviving blocks by an asymmetric quartic Hilbert--
Schmidt argument. Finally, we combine the resulting commutator bound with the second-order
correlation inequality Theorem \ref{thm:correlation-intro} for the observable
\(
B_Q:=\frac{\lambda}{2}(\cN_Q-N_{0,Q})
\)
and compare the perturbed interacting and free Gibbs states.

To organize the commutator, it is convenient to classify each quartic monomial according to how the two legs move across the cutoff. This bookkeeping has two advantages. First, the coefficient $\Delta_{p,q,k}$ only depends on whether each incoming and outgoing momentum lies in $P$ or $Q$, so the block structure makes the value of $\Delta_{p,q,k}$ immediate. Second, it separates the configurations for which a low-frequency leg forces $|k|\lesssim \Lambda$ from the genuinely high--high configurations where no such restriction is available.

\paragraph{Block labels.}
We label each leg by whether it stays in $P$ or $Q$ before/after scattering:
\[
\text{first leg: } p\to p+k\in\{P\to P,\ P\to Q,\ Q\to P,\ Q\to Q\},
\]
\[
\text{second leg: } q\to q-k\in\{P\to P,\ P\to Q,\ Q\to P,\ Q\to Q\}.
\]
Thus a \emph{block} is a pair $(\sigma_1,\sigma_2)$ of such transitions.

\paragraph{Finite $k$--sum blocks.}
If at least one leg is of type $P\to P$, then necessarily $|k|\lesssim \Lambda$.
Indeed, if $p,p+k\in\mathcal K$ for $\mathcal{K}=\{k\in\Z^2, h(k)\leq \Lambda^2\}$ then
\[
|k|=|(p+k)-p|\le |p+k|+|p|\lesssim \Lambda,
\]
and similarly for $q,q-k\in\mathcal K$.

Hence the following $7$ blocks have a \emph{finite} $k$--sum (automatically restricted to $|k|\lesssim \Lambda$):

\begin{align*}
	\text{(A1)}\quad (P\to P,\ P\to P) &:\quad
	p,p+k\in\mathcal K,\ \ q,q-k\in\mathcal K,\\
	\text{(A2)}\quad (P\to P,\ P\to Q) &:\quad
	p,p+k\in\mathcal K,\ \ q\in\mathcal K,\ q-k\notin\mathcal K,\\
	\text{(A3)}\quad (P\to P,\ Q\to P) &:\quad
	p,p+k\in\mathcal K,\ \ q\notin\mathcal K,\ q-k\in\mathcal K,\\
	\text{(A4)}\quad (P\to P,\ Q\to Q) &:\quad
	p,p+k\in\mathcal K,\ \ q,q-k\notin\mathcal K,\\
	\text{(A5)}\quad (P\to Q,\ P\to P) &:\quad
	p\in\mathcal K,\ p+k\notin\mathcal K,\ \ q,q-k\in\mathcal K,\\
	\text{(A6)}\quad (Q\to P,\ P\to P) &:\quad
	p\notin\mathcal K,\ p+k\in\mathcal K,\ \ q,q-k\in\mathcal K,\\
	\text{(A7)}\quad (Q\to Q,\ P\to P) &:\quad
	p,p+k\notin\mathcal K,\ \ q,q-k\in\mathcal K.
\end{align*}

\paragraph{Non-finite $k$--sum blocks.}
If \emph{neither} leg is $P\to P$, then there is \emph{no} geometric restriction forcing
$|k|\lesssim \Lambda$. These $9$ blocks allow $k$ of arbitrary size:

\begin{align*}
	\text{(B1)}\quad (P\to Q,\ P\to Q) &:\quad
	p\in\mathcal K,\ p+k\notin\mathcal K,\ \ q\in\mathcal K,\ q-k\notin\mathcal K,\\
	\text{(B2)}\quad (P\to Q,\ Q\to P) &:\quad
	p\in\mathcal K,\ p+k\notin\mathcal K,\ \ q\notin\mathcal K,\ q-k\in\mathcal K,\\
	\text{(B3)}\quad (P\to Q,\ Q\to Q) &:\quad
	p\in\mathcal K,\ p+k\notin\mathcal K,\ \ q,q-k\notin\mathcal K,\\
	\text{(B4)}\quad (Q\to P,\ P\to Q) &:\quad
	p\notin\mathcal K,\ p+k\in\mathcal K,\ \ q\in\mathcal K,\ q-k\notin\mathcal K,\\
	\text{(B5)}\quad (Q\to P,\ Q\to P) &:\quad
	p\notin\mathcal K,\ p+k\in\mathcal K,\ \ q\notin\mathcal K,\ q-k\in\mathcal K,\\
	\text{(B6)}\quad (Q\to P,\ Q\to Q) &:\quad
	p\notin\mathcal K,\ p+k\in\mathcal K,\ \ q,q-k\notin\mathcal K,\\
	\text{(B7)}\quad (Q\to Q,\ P\to Q) &:\quad
	p,p+k\notin\mathcal K,\ \ q\in\mathcal K,\ q-k\notin\mathcal K,\\
	\text{(B8)}\quad (Q\to Q,\ Q\to P) &:\quad
	p,p+k\notin\mathcal K,\ \ q\notin\mathcal K,\ q-k\in\mathcal K,\\
	\text{(B9)}\quad (Q\to Q,\ Q\to Q) &:\quad
	p,p+k\notin\mathcal K,\ \ q,q-k\notin\mathcal K.
\end{align*}

We begin with the exact commutator structure behind the $Q$-fluctuation observable. Recall $W^{\mathrm{re}}$ defined in \eqref{eq:Wre-def}.

\begin{lemma}[Double commutator structure for $\cN_Q$]
	\label{lem:NQ-double-commutator-direct}
	The following hold.
	
	\medskip
	
	\noindent\textnormal{(i) Exact formula.}
	For
	\[
	\cN_Q=\dG(Q),
	\qquad
	M_{p,q,k}:=a_{p+k}^\ast a_{q-k}^\ast a_q a_p,
	\]
	define
	\[
	\Delta_{p,q,k}:=Q_{p+k}+Q_{q-k}-Q_q-Q_p,
	\qquad
	Q_r:=\mathbf 1_{\{h(r)>\Lambda^2\}}\in\{0,1\}.
	\]
	Then, as a quadratic form on finite-particle vectors,
	\begin{equation}\label{eq:exact-double-commutator-DQ}
		[\cN_Q,[W^{\mathrm{re}},\cN_Q]]
		=
		-\frac{\lambda^2}{2(2\pi)^2}
		\sum_{k\neq0}\hvb(k)\sum_{p,q\in\mathbb Z^2}
		\Delta_{p,q,k}^2\, M_{p,q,k}.
	\end{equation}
	
	\medskip
	
	\noindent\textnormal{(ii) Blockwise values of $\Delta_{p,q,k}$.}
	With the $A/B$-block decomposition introduced earlier, one has
	\[
	\Delta_{p,q,k}=0
	\quad\text{for}\quad
	A1,\ A4,\ A7,\ B2,\ B4,\ B9,
	\]
	and
	\[
	|\Delta_{p,q,k}|=1
	\quad\text{for}\quad
	A2,\ A3,\ A5,\ A6,\ B3,\ B6,\ B7,\ B8,
	\]
	while
	\[
	|\Delta_{p,q,k}|=2
	\quad\text{for}\quad
	B1,\ B5.
	\]
	Hence only the ten blocks
	\[
	A2,\ A3,\ A5,\ A6,\ B1,\ B3,\ B5,\ B6,\ B7,\ B8
	\]
	survive in \eqref{eq:exact-double-commutator-DQ}.
	
	%
\end{lemma}

\begin{proof}
	We first prove part~\textnormal{(i)}. Since
	\[
	\cN_Q=\sum_{r\in\mathbb Z^2}Q_r\,a_r^\ast a_r,
	\]
	the CCR imply
	\begin{equation}\label{eq:DQ-on-creation-annihilation}
		[\cN_Q,a_j^\ast]=Q_j a_j^\ast,
		\qquad
		[\cN_Q,a_j]=-Q_j a_j.
	\end{equation}
	Hence,
	the Leibniz rule gives
	\[
	[\cN_Q,M_{p,q,k}]
	=
	(Q_{p+k}+Q_{q-k}-Q_q-Q_p)M_{p,q,k}
	=
	\Delta_{p,q,k}M_{p,q,k},
	\]
	which implies that
	\[
	[\cN_Q,[M_{p,q,k},\cN_Q]]
	=
	[\cN_Q,-\Delta_{p,q,k}M_{p,q,k}]
	=
	-\Delta_{p,q,k}^2M_{p,q,k}.
	\]
	Summing over $p,q\in\mathbb Z^2$ and $k\neq0$ yields
	\eqref{eq:exact-double-commutator-DQ}.
	
	\medskip
	
	For part~\textnormal{(ii)}, it is convenient to record the contribution of a single leg to
	$\Delta_{p,q,k}$.
	Therefore $\Delta_{p,q,k}=0$ exactly when the two increments are either $0+0$, $1-1$, or
	$-1+1$, which gives the blocks
	\[
	A1,\ A4,\ A7,\ B2,\ B4,\ B9.
	\]
	Likewise, $|\Delta_{p,q,k}|=1$ when exactly one leg crosses the cutoff, namely on
	\[
	A2,\ A3,\ A5,\ A6,\ B3,\ B6,\ B7,\ B8,
	\]
	and $|\Delta_{p,q,k}|=2$ when both legs cross in the same direction, namely on
	\[
	B1,\ B5.
	\]
	This proves the claimed list of surviving blocks.
\end{proof}

\begin{remark}\label{rem:B9-disappears}
	The block decomposition is already useful at the level of \eqref{eq:exact-double-commutator-DQ}. In particular, the purely high--high block
	\[
	B9=(Q\to Q,\,Q\to Q)
	\]
	disappears because $\Delta_{p,q,k}=0$ there. This is crucial for the later Hilbert--Schmidt bounds: $B9$ carries no $P$-projector, hence no low-frequency input and no geometric restriction on $k$, so it would be the hardest channel to estimate directly. Every surviving block still contains at least one cutoff crossing, and therefore at least one $P$-leg from which the later bounds extract the truncated kernels $S_k^{(P)}$ and in the following we can take $s_1+s_2<1$.
\end{remark}

	By the exact identity \eqref{eq:exact-double-commutator-DQ},
\(
[\cN_Q,[W_{\neq0}^{\rm re},\cN_Q]]\)
is a quartic operator of the form
\[
Q_\gamma=\sum_{r_1,r_2,u_1,u_2\in\Z^2}
\gamma_{r_1,r_2,u_1,u_2}\,a_{r_1}^\ast a_{r_2}^\ast a_{u_1}a_{u_2},
\]
with coefficients
\[
\gamma_{p+k,q-k,q,p}
=
-\frac{\lambda^2}{2(2\pi)^2}\,\hvb(k)\,\Delta_{p,q,k}^2,
\qquad (k\neq0),
\]
and $\gamma_{r_1,r_2,u_1,u_2}=0$ otherwise.

The next ingredient is a Hilbert--Schmidt reduction that converts the surviving blocks into weighted quartic kernels.


\paragraph{Hilbert--Schmidt reduction.}
Let
\[
\mathbb H_s:=\dG(h^s)=\sum_{p\in\mathbb Z^2} h(p)^s a_p^\ast a_p.
\]

To turn the exact block decomposition into operator bounds, we now use an asymmetric quartic Hilbert--Schmidt estimate.

\begin{lemma}[Asymmetric quartic Hilbert--Schmidt bound]\label{lem:aH-S}
	Let \(s_1,s_2\ge 0\). For a quartic kernel
	\[
	\gamma=\gamma_{r_1,r_2,u_1,u_2},
	\]
	define
	\[
	Q_\gamma
	:=
	\sum_{r_1,r_2,u_1,u_2\in\mathbb Z^2}
	\gamma_{r_1,r_2,u_1,u_2}\,
	a_{r_1}^\ast a_{r_2}^\ast a_{u_1}a_{u_2}.
	\]
	Define the asymmetrically rescaled kernel
	\[
	\widetilde\gamma^{(s_1,s_2)}_{r_1,r_2,u_1,u_2}
	:=
	h(r_1)^{-s_1/2}h(r_2)^{-s_2/2}
	\gamma_{r_1,r_2,u_1,u_2}
	h(u_1)^{-s_2/2}h(u_2)^{-s_1/2}.
	\]
	If \(\widetilde\gamma^{(s_1,s_2)}\in \ell^2((\mathbb Z^2)^4)\), then for every
	\(\psi\in\mathfrak F_{\mathrm{fin}}\),
	\[
	\bigl|\langle \psi,Q_\gamma\psi\rangle\bigr|
	\le
	\|\widetilde\gamma^{(s_1,s_2)}\|_{\ell^2}\,
	\langle \psi,(\mathbb{H}_{s_1}\mathbb{H}_{s_2}-\mathbb{H}_{s_1+s_2})\psi\rangle
	\le
	\|\widetilde\gamma^{(s_1,s_2)}\|_{\ell^2}\,
	\langle \psi,\mathbb{H}_{s_1}\mathbb{H}_{s_2}\psi\rangle.
	\]
	Here $\mathfrak F_{\mathrm{fin}}$ denotes finite-particle subspace of the bosonic Fock space.
\end{lemma}

\begin{proof}
	Set
	\[
	b_u:=h(u)^{s_1/2}a_u,\qquad
	c_u:=h(u)^{s_2/2}a_u.
	\]
	Then
	\[
	a_{u_2}=h(u_2)^{-s_1/2}b_{u_2},\qquad
	a_{u_1}=h(u_1)^{-s_2/2}c_{u_1},
	\]
	and similarly
	\[
	a_{r_1}^\ast=h(r_1)^{-s_1/2}b_{r_1}^\ast,\qquad
	a_{r_2}^\ast=h(r_2)^{-s_2/2}c_{r_2}^\ast.
	\]
	Hence
	\[
	Q_\gamma
	=
	\sum_{r_1,r_2,u_1,u_2}
	\widetilde\gamma^{(s_1,s_2)}_{r_1,r_2,u_1,u_2}\,
	b_{r_1}^\ast c_{r_2}^\ast c_{u_1}b_{u_2}.
	\]
	For \(\psi\in\mathfrak F_{\mathrm{fin}}\), define
	\[
	X_{u_1,u_2}:=c_{u_1}b_{u_2}\psi,\qquad
	Y_{r_1,r_2}:=c_{r_2}b_{r_1}\psi.
	\]
	Then
	\[
	\langle \psi,Q_\gamma\psi\rangle
	=
	\sum_{r_1,r_2,u_1,u_2}
	\widetilde\gamma^{(s_1,s_2)}_{r_1,r_2,u_1,u_2}
	\langle Y_{r_1,r_2},X_{u_1,u_2}\rangle.
	\]
	By Cauchy--Schwarz,
	\[
	|\langle \psi,Q_\gamma\psi\rangle|^2
	\le
	\|\widetilde\gamma^{(s_1,s_2)}\|_{\ell^2}^2
	\Bigl(\sum_{r_1,r_2}\|Y_{r_1,r_2}\|^2\Bigr)
	\Bigl(\sum_{u_1,u_2}\|X_{u_1,u_2}\|^2\Bigr).
	\]
	Since the two factors are equal, it follows that
	\[
	|\langle \psi,Q_\gamma\psi\rangle|
	\le
	\|\widetilde\gamma^{(s_1,s_2)}\|_{\ell^2}
	\sum_{u_1,u_2}\|c_{u_1}b_{u_2}\psi\|^2.
	\]
	Now
	\[
	\sum_{u_1,u_2}\|c_{u_1}b_{u_2}\psi\|^2
	=
	\Bigl\langle\psi,
	\sum_{u_1,u_2} b_{u_2}^\ast c_{u_1}^\ast c_{u_1}b_{u_2}
	\,\psi\Bigr\rangle
	=
	\Bigl\langle\psi,
	\sum_{u_2} b_{u_2}^\ast \mathbb{H}_{s_2}b_{u_2}
	\,\psi\Bigr\rangle.
	\]
	Using
	\[
	[\mathbb{H}_{s_2},a_u]= -h(u)^{s_2}a_u,
	\]
	we get
	\[
	\mathbb{H}_{s_2}b_u=b_u(\mathbb{H}_{s_2}-h(u)^{s_2}).
	\]
	Therefore
	\[
	\sum_{u_2} b_{u_2}^\ast \mathbb{H}_{s_2}b_{u_2}
	=
	\sum_{u_2} h(u_2)^{s_1}a_{u_2}^\ast a_{u_2}(\mathbb{H}_{s_2}-h(u_2)^{s_2})
	=
	\mathbb{H}_{s_1}\mathbb{H}_{s_2}-\mathbb{H}_{s_1+s_2}.
	\]
 This proves
	\[
	|\langle \psi,Q_\gamma\psi\rangle|
	\le
	\|\widetilde\gamma^{(s_1,s_2)}\|_{\ell^2}
	\langle \psi,(\mathbb{H}_{s_1}\mathbb{H}_{s_2}-\mathbb{H}_{s_1+s_2})\psi\rangle
	\le
	\|\widetilde\gamma^{(s_1,s_2)}\|_{\ell^2}
	\langle \psi,\mathbb{H}_{s_1}\mathbb{H}_{s_2}\psi\rangle.
	\]
\end{proof}

\paragraph{Exact block decomposition of the direct commutator.}
By Lemma \ref{lem:NQ-double-commutator-direct}, only the ten blocks
\[
A2,\ A3,\ A5,\ A6,\ B1,\ B3,\ B5,\ B6,\ B7,\ B8
\]
survive. We group them into five packages according to the geometric pattern of cutoff crossings,
\[
(A2+A3),\qquad (A5+A6),\qquad (B1+B5),\qquad (B3+B6),\qquad (B7+B8),
\]
because each package has a common Hilbert--Schmidt structure. The first two packages contain one low--low leg together with one crossing leg, the third contains two same-direction crossings, and the last two contain one crossing leg together with one $Q\to Q$ leg; these are the natural model configurations for the bounds below.
Accordingly define
\begin{equation}\label{dec:com}
[\cN_Q,[W^{\mathrm{re}},\cN_Q]]
=
\mathcal C_{23}+\mathcal C_{56}+\mathcal C_{15}+\mathcal C_{36}+\mathcal C_{78},
\end{equation}
where
\[
\mathcal C_{23}
:=
-\tilde c_\lambda\sum_{k\neq0}\hvb(k)\sum_{p,q}
P_pP_{p+k}\bigl(P_qQ_{q-k}+Q_qP_{q-k}\bigr)\,M_{p,q,k},
\]
\[
\mathcal C_{56}
:=
-\tilde c_\lambda\sum_{k\neq0}\hvb(k)\sum_{p,q}
\bigl(P_pQ_{p+k}+Q_pP_{p+k}\bigr)P_qP_{q-k}\,M_{p,q,k},
\]
\[
\mathcal C_{15}
:=
-4\tilde c_\lambda\sum_{k\neq0}\hvb(k)\sum_{p,q}
\Bigl(
P_pQ_{p+k}P_qQ_{q-k}
+
Q_pP_{p+k}Q_qP_{q-k}
\Bigr)\,M_{p,q,k},
\]
\[
\mathcal C_{36}
:=
-\tilde c_\lambda\sum_{k\neq0}\hvb(k)\sum_{p,q}
\bigl(P_pQ_{p+k}+Q_pP_{p+k}\bigr)Q_qQ_{q-k}\,M_{p,q,k},
\]
\[
\mathcal C_{78}
:=
-\tilde c_\lambda\sum_{k\neq0}\hvb(k)\sum_{p,q}
Q_pQ_{p+k}\bigl(P_qQ_{q-k}+Q_qP_{q-k}\bigr)\,M_{p,q,k},
\]
for $\tilde c_\lambda=\frac{\lambda^2}{2(2\pi)^2}$.	In fact, \(\mathcal C_{23}=\mathcal C_{56}\) and \(\mathcal C_{78}=\mathcal C_{36}\) by the symmetry
	\((p,q,k)\mapsto(q,p,-k)\), the evenness \(\hvb(-k)=\hvb(k)\), and bosonic
	symmetry of the quartic monomial.

Recall that for $s>0$ we write
\[
S_k(s):=\sum_{u\in\mathbb Z^2} h(u+k)^{-s}h(u)^{-s},
\qquad
S_k^{(P)}(s):=\sum_{u\in\mathbb Z^2} P_u\,h(u+k)^{-s}h(u)^{-s},
\]
where $P_u=\mathbf 1_{\{h(u)\le \Lambda^2\}}$.

\begin{lemma}[Blockwise Hilbert--Schmidt bounds]\label{lem:com1}
	Assume \(0<s_1<1\) and \(\frac12<s_2<1\). Then
	\begin{align}
		\pm (\mathcal C_{23}+\mathcal C_{56}+\mathcal C_{15})
		&\le
		C_{s_1,s_2}\lambda^2 \Lambda^{2-s_1-s_2}
		\mathbb{H}_{s_1}\mathbb{H}_{s_2},
		\label{eq:blockwise-HS-two-cutoff}
		\\
		\pm (\mathcal C_{36}+ \mathcal C_{78})
		&\le
		C_{s_1,s_2}\lambda^2 \Lambda^{1-s_1}
		\mathbb{H}_{s_1}\mathbb{H}_{s_2},
		\label{eq:blockwise-HS-36}
	\end{align}
	In particular, as a quadratic form on \(\mathfrak F_{\mathrm{fin}}\),
	\[
	\pm [\cN_Q,[W^{\mathrm{re}}_{\neq0},\cN_Q]]
	\le
	C_{s_1,s_2}\lambda^2
	\Lambda^{2-s_1-s_2}
	\mathbb{H}_{s_1}\mathbb{H}_{s_2}.
	\]
\end{lemma}

\begin{proof}
	We begin with the two-cutoff contribution. It is enough to estimate \(\mathcal C_{23}\), since
	the arguments for \(\mathcal C_{56}\) and \(\mathcal C_{15}\) are identical.
	
	Write \(\mathcal C_{23}=Q_{\gamma_{23}}\) with kernel
	\[
	\gamma_{23;r_1,r_2,u_1,u_2}
	=
	-\tilde c_\lambda\sum_{k\neq0} \hvb(k)\,
	\delta_{r_1,u_2+k}\delta_{r_2,u_1-k}\,
	P_{u_2}P_{u_2+k}\bigl(P_{u_1}Q_{u_1-k}+Q_{u_1}P_{u_1-k}\bigr).
	\]
	Let
	\(
	\widetilde\gamma_{23}:=\widetilde\gamma_{23}^{(s_1,s_2)}.
	\)
	We have
	\[
	\|\widetilde\gamma_{23}\|_{\ell^2}^2
	=
	\tilde c_\lambda^2 \sum_{k\neq0}|\hvb(k)|^2 A_{23}^{(p)}(k)A_{23}^{(q)}(k),
	\]
	where
	\[
	A_{23}^{(p)}(k)
	:=
	\sum_{u_2}
	P_{u_2}P_{u_2+k}\,h(u_2+k)^{-s_1}h(u_2)^{-s_1},
	\]
	\[
	A_{23}^{(q)}(k)
	:=
	\sum_{u_1}
	\bigl(P_{u_1}Q_{u_1-k}+Q_{u_1}P_{u_1-k}\bigr)
	h(u_1-k)^{-s_2}h(u_1)^{-s_2}.
	\]
	Using
	\[
	A_{23}^{(p)}(k)\le S_k^{(P)}(s_1),
	\qquad
	A_{23}^{(q)}(k)\le 2S_k^{(P)}(s_2),
	\]
	we obtain
	\[
	\|\widetilde\gamma_{23}\|_{\ell^2}^2
	\le
	C\lambda^4
	\sum_{k\neq0}|\hvb(k)|^2
	S_k^{(P)}(s_1)S_k^{(P)}(s_2).
	\]
	Lemma~\ref{lem:convolution-kernels} gives
	\[
	\|\widetilde\gamma_{23}\|_{\ell^2}^2
	\le
	C_{s_1,s_2}\lambda^4
	\Lambda^{4-2s_1-2s_2}
	\sum_{k\neq0}|\hvb(k)|^2\langle k\rangle^{-2s_1-2s_2}.
	\]
	{\color{red}Since $|\hvb(k)|^2=\langle k\rangle^{-2\beta}$ and $\beta+s_1+s_2>1$ (indeed $\beta>\frac32$ and $s_1,s_2>0$), the $k$-sum is finite, hence}
	\[
	\|\widetilde\gamma_{23}\|_{\ell^2}
	\le
	C_{s_1,s_2}\lambda^2\Lambda^{2-s_1-s_2}.
	\]
	Applying Lemma~\ref{lem:aH-S} yields the first bound in
	\eqref{eq:blockwise-HS-two-cutoff}; summing the identical estimates for
	\(\mathcal C_{56}\) and \(\mathcal C_{15}\) only changes the constant.
	
	\medskip
	
	For the one-cutoff pair \(\mathcal C_{36}\), write \(\mathcal C_{36}=Q_{\gamma_{36}}\) with kernel
	\[
	\gamma_{36;r_1,r_2,u_1,u_2}
	=
	-\tilde c_\lambda\sum_{k\neq0} \hvb(k)\,
	\delta_{r_1,u_2+k}\delta_{r_2,u_1-k}\,
	\bigl(P_{u_2}Q_{u_2+k}+Q_{u_2}P_{u_2+k}\bigr)Q_{u_1}Q_{u_1-k}.
	\]
	Let \(\widetilde\gamma_{36}:=\widetilde\gamma_{36}^{(s_1,s_2)}\). As above,
	\[
	\|\widetilde\gamma_{36}\|_{\ell^2}^2
	=
	\tilde c_\lambda^2\sum_{k\neq0}|\hvb(k)|^2 A_{36}^{(p)}(k)A_{36}^{(q)}(k),
	\]
	with
	\[
	A_{36}^{(p)}(k)
	:=
	\sum_{u_2}
	\bigl(P_{u_2}Q_{u_2+k}+Q_{u_2}P_{u_2+k}\bigr)
	h(u_2+k)^{-s_1}h(u_2)^{-s_1},
	\]
	\[
	A_{36}^{(q)}(k)
	:=
	\sum_{u_1}
	Q_{u_1}Q_{u_1-k}
	h(u_1-k)^{-s_2}h(u_1)^{-s_2}.
	\]
	Using
	\[
	A_{36}^{(p)}(k)\le 2S_k^{(P)}(s_1),
	\qquad
	A_{36}^{(q)}(k)\le S_k(s_2),
	\]
	we find
	\[
	\|\widetilde\gamma_{36}\|_{\ell^2}^2
	\le
	C\lambda^4\sum_{k\neq0}|\hvb(k)|^2 S_k^{(P)}(s_1)S_k(s_2).
	\]
	Now Lemma~\ref{lem:convolution-kernels} gives
	\[
	\|\widetilde\gamma_{36}\|_{\ell^2}^2
	\le
	C_{s_1,s_2}\lambda^4\Lambda^{2-2s_1}
	\sum_{k\neq0}|\hvb(k)|^2\langle k\rangle^{-2s_1},
	\]
	{\color{red}and the $k$-sum is finite because $|\hvb(k)|^2\langle k\rangle^{-2s_1}=\langle k\rangle^{-2\beta-2s_1}\in \ell^1(\Z^2)$ for every $\beta>1$. Hence}
	\[
	\|\widetilde\gamma_{36}\|_{\ell^2}
	\le
	C_{s_1,s_2}\lambda^2\Lambda^{1-s_1}.
	\]
	Lemma~\ref{lem:aH-S} then gives \eqref{eq:blockwise-HS-36}. The final result then follows from \eqref{dec:com}.
	\end{proof}

\subsection{Proof of the main fluctuation bound.}

\begin{proof}[Proof of Theorem~\ref{prop:high-frequency-fluctuations}]
	We split the argument into three steps.
	
	\medskip
	
	\noindent\textbf{Step 1: application of the second-order correlation inequality.}
	
	Recall the self-adjoint observable
	\(
	B_Q=\frac{\lambda}{2}(\cN_Q-N_{0,Q}),
	\)
	and define the perturbed Gibbs states
	\[
	\Gamma_{\lambda,t}^{(Q)}
	:=
	\frac{e^{-\mathbb{H}_\lambda^{\mathrm{re}}+tB_Q}}
	{\Tr(e^{-\mathbb{H}_\lambda^{\mathrm{re}}+tB_Q})},
	\qquad
	\Gamma_{0,t}^{(Q)}
	:=
	\frac{e^{-\lambda \dG(h)+tB_Q}}
	{\Tr(e^{-\lambda \dG(h)+tB_Q})},
	\qquad t\in[-1,1].
	\]
	Set
	\[
	a_{Q,\lambda}:=
	\sup_{|t|\le 1}
	\left|
	\Tr(B_Q\Gamma_{\lambda,t}^{(Q)})
	\right|.
	\]
	By the second-order correlation inequality Theorem \ref{thm:correlation-intro}, we have
	\begin{equation}\label{eq:DNN-BQ}
		\Tr(B_Q^2\Gamma_\lambda)
		\le
		a_{Q,\lambda}e^{a_{Q,\lambda}}
		+
		\frac14\Tr\!\Big([B_Q,[\mathbb{H}_\lambda^{\mathrm{re}},B_Q]]\Gamma_\lambda\Big).
	\end{equation}
Therefore
	it remains to show that both terms on the right-hand side of \eqref{eq:DNN-BQ}
	tend to zero.
	
	\medskip
	
	\noindent\textbf{Step 2: the first moment of $B_Q$ under the perturbed Gibbs states is $o(1)$.}
	
	We first compare $\Gamma_{\lambda,t}^{(Q)}$ with the corresponding free perturbed state.
	Observe that
	\[
	-\lambda \dG(h)+tB_Q
	=
	-\lambda \dG(h_t)-\frac{t\lambda}{2}N_{0,Q},
	\qquad
	h_t:=h-\frac{t}{2}Q.
	\]
	Hence $\Gamma_{0,t}^{(Q)}$ is the quasi-free Gibbs state associated with $h_t$.
	
	For $|t|\le1$ and $\lambda$ small enough, the operator $h_t$ is uniformly comparable with $h$. We apply Proposition~\ref{prop:one-body-HS-safe} to have
	\begin{equation}\label{eq:HS-perturbed-Q}
		\lambda\,
		\big\|
		\sqrt h\,(\gamma_{\lambda,t}^{(Q)}-\gamma_{0,t}^{(Q)})\,\sqrt h
		\big\|_{\mathrm{HS}}
		\le
		C(1+|\log\lambda|^2),
		\qquad |t|\le1,
	\end{equation}
	where $\gamma_{\lambda,t}^{(Q)}$ and $\gamma_{0,t}^{(Q)}$ denote the one-body density matrices of
	$\Gamma_{\lambda,t}^{(Q)}$ and $\Gamma_{0,t}^{(Q)}$.
	
	By Cauchy--Schwarz in Hilbert--Schmidt norm, we obtain
	\begin{align}
		\lambda\,
		\left|
		\Tr\!\big[Q(\gamma_{\lambda,t}^{(Q)}-\gamma_{0,t}^{(Q)})\big]
		\right|\le 	
			\lambda\|Qh^{-1}\|_{\mathrm{HS}}
		\big\|\sqrt h\,(\gamma_{\lambda,t}^{(Q)}-\gamma_{0,t}^{(Q)})\,\sqrt h\big\|_{\mathrm{HS}}
		&\le
		C(1+|\log\lambda|^2)\,\|Qh^{-1}\|_{\mathrm{HS}}.
		\label{eq:Q-difference-perturbed-1}
	\end{align}
	Now, in dimension $2$,
	\[
	\|Qh^{-1}\|_{\mathrm{HS}}^2
	=
	\sum_{p:\,h(p)>\Lambda^2} h(p)^{-2}
	\le C\Lambda^{-2},
	\]
	hence
	\begin{equation}\label{eq:Q-difference-perturbed-2}
		\sup_{|t|\le1}
		\lambda\,
		\left|
		\Tr\!\big[Q(\gamma_{\lambda,t}^{(Q)}-\gamma_{0,t}^{(Q)})\big]
		\right|
		\le
		C(1+|\log\lambda|^2)\Lambda^{-1}
		\longrightarrow 0.
	\end{equation}
	
	Next, we compare the free perturbed and unperturbed high-frequency means. Since
\[
\gamma_{0,t}^{(Q)}
=
\frac{1}{e^{\lambda h_t}-1},
\qquad
\gamma_0=\frac{1}{e^{\lambda h}-1},
\]
we have
\[
\lambda\,
\Tr\!\big[Q(\gamma_{0,t}^{(Q)}-\gamma_0)\big]
=
\sum_{p:\,h(p)>\Lambda^2}
\left(
\frac{\lambda}{e^{\lambda(h(p)-t/2)}-1}
-
\frac{\lambda}{e^{\lambda h(p)}-1}
\right).
\]
For each $p$ and $|t|\le 1$, the mean value theorem gives a point $\theta_{p,t}$ between $0$ and
$t$ such that
\[
\left|
\frac{\lambda}{e^{\lambda(h(p)-t/2)}-1}
-
\frac{\lambda}{e^{\lambda h(p)}-1}
\right|
=
\frac{|t|}{2}\,\lambda^2
\frac{e^{\lambda(h(p)-\theta_{p,t}/2)}}{\bigl(e^{\lambda(h(p)-\theta_{p,t}/2)}-1\bigr)^2}.
\]
Since $h(p)>\Lambda^2$ and $|\theta_{p,t}|\le 1$,  we have
$h(p)-\theta_{p,t}/2\ge h(p)-1/2\ge \frac12 h(p)$, and therefore
\[
\lambda^2
\frac{e^{\lambda(h(p)-\theta_{p,t}/2)}}{\bigl(e^{\lambda(h(p)-\theta_{p,t}/2)}-1\bigr)^2}
\le
C\,h(p)^{-2},
\]
because $\sup_{x>0} x^2 e^x (e^x-1)^{-2}<\infty$.
Hence
\begin{align}
	\sup_{|t|\le1}
	\lambda\,
	\left|
	\Tr\!\big[Q(\gamma_{0,t}^{(Q)}-\gamma_0)\big]
	\right|
	&\le
	C\sum_{p:\,h(p)>\Lambda^2} h(p)^{-2}
	\le C\Lambda^{-2}
	\longrightarrow 0.
	\label{eq:Q-difference-free-perturbed}
\end{align}
	Combining \eqref{eq:Q-difference-perturbed-2} and
	\eqref{eq:Q-difference-free-perturbed}, we find
	\begin{align}
		a_{Q,\lambda}
		&=
		\sup_{|t|\le1}
		\left|
		\Tr(B_Q\Gamma_{\lambda,t}^{(Q)})
		\right|
		=
		\frac12
		\sup_{|t|\le1}
		\lambda\,
		\left|
		\Tr(Q\gamma_{\lambda,t}^{(Q)})-N_{0,Q}
		\right|
		\nonumber\\
		&\le
		\frac12
		\sup_{|t|\le1}
		\lambda\,
		\left|
		\Tr\!\big[Q(\gamma_{\lambda,t}^{(Q)}-\gamma_{0,t}^{(Q)})\big]
		\right|
		+
		\frac12
		\sup_{|t|\le1}
		\lambda\,
		\left|
		\Tr\!\big[Q(\gamma_{0,t}^{(Q)}-\gamma_0)\big]
		\right|\lesssim (1+|\log\lambda|^2)\Lambda^{-1}.
		\label{eq:aQlambda-o1}
	\end{align}
	
	\medskip
	
	\noindent\textbf{Step 3: the double commutator term.}
	
	Using
	\(
	[\cN_Q,\lambda \dG(h)]=0,
	\)
we obtain
	\begin{equation}\label{eq:BQ-comm-reduction}
		[B_Q,[\mathbb{H}_\lambda^{\mathrm{re}},B_Q]]
		=
		\frac{\lambda^2}{4}\,
		[\cN_Q,[W^{\mathrm{re}},\cN_Q]].
	\end{equation}
	Fix $\varepsilon\in(0,\frac\alpha4)$ and choose
	\[
	s_1:=\varepsilon,
	\qquad
	s_2:=\frac12+\varepsilon,
	\qquad
	\theta:=s_1+s_2=\frac12+2\varepsilon<1.
	\]
	Then Lemma~\ref{lem:com1}, together with \eqref{eq:BQ-comm-reduction}, gives
	\begin{equation}\label{eq:BQ-comm-bound-final}
		\Tr\!\Big([B_Q,[\mathbb{H}_\lambda^{\mathrm{re}},B_Q]]\Gamma_\lambda\Big)
		\le
		C_\varepsilon\lambda^4\Lambda^{2-\theta}
		\Tr\!\big[\mathbb{H}_{s_1}\mathbb{H}_{s_2}\Gamma_\lambda\big].
	\end{equation}
	Now Lemma~\ref{lem:holder}, in the form \eqref{eq:shell-energy-number-moment-local}, yields
	\[
	\lambda^4
	\Tr\!\big[\mathbb{H}_{s_1}\mathbb{H}_{s_2}\Gamma_\lambda\big]
	\le
	C_\varepsilon\lambda^{2-\theta}(1+|\log\lambda|^2).
	\]
	Therefore
	\[
	\Tr\!\Big([B_Q,[\mathbb{H}_\lambda^{\mathrm{re}},B_Q]]\Gamma_\lambda\Big)
	\le
	C_\varepsilon\lambda^{2-\theta}\Lambda^{2-\theta}(1+|\log\lambda|^2).
	\]
	Since \(\theta=\frac12+2\varepsilon\) and \(\varepsilon\) small enough,
	\[
	\lambda^{2-\theta}\Lambda^{2-\theta}
	=
	\lambda^{(2-\theta)(1-\alpha)}
	\le
	\lambda^{3/2-2\alpha}
	=
	\lambda^{3/2}\Lambda^2.
	\]
	Hence
	\[
	\Tr\!\Big([B_Q,[\mathbb{H}_\lambda^{\mathrm{re}},B_Q]]\Gamma_\lambda\Big)
	\le
	C\lambda^{\frac32}\Lambda^2(1+|\log\lambda|^2).
	\]
	Since \eqref{eq:aQlambda-o1} implies $a_{Q,\lambda}\to0$, we have $a_{Q,\lambda}e^{a_{Q,\lambda}}\le
	2a_{Q,\lambda}$ for small $\lambda$. Inserting this and the last estimate into
	\eqref{eq:DNN-BQ} yields \eqref{eq:high-frequency-fluctuation-bound}.
\end{proof}


\begin{proof}[Proof of Theorem \ref{cor:HF0-vanishes}]
	Set
	\[
	A_\lambda:=\lambda^2\Tr\!\big[(\cN_Q-N_{0,Q})^2\Gamma_\lambda\big],
	\qquad
	B_\lambda:=\lambda^2\Tr\!\big[(\cN_P-N_{0,P})^2\Gamma_\lambda\big].
	\]
	By Cauchy--Schwarz,
	\[
	|HF_0(\Gamma_\lambda)|
	\le
	CA_\lambda
	+
	C B_\lambda^{1/2}A_\lambda^{1/2}.
	\]
	Using Theorem~\ref{prop:high-frequency-fluctuations} with $\Lambda=\lambda^{-\alpha}$, we obtain
	\[
	A_\lambda
	\le
	C(1+|\log\lambda|^2)\bigl(\Lambda^{-1}+\lambda^{3/2}\Lambda^2\bigr)
	=
	C(1+|\log\lambda|^2)\bigl(\lambda^\alpha+\lambda^{3/2-2\alpha}\bigr).
	\]
	On the other hand, we have
	\begin{equation}\label{eq:low-number-fluctuation-log2}B_\lambda
		\le C(1+|\log\lambda|^2).
	\end{equation}
	In fact, we have
	\[
	(\cN_P-N_{0,P})^2\le 2\cN_P^2+2N_{0,P}^2.
	\]
	For the first term, the operator inequality $\cN_P\le\cN$ gives
	\[
	\lambda^2\Tr_{\mathfrak F(H)}(\cN_P^2\Gamma_{\lambda})
	\le
	\lambda^2\Tr_{\mathfrak F(H)}(\cN^2\Gamma_\lambda)
	\le C(1+|\log\lambda|^2),
	\]
	by Proposition~\ref{prop:number-moments-log}. For the second term, $0\le P\le1$ implies
	\[
	\lambda^2N_{0,P}^2\le(\lambda N_0)^2\le C(1+|\log\lambda|^2).
	\]
	This proves \eqref{eq:low-number-fluctuation-log2}. Therefore
	\[
	|HF_0(\Gamma_\lambda)|
	\le
	C(1+|\log\lambda|^2)\bigl(\lambda^\alpha+\lambda^{3/2-2\alpha}\bigr)
	+
	C(1+|\log\lambda|^2)\bigl(\lambda^{\alpha/2}+\lambda^{3/4-\alpha}\bigr).
	\]
	Since $0<\alpha<\frac34$, one has $\lambda^\alpha\le \lambda^{\alpha/2}$ and
	$\lambda^{3/2-2\alpha}\le \lambda^{3/4-\alpha}$ for $0<\lambda<1$. Hence
\eqref{eq:HF0-vanishes-rate} and the convergence
	\eqref{eq:HF0-vanishes} follow immediately.
\end{proof}

\section{Vanishing of the shell contribution}
\label{sec:shell-vanishing-rewritten}

Recall from Section~\ref{sec:dec} that
\[
\mathcal E_{\rm shell}(\Gamma)
:=
\frac{\lambda^2}{2}
\Tr_{H^{\otimes 2}}
\Big[
\bigl(\Pi_1V_{\neq0}\Pi_1-\Pi V_{\neq0}\Pi\bigr)\Gamma^{(2)}
\Big],
\]
where
\[
P:=\1_{\{h\le \Lambda^2\}},
\qquad
P_1:=\1_{\{\Lambda^2<h\le \Lambda_1^2\}},
\qquad
S:=P+P_1=\1_{\{h\le \Lambda_1^2\}},
\]
and
\[
\Pi:=P^{\otimes 2},
\qquad
\Pi_1:=S^{\otimes 2}.
\]

Write throughout this section
\[
L_\lambda:=1+|\log\lambda|,\qquad \Lambda=\lambda^{-\alpha},\qquad\Lambda_1=\lambda^{-\alpha_1}, \qquad\alpha\in(0,1/4), \quad\alpha_1\in(1/2,3/4).
\]
The main conclusion of this section is Theorem~\ref{prop:shell-vanishing-local} below: the
shell contribution vanishes.

\begin{theorem}[Vanishing of the shell contribution]
	\label{prop:shell-vanishing-local}
	{\rmk{Fix $\frac32<\beta\le2$ and choose $\eta\in(2-\beta,\beta-1)$ and \(\frac34-\alpha_1-\frac\alpha2(2-\beta)>0\). Then}}
	\begin{equation}
		\label{eq:shell-vanishing-rate-local}
		|\mathcal E_{\rm shell}(\Gamma_\lambda)|
		\le
		{\rmk{C_\beta(1+|\log\lambda|)^{5/2}\Lambda^{(2-\beta)/2}\bigl(\Lambda^{-\eta/2}+\lambda^{3/4}\Lambda_1\bigr)
		+
		C_\beta\lambda (1+|\log\lambda|^2)\Lambda_1^{2-\beta}\longrightarrow0.}}
	\end{equation}
\end{theorem}

The rest of Section~\ref{sec:shell-vanishing-rewritten} is devoted to proving this estimate from the
moment bounds. The strategy is parallel to the analysis of the zero-mode remainder in Section~\ref{sec:high0}, but the relevant observables are now the shell-mode fluctuations. We first rewrite $\mathcal E_{\rm shell}$ in terms of the localized operators $A_\ell$ and $B_\ell$, and then reduce the problem to weighted estimates on the shell channels. For each polarized shell observable we apply Theorem~\ref{thm:correlation-intro}, which requires us to bound the double commutators with both the kinetic part and the interaction part of the Hamiltonian. Since $\hvb$ is not summable, the interaction commutator is again written as a quartic operator; the asymmetric quartic Hilbert--Schmidt bound from Lemma~\ref{lem:aH-S}, combined with the a priori kinetic estimates, then yields the required mode-by-mode control.

For the three shell channels we write
\[
\mathcal C_{\rm shell}
:=
\{(P,P_1),(P_1,P),(P_1,P_1)\},
\qquad
\rho_\ell^{A,B}:=\dG(Ae_\ell B),
\qquad
(A,B)\in \mathcal C_{\rm shell}.
\]
It is also convenient to introduce
\begin{equation}
	\label{eq:def-shell-A-B}
	A_\ell:=\dG(Pe_\ell P),
	\qquad
	B_\ell:=\dG(Se_\ell S)-\dG(Pe_\ell P)
	=
	\sum_{(A,B)\in \mathcal C_{\rm shell}}\rho_\ell^{A,B}.
\end{equation}

We first rewrite $\mathcal E_{\rm shell}$ in terms of the operators $A_\ell$ and $B_\ell$.

\begin{lemma}[Quadratic decomposition of the shell part]
	\label{lem:shell-decomposition-local}
	There exists a one-body remainder $R_{\rm shell}^{(1)}(\Gamma)$ such that
	\begin{equation}
		\label{eq:shell-decomposition-local}
		\mathcal E_{\rm shell}(\Gamma)
		=
		\frac{\lambda^2}{4}
		\sum_{\ell\neq 0}\hvb(\ell)\,
		\Tr\Big[\Big(
		A_\ell B_{-\ell}+B_\ell A_{-\ell}+B_\ell B_{-\ell}
		\Big)\Gamma\Big]
		+
		R_{\rm shell}^{(1)}(\Gamma),
	\end{equation}
	where $A_\ell$ and $B_\ell$ are defined in \eqref{eq:def-shell-A-B}.  Moreover,
	\begin{equation}
		\label{eq:shell-remainder-local}
		{\color{red}|R_{\rm shell}^{(1)}(\Gamma_\lambda)|
		\le
		C_\beta\lambda L_\lambda(1+\log\Lambda_1)\Lambda_1^{2-\beta}.}
	\end{equation}
\end{lemma}

\begin{proof}
	We expand the quadratic part first and isolate the one-body remainder afterwards.
	
	Using the Fourier representation
	\[
	V_{\neq0}
	=
	\sum_{\ell\neq 0}\hvb(\ell)
	(M_{e_\ell}\otimes M_{e_{-\ell}})
	\]
	and the bosonic quantization identity
	\[
	\mathbb{W}(X\otimes Y)=\frac12\bigl(\dG(X)\dG(Y)-\dG(XY)\bigr),
	\]
	we obtain
	\[
	\mathbb{W}(\Pi_1V_{\neq0}\Pi_1)
	=
	\frac{1}{2}
	\sum_{\ell\neq0}\hvb(\ell)
	\Big(
	\dG(Se_\ell S)\,\dG(Se_{-\ell}S)
	-
	\dG(Se_\ell Se_{-\ell}S)
	\Big),
	\]
	and the same formula with $S$ replaced by $P$ for $\Pi V_{\neq0}\Pi$.
	
	Since
	\[
	\dG(Se_\ell S)=A_\ell+B_\ell,
	\]
	the quadratic difference is exactly
	\[
	(A_\ell+B_\ell)(A_{-\ell}+B_{-\ell})-A_\ell A_{-\ell}
	=
	A_\ell B_{-\ell}+B_\ell A_{-\ell}+B_\ell B_{-\ell}.
	\]
	This is the quadratic part appearing in \eqref{eq:shell-decomposition-local}. The leftover
	one-body contribution is therefore
	\[
	R_{\rm shell}^{(1)}(\Gamma)
	=
	-\frac{\lambda^2}{4}
	\sum_{\ell\neq0}\hvb(\ell)
	\Tr\!\Big(\dG(Se_\ell Se_{-\ell}S-Pe_\ell Pe_{-\ell}P)\Gamma\Big).
	\]
	
	We then bound the one-body remainder.
	Introduce
	\[
	K_{\rm shell}
	:=
	\sum_{\ell\neq0}\hvb(\ell)
	\bigl(Se_\ell Se_{-\ell}S-Pe_\ell Pe_{-\ell}P\bigr).
	\]
	{\rmk{This operator is diagonal in the Fourier basis, and its symbol is bounded by
	\[
	\sup_p \sum_{q:\,h(q)\leq \Lambda_1^2}\hvb(q-p)
	\le C_\beta(1+\log\Lambda_1)\Lambda_1^{2-\beta}.
	\]
	Hence
	\[
	\|K_{\rm shell}\|\le C_\beta(1+\log\Lambda_1)\Lambda_1^{2-\beta}.
	\]
	It follows that
	\[
	|R_{\rm shell}^{(1)}(\Gamma)|
	\le C_\beta\lambda^2(1+\log\Lambda_1)\Lambda_1^{2-\beta}\Tr(\cN\Gamma).
	\]
	which by Proposition~\ref{prop:number-moments-log-safe} gives
	\eqref{eq:shell-remainder-local}.}}
\end{proof}
\begin{lemma}[Polarization of the shell channels]
	\label{lem:shell-polarization-local}
	For $(A,B)\in \mathcal C_{\rm shell}$ and $\ell\in\mathbb Z^2$, define
	\[
	g_\ell^{A,B}:=Ae_\ell B,
	\]
	and
	\[
	f_{\ell,+}^{A,B}:=g_\ell^{A,B}+(g_\ell^{A,B})^*,
	\qquad
	f_{\ell,-}^{A,B}:=i(g_\ell^{A,B}-(g_\ell^{A,B})^*).
	\]
	Then
	\begin{equation}
		\label{eq:shell-polarization-local}
		(\rho_\ell^{A,B})^*\rho_\ell^{A,B}
		+
		(\rho_{-\ell}^{B,A})^*\rho_{-\ell}^{B,A}
		=
		\frac12
		\Big(
		\dG(f_{\ell,+}^{A,B})^2+\dG(f_{\ell,-}^{A,B})^2
		\Big).
	\end{equation}
	In particular,
	\begin{equation}
		\label{eq:shell-polarization-ineq-local}
		(\rho_\ell^{A,B})^*\rho_\ell^{A,B}
		\le
		\frac12
		\Big(
		\dG(f_{\ell,+}^{A,B})^2+\dG(f_{\ell,-}^{A,B})^2
		\Big).
	\end{equation}
\end{lemma}

\begin{proof}

	Now use the elementary identity
	\[
	X^*X+XX^*=
	\frac12\Big((X+X^*)^2+(i(X-X^*))^2\Big).
	\]
	with $X:=\dG(Ae_\ell B)=\rho_\ell^{A,B},X^*=\dG(Be_{-\ell}A)=\rho_{-\ell}^{B,A}$, we obtain
	\eqref{eq:shell-polarization-local}. Since the second term on the left-hand side of
	\eqref{eq:shell-polarization-local} is nonnegative, \eqref{eq:shell-polarization-ineq-local}
	follows immediately.
\end{proof}

In order to control the shell modes, it therefore suffices to estimate the right-hand side of \eqref{eq:shell-polarization-ineq-local}. This will be done by applying the second-order correlation inequality in Theorem~\ref{thm:correlation-intro} to the polarized observables $\dG(f_{\ell,\pm}^{A,B})$.

\begin{lemma}[Hilbert--Schmidt bounds for the shell observables]
	\label{lem:shell-HS-local}
	Fix $\eta\in(0,1)$. Then for every $(A,B)\in\mathcal C_{\rm shell}$ and every
	$\ell\in\mathbb Z^2$,
	\begin{equation}
		\label{eq:shell-HS-local}
		\bigl\|
		h^{-1/2}f_{\ell,\pm}^{A,B}h^{-1/2}
		\bigr\|_{\mathrm{HS}}
		\le
		C_\eta\,
		\langle\ell\rangle^{-1+\eta+\kappa}\Lambda^{-\eta}\,\mathbf 1_{\{|\ell|\lesssim \Lambda_1\}},
	\end{equation}
	for $\kappa>0$.
\end{lemma}

\begin{proof}
	We explain the argument first for the mixed channel $(A,B)=(P,P_1)$; the other cases are
	obtained from the same computation after a harmless relabelling.
	
	\smallskip
	For $g_\ell^{P,P_1}=Pe_\ell P_1$, we have
	\[
	\bigl\|h^{-1/2}g_\ell^{P,P_1}h^{-1/2}\bigr\|_{\mathrm{HS}}^2
	=
	\sum_{u\in\mathbb Z^2}
	P_{u+\ell}P_{1,u}\,h(u+\ell)^{-1}h(u)^{-1}.
	\]
	Whenever $P_{1,u}\neq0$, the momentum $u$ lies in the shell, hence $h(u)\ge \Lambda^2$.
	Therefore, for every $\eta\in(0,1)$,
	\[
	h(u)^{-1}\le \Lambda^{-2\eta}h(u)^{-1+\eta}.
	\]
	Substituting this into the previous sum gives
	\[
	\bigl\|h^{-1/2}g_\ell^{P,P_1}h^{-1/2}\bigr\|_{\mathrm{HS}}^2
	\le
	\Lambda^{-2\eta}
	\sum_{u}P_{u+\ell}h(u+\ell)^{-1}h(u)^{-1+\eta}
	\le C_\eta\langle\ell\rangle^{-2+2\eta+\kappa}\Lambda^{-2\eta},
	\]
	where we used Lemma \ref{lem:sum} in the last step.
	The case $P_1e_\ell P$ is identical after exchanging the two momentum variables. For
	$P_1e_\ell P_1$, one may extract the same factor $\Lambda^{-2\eta}$ from either of the two
	heat weights, because both sides lie in the shell. Finally,
	$f_{\ell,\pm}^{A,B}$ is the sum of at most two such elementary pieces, so the triangle
	inequality yields \eqref{eq:shell-HS-local} for the polarized observables as well.
\end{proof}

\begin{lemma}[Double commutators with $h$]
	\label{lem:shell-h-comm-local}
	For every $(A,B)\in\mathcal C_{\rm shell}$ and every $\ell\neq0$,
	\begin{equation}
		\label{eq:shell-h-comm-local}
		\bigl\|
		[f_{\ell,\pm}^{A,B},[f_{\ell,\pm}^{A,B},h]]
		\bigr\|
		\le
		C\Lambda_1^2\,\mathbf 1_{\{|\ell|\lesssim \Lambda_1\}}.
	\end{equation}
\end{lemma}

\begin{proof}
	Since $A,B\in\{P,P_1\}$ are spectral projections of $h$, we have
	\[
	[A,h]=[B,h]=0.
	\]
	Hence
	\[
	[g_\ell^{A,B},h]=A[e_\ell,h]B.
	\]
	In the Fourier basis,
	\[
	[g_\ell^{A,B},h]e_p
	=
	\bigl(h(p)-h(p+\ell)\bigr)A_{p+\ell}B_p\,e_{p+\ell}.
	\]
	Therefore $[g_\ell^{A,B},h]e_p\neq0$ can occur only if $A_{p+\ell}B_p\neq0$.
	Since $(A,B)\in\mathcal C_{\rm shell}=\{(P,P_1),(P_1,P),(P_1,P_1)\}$, both $A$ and $B$
	are dominated by
	\[
	S=P+P_1=\mathbf 1_{\{h\le \Lambda_1^2\}}.
	\]
	Thus $A_{p+\ell}B_p\neq0$ implies $S_{p+\ell}S_p=1$, hence
	\[
	|p|\lesssim \Lambda_1,
	\qquad
	|p+\ell|\lesssim \Lambda_1,
	\qquad
	|\ell|\le |p|+|p+\ell|\lesssim \Lambda_1.
	\]
	On this support,
	\[
	|h(p+\ell)-h(p)|
	=
	\bigl||p+\ell|^2-|p|^2\bigr|
	\le 2|p||\ell|+|\ell|^2
	\lesssim \Lambda_1^2.
	\]
	It follows that
	\[
	\|[g_\ell^{A,B},h]\|
	\le C\Lambda_1^2\,\mathbf 1_{\{|\ell|\lesssim \Lambda_1\}}.
	\]
	The same bound holds for $(g_\ell^{A,B})^*=g_{-\ell}^{B,A}$, namely
	\[
	\|[(g_\ell^{A,B})^*,h]\|
	\le C\Lambda_1^2\,\mathbf 1_{\{|\ell|\lesssim \Lambda_1\}}.
	\]
	Moreover,
	\[
	\|g_\ell^{A,B}\|
	\le \|A\|\,\|e_\ell\|\,\|B\|
	\le 1,
	\qquad
	\|(g_\ell^{A,B})^*\|\le 1,
	\]
	so that
	\[
	\|f_{\ell,\pm}^{A,B}\|\le 2
	\]
	and
	\[
	\|[f_{\ell,\pm}^{A,B},h]\|
	\le
	\|[g_\ell^{A,B},h]\|+\|[(g_\ell^{A,B})^*,h]\|
	\le
	C\Lambda_1^2\,\mathbf 1_{\{|\ell|\lesssim \Lambda_1\}}.
	\]
	Finally, using
	\[
	[f_{\ell,\pm}^{A,B},[f_{\ell,\pm}^{A,B},h]]
	=
	f_{\ell,\pm}^{A,B}[f_{\ell,\pm}^{A,B},h]
	-
	[f_{\ell,\pm}^{A,B},h]f_{\ell,\pm}^{A,B},
	\]
	we obtain
	\[
	\bigl\|
	[f_{\ell,\pm}^{A,B},[f_{\ell,\pm}^{A,B},h]]
	\bigr\|
	\le
	2\|f_{\ell,\pm}^{A,B}\|\,\|[f_{\ell,\pm}^{A,B},h]\|
	\le
	C\Lambda_1^2\,\mathbf 1_{\{|\ell|\lesssim \Lambda_1\}}.
	\]
	This proves \eqref{eq:shell-h-comm-local}.
\end{proof}

For the interaction commutator we use the asymmetric quartic Hilbert--Schmidt estimate
proved earlier, together with two elementary convolution bounds from appendix. Recall
	\[
	\tilde c_\lambda := \frac{\lambda^2}{2(2\pi)^2},
	\qquad
	M_{p,q,k} := a^*_{p+k} a^*_{q-k} a_q a_p,
	\qquad
	W_{\neq 0}^{\mathrm{re}}
	=
	\tilde c_\lambda
	\sum_{k\neq 0}\hvb(k)\sum_{p,q\in\mathbb Z^2} M_{p,q,k},
	\]
	and
		for \(s>0\),
	\[
	S_k(s):=\sum_{u\in\mathbb Z^2} h(u+k)^{-s}h(u)^{-s},
	\qquad
	S_k^{(S)}(s):=\sum_{u\in\mathbb Z^2} S_u\,h(u+k)^{-s}h(u)^{-s}.
	\]

\begin{lemma}[Double commutator with $W_{\neq0}^{\rm re}$ for the shell modes]
	\label{lem:shell-W-comm-local}
	Fix
	\(
	0<s_1\le s_2<1,
	s_2>\frac12,
	s_1+s_2<1.
	\)
	Then for every $(A,B)\in\mathcal C_{\rm shell}$ and every $\ell\neq0$,
	\begin{equation}
		\label{eq:shell-W-comm-local}
		\pm
		[\dG(f_{\ell,\pm}^{A,B}),
		[\dG(f_{\ell,\pm}^{A,B}),W_{\neq0}^{\rm re}]]
		\le
		C\lambda^2\Lambda_1^{2-s_1-s_2}\langle\ell\rangle^{-s_1}
		\mathbb{H}_{s_1}\mathbb{H}_{s_2}\,\mathbf 1_{\{|\ell|\lesssim \Lambda_1\}}
	\end{equation}
	as quadratic forms on $\mathfrak F_{\rm fin}$.
\end{lemma}

\begin{proof}
	
	For $\sigma\in\{\pm1\}$ define
	\[
	X_{+1}:=g_\ell^{A,B}=Ae_\ell B,
	\qquad
	X_{-1}:=(g_\ell^{A,B})^*.
	\]
	Hence
	\[
	[\dG(f_{\ell,\pm}^{A,B}),[\dG(f_{\ell,\pm}^{A,B}),W_{\neq0}^{\mathrm{re}}]]
	\]
	is a linear combination, with coefficients of modulus at most $1$, of the four operators
	\[
	C_{\sigma,\tau}
	:=
	[\dG(X_\sigma),[\dG(X_\tau),W_{\neq0}^{\mathrm{re}}]],
	\qquad
	\sigma,\tau\in\{\pm1\}.
	\]
	It is therefore enough to prove a uniform bound for $C_{\sigma,\tau}$.
	
	Write the Fourier multipliers of $A,B,S=P+P_1$ again as $A_p,B_p,S_p\in\{0,1\}$.
	A direct computation in the Fourier basis gives, for every $r\in\mathbb Z^2$,
	\[
	[\dG(X_\sigma),a_r^*]
	=
	\alpha_r^{(\sigma)} a_{r+\sigma\ell}^*,
	\qquad
	[\dG(X_\sigma),a_r]
	=
	-{\beta}_r^{(\sigma)} a_{r-\sigma\ell},
	\]
	where
	\[
	\alpha_r^{(+1)} = A_{r+\ell}B_r,
	\qquad
	{\beta}_r^{(+1)} = A_r B_{r-\ell},
	\]
	and
	\[
	\alpha_r^{(-1)} = B_{r-\ell}A_r,
	\qquad
	{\beta}_r^{(-1)} = B_r A_{r+\ell}.
	\]
	In particular,
	\[
	0\le \alpha_r^{(\sigma)} \le S_r S_{r+\sigma\ell},
	\qquad
	0\le {\beta}_r^{(\sigma)} \le S_r S_{r-\sigma\ell}.
	\]
	Since $A,B\le S$, all these coefficients vanish unless the shift $\ell$ connects two points of
	$\text{supp} S$, hence every term below is automatically zero unless $|\ell|\lesssim \Lambda_1$.
	
	Applying the derivation rule to $M_{p,q,k}$ yields
	\begin{align*}
		[\dG(X_\sigma),M_{p,q,k}]
		&=
		\alpha_{p+k}^{(\sigma)}
		\,a^*_{p+k+\sigma\ell}a^*_{q-k}a_qa_p
		+
		\alpha_{q-k}^{(\sigma)}
		\,a^*_{p+k}a^*_{q-k+\sigma\ell}a_qa_p
		\\
		&\quad
		-
		{\beta}_q^{(\sigma)}
		\,a^*_{p+k}a^*_{q-k}a_{q-\sigma\ell}a_p
		-
		{\beta}_p^{(\sigma)}
		\,a^*_{p+k}a^*_{q-k}a_q a_{p-\sigma\ell}.
	\end{align*}
	The four operator slots naturally split into the $p$-pair
	\[
	\{a^*_{p+k},a_p\}
	\]
	and the $q$-pair
	\[
	\{a^*_{q-k},a_q\}.
	\]
	After commuting once more with $\dG(X_\tau)$ one obtains exactly $4\times 4=16$ quartic
	monomials for each ordered pair $(\sigma,\tau)$.
	
	These $16$ terms split into two families.
	
	\medskip
	
	\noindent
	\emph{Case 1: one hit in the $p$-pair and one hit in the $q$-pair.}
	There are $8$ such terms. A typical example is
	\[
	Q_\gamma
	=
	\tilde c_\lambda
	\sum_{k\neq0}\hvb(k)
	\sum_{p,q}
	\alpha_{p+k}^{(\sigma)}\alpha_{q-k}^{(\tau)}
	\,a^*_{p+k+\sigma\ell}a^*_{q-k+\tau\ell}a_q a_p,
	\]
	and the other $7$ terms are obtained by replacing one or both creation hits by annihilation hits.
	In every such term the shell coefficient factorizes into one factor depending only on the $p$-variable
	and one factor depending only on the $q$-variable. By Lemma~\ref{lem:aH-S}, to obtain the desired quadratic-form bound for the associated quartic operator, it suffices to control the $\ell^2$-norm of the asymmetrically rescaled kernel $\widetilde \gamma$. After a harmless relabelling of the free variables,
	the kernel of any such term satisfies
	\[
	\|\widetilde\gamma\|_{\ell^2}^2
	\lesssim
	\lambda^4
	\sum_{k\neq0}
	|\hvb(k)|^2
	\,S^{(S)}_{k+\nu_1\ell}(s_1)\,
	S^{(S)}_{\nu_2\ell-k}(s_2),
	\]
	for some $\nu_1,\nu_2\in\{\pm1\}$, where $\widetilde\gamma=\widetilde\gamma^{(s_1,s_2)}$ is the
	asymmetrically rescaled kernel from Lemma~\ref{lem:aH-S}.
	
	Indeed, the two Kronecker constraints coming from the quartic monomial determine $k$ uniquely, so
	the $k$-slices are orthogonal in $\ell^2((\mathbb Z^2)^4)$. Using
	Lemma~\ref{lem:convolution-kernels},
	we get
	\[
	\|\widetilde\gamma\|_{\ell^2}^2
	\lesssim
	\lambda^4 \Lambda_1^{4-2s_1-2s_2}
	\sum_{k\neq0}
	|\hvb(k)|^2
	\langle k+\nu_1\ell\rangle^{-2s_1}
	\langle \nu_2\ell-k\rangle^{-2s_2},
	\]
	which by Lemma~\ref{lem:shifted-Yukawa-sums} yields
	\[
	\|\widetilde\gamma\|_{\ell^2}
	\lesssim
	\lambda^2 \Lambda_1^{2-s_1-s_2}\langle\ell\rangle^{-s_1}\1_{\{|\ell|\lesssim\Lambda_1\}}.
	\]
	
	\medskip
	
	\noindent
	\emph{Case 2: both hit in the same pair.}
	There are again $8$ such terms: $4$ with both hits in the $p$-pair and $4$ with both hits in the
	$q$-pair. A typical example with both hits in the $p$-pair is
	\[
	Q_\gamma
	=
	\tilde c_\lambda
	\sum_{k\neq0}\hvb(k)
	\sum_{p,q}
	\alpha_{p+k}^{(\sigma)}{\beta}_p^{(\tau)}
	\,a^*_{p+k+\sigma\ell}a^*_{q-k}a_q a_{p-\tau\ell},
	\]
	and another typical example is
	\[
	Q_\gamma
	=
	\tilde c_\lambda
	\sum_{k\neq0}\hvb(k)
	\sum_{p,q}
	\alpha_{p+k}^{(\sigma)}\alpha_{p+k+\sigma\ell}^{(\tau)}
	\,a^*_{p+k+(\sigma+\tau)\ell}a^*_{q-k}a_q a_p.
	\]
	In all such terms the product of the two shell coefficients depends on only one free variable
	($p$ or $q$). Since the coefficients are projectors,
	\[
	\alpha_r^{(\sigma)}\alpha_{r+\sigma\ell}^{(\tau)},
	\quad
	\alpha_r^{(\sigma)}{\beta}_{r-k}^{(\tau)},
	\quad
	{\beta}_r^{(\sigma)}{\beta}_{r-\sigma\ell}^{(\tau)},
	\ \ldots
	\]
	are bounded by a single shell projector after reindexing, and the term is nonzero only if
	$|\ell|\lesssim\Lambda_1$. Consequently, after relabelling,
	\[
	\|\widetilde\gamma\|_{\ell^2}^2
	\lesssim
	\lambda^4 \1_{\{|\ell|\lesssim\Lambda_1\}}
	\sum_{k\neq0}
	|\hvb(k)|^2
	\,S^{(S)}_{k+\nu\ell}(s_1)\,S_k(s_2),
	\]
	for some $\nu\in\{0,\pm2\}$.
	Using Lemma~\ref{lem:convolution-kernels},
	we obtain
	\[
	\|\widetilde\gamma\|_{\ell^2}
	\lesssim
	\lambda^2 \Lambda_1^{1-s_1} \1_{\{|\ell|\lesssim\Lambda_1\}}.
	\]
	Since $|\ell|\lesssim \Lambda_1$ on the support of this term, we also have
	\[
	\1_{\{|\ell|\lesssim\Lambda_1\}}
	\le
	C\Lambda_1^{s_1}\langle\ell\rangle^{-s_1},
	\]
	and therefore
	\[
	\|\widetilde\gamma\|_{\ell^2}
	\lesssim
	\lambda^2 \Lambda_1\langle\ell\rangle^{-s_1}\1_{\{|\ell|\lesssim\Lambda_1\}}.
	\]
	Because $s_1+s_2<1$, one has $\Lambda_1\le \Lambda_1^{2-s_1-s_2}$, so the Case~2
	contribution is bounded by the same right-hand side as in Case~1.

	We now apply Lemma~\ref{lem:aH-S} and obtain for every $\psi\in\mathfrak F_{\mathrm{fin}}$,
	\[
	\big|
	\langle \psi,
	[\dG(f_{\ell,\pm}^{A,B}),
	[\dG(f_{\ell,\pm}^{A,B}),W_{\neq0}^{\mathrm{re}}]]
	\psi\rangle
	\big|
	\lesssim
	\lambda^2
	\Lambda_1^{2-s_1-s_2}\langle\ell\rangle^{-s_1}
	\,
	\langle \psi,\mathbb{H}_{s_1}\mathbb{H}_{s_2}\psi\rangle\1_{\{|\ell|\lesssim\Lambda_1\}},
	\]
	which implies \eqref{eq:shell-W-comm-local}.
\end{proof}

\begin{proposition}[Mode-by-mode shell bound]
	\label{prop:shell-mode-local}

	It holds that for every $(A,B)\in\mathcal C_{\rm shell}$, every $\ell\neq0$, and every
	{\color{red}$0<s_1\le s_2<1$ with $s_2>\frac12$ and $s_1+s_2<1$,}
	\begin{equation}
		\label{eq:shell-mode-local}
		\lambda^2
		\Tr\!\Big(
		(\rho_\ell^{A,B})^*\rho_\ell^{A,B}\Gamma_\lambda
		\Big)
		\le
		C\mathbf 1_{\{|\ell|\lesssim \Lambda_1\}}\Big(
		\langle\ell\rangle^{-1+\eta+\kappa}\Lambda^{-\eta}L_\lambda^2
		+
		\lambda^2\Lambda_1^2L_\lambda^2
		+
		{\color{red}\lambda^{2-s_1-s_2}\Lambda_1^{2-s_1-s_2}\langle\ell\rangle^{-s_1}L_\lambda^2}
		\Big).
	\end{equation}
{\color{red}Here $\eta\in(2-\beta,\beta-1-\kappa)$ is fixed and $\kappa>0$.}
\end{proposition}

\begin{proof}
	Fix one of the two polarized observables and write
	\[
	f:=f_{\ell,\pm}^{A,B}.
	\]
	For $\ell\neq 0$, the free state preserves momentum, hence
	\[
	\langle \dG(f)\rangle_{\Gamma_0}=0.
	\]

Set
	\[
	B:=\frac{\lambda}{2}\dG(f),
	\qquad
	\Gamma_{\lambda,t}
	:=
	\frac{e^{-\mathbb{H}_\lambda^{\rm re}+tB}}{\Tr(e^{-\mathbb{H}_\lambda^{\rm re}+tB})},
	\qquad
	\Gamma_{0,t}
	:=
	\frac{e^{-\lambda \dG(h)+tB}}{\Tr(e^{-\lambda \dG(h)+tB})},
	\qquad |t|\le 1.
	\]

	Using \eqref{eq:one-body-HS-safe} together with \cite[Lemma 6.3]{LewNamRou-21}
	we obtain
	\begin{align*}
		\bigl|\Tr(B\Gamma_{\lambda,t})\bigr|
		&=
		\frac{\lambda}{2}
		\bigl|
		\Tr\bigl(f(\gamma_{\lambda,t}-\gamma_0)\bigr)
		\bigr|
		\\
		&\le
		\frac{\lambda}{2}
		\bigl|
		\Tr\bigl(f(\gamma_{\lambda,t}-\gamma_{0,t})\bigr)
		\bigr|
		+
		\frac{\lambda}{2}
		\bigl|
		\Tr\bigl(f(\gamma_{0,t}-\gamma_0)\bigr)
		\bigr|
		\\
		&\le
		C\lambda
		\bigl\|
		\sqrt h\,(\gamma_{\lambda,t}-\gamma_{0,t})\,\sqrt h
		\bigr\|_{\mathrm{HS}}
		\bigl\|h^{-1/2}fh^{-1/2}\bigr\|_{\mathrm{HS}}
		+
		Ca_f^2
		\\
		&\le
		Ca_fL_\lambda^2+Ca_f^2.
	\end{align*}
Here by \eqref{eq:shell-HS-local},
	\[
	a_f:=\bigl\|h^{-1/2}fh^{-1/2}\bigr\|_{\mathrm{HS}}
	\le
	C_\eta\langle\ell\rangle^{-1+\eta+\kappa}\Lambda^{-\eta}
	\le
	C_\eta\Lambda^{-\eta}.
	\]

Since $\Lambda=\lambda^{-\alpha}$, we have $a_fL_\lambda^2\ll 1$. Hence,   Theorem~\ref{thm:correlation-intro} gives
	\begin{equation}
		\label{eq:shell-local-DNN}
		\Tr(B^2\Gamma_\lambda)
		\le
		Ca_fL_\lambda^2
		+
		\frac14
		\Tr\!\bigl([B,[\mathbb{H}_\lambda^{\rm re},B]]\Gamma_\lambda\bigr).
	\end{equation}
	Because $\dG(f)$ commutes with $\cN$, the centered zero-mode part of $W^{\rm re}$ drops
	out and
	\[
	[B,[\mathbb{H}_\lambda^{\rm re},B]]
	=
	[B,[\lambda \dG(h),B]]
	+
	[B,[W_{\neq0}^{\rm re},B]].
	\]

	For the kinetic part,
	\[
	[B,[\lambda \dG(h),B]]
	=
	\frac{\lambda^3}{4}\,\dG([f,[f,h]]).
	\]
	Hence, by \eqref{eq:shell-h-comm-local} and Proposition~\ref{prop:number-moments-log-safe},
	\[
	\Bigl|
	\Tr\!\bigl([B,[\lambda \dG(h),B]]\Gamma_\lambda\bigr)
	\Bigr|
	\le
	C\lambda^3\Lambda_1^2\mathbf 1_{\{|\ell|\lesssim \Lambda_1\}}\Tr(\cN\Gamma_\lambda)
	\le
	C\lambda^2\Lambda_1^2\mathbf 1_{\{|\ell|\lesssim \Lambda_1\}}L_\lambda.
	\]

	For the interaction part,  \eqref{eq:shell-W-comm-local} and
	\eqref{eq:shell-energy-number-moment-local} imply
	\[
	\Bigl|
	\Tr\!\bigl([B,[W_{\neq0}^{\rm re},B]]\Gamma_\lambda\bigr)
	\Bigr|
	\le
	C\lambda^4 \Lambda_1^{2-s_1-s_2}\langle\ell\rangle^{-s_1}
	\Tr\!\Big(
	\mathbb{H}_{s_1}\mathbb{H}_{s_2}\Gamma_\lambda
	\Big)
	\le
	C\,\Lambda_1^{2-s_1-s_2}\langle\ell\rangle^{-s_1}\lambda^{2-s_1-s_2}L_\lambda^2.
	\]
	Substituting these bounds into \eqref{eq:shell-local-DNN}, we obtain
	\[
	\lambda^2
	\Tr\!\Big(
	(\rho_\ell^{A,B})^*\rho_\ell^{A,B}\Gamma_\lambda
	\Big)
	\le
	C\mathbf 1_{\{|\ell|\lesssim \Lambda_1\}}\Big(
	a_fL_\lambda^2
	+
	\lambda^2\Lambda_1^2L_\lambda
	+
	\lambda^{2-s_1-s_2}\Lambda_1^{2-s_1-s_2}\langle\ell\rangle^{-s_1}L_\lambda^2
	\Big).
	\]
	{\color{red}Since $L_\lambda\le L_\lambda^2$, this proves \eqref{eq:shell-mode-local}.}
\end{proof}

\begin{corollary}[Weighted shell summability]
	\label{cor:shell-summability-local}
	{\color{red}Let $\eta\in(2-\beta,\beta-1)$. Choose $0<s_1\le s_2<1$ with $s_2>\frac12$ and $s_1+s_2<1$, and assume in addition that $s_1$ is sufficiently small while $s_2$ is sufficiently close to $\frac12$.}
	Then for every $(A,B)\in\mathcal C_{\rm shell}$,
	\begin{equation}
		\label{eq:shell-summability-local}
		\sum_{\ell\neq 0}\hvb(\ell)\,
		\lambda^2
		\Tr\!\Big(
		(\rho_\ell^{A,B})^*\rho_\ell^{A,B}\Gamma_\lambda
		\Big)
		\le
		C(\Lambda^{-\eta}L_\lambda^2
		+
		\lambda^{3/2}\Lambda_1^2L_\lambda^2)
		\longrightarrow 0.
	\end{equation}
	Consequently,
	\begin{equation}
		\label{eq:B-shell-summability-local}
		\sum_{\ell\neq 0}\hvb(\ell)\,
		\lambda^2
		\Tr(B_\ell^*B_\ell\Gamma_\lambda)
		+
		\sum_{\ell\neq 0}\hvb(\ell)\,
		\lambda^2
		\Tr(B_\ell B_\ell^*\Gamma_\lambda)
		\le
		C(\Lambda^{-\eta}L_\lambda^2
		+
		\lambda^{3/2}\Lambda_1^2L_\lambda^2)
		\longrightarrow 0.
	\end{equation}
\end{corollary}

\begin{proof}
	{\color{red}Summing \eqref{eq:shell-mode-local} against the weights $\hvb(\ell)=\langle\ell\rangle^{-\beta}$, we again treat the three contributions separately.

	For the first series,
	\[
	\sum_{\ell\neq0}\langle \ell\rangle^{-\beta-1+\eta+\kappa}\mathbf 1_{\{|\ell|\lesssim \Lambda_1\}}
	\]
	is uniformly bounded because $\beta+1-\eta-\kappa>2$ when $\beta>\frac32$ and $\eta<\beta-1-\kappa$ for $\kappa>0$ small enough. Hence the first contribution is $C_{\beta,\eta}\Lambda^{-\eta}L_\lambda^2$.

	For the second term we use
	\[
	\sum_{|\ell|\lesssim \Lambda_1}\hvb(\ell)\lesssim_\beta (1+\log\Lambda_1)\Lambda_1^{2-\beta},
	\]
	whence
	\[
	\lambda^2\Lambda_1^2L_\lambda^2\sum_{|\ell|\lesssim \Lambda_1}\hvb(\ell)
	\le
	C_\beta (1+\log\Lambda_1)\lambda^2\Lambda_1^{4-\beta}L_\lambda^2
	\le
	C_\beta \lambda^{3/2}\Lambda_1^2L_\lambda^2,
	\]
	because $\beta>\frac32$ and $\alpha_1>\frac12$.

	For the third series we have
	\[
	\sum_{|\ell|\lesssim \Lambda_1}\hvb(\ell)\langle \ell\rangle^{-s_1}
	=
	\sum_{|\ell|\lesssim \Lambda_1}\langle \ell\rangle^{-\beta-s_1}
	\le
	C_{\beta,s_1}
	\begin{cases}
	(1+\log\Lambda_1)\Lambda_1^{2-\beta-s_1},& \beta+s_1\le2,\\[1mm]
	1,& \beta+s_1>2.
	\end{cases}
	\]
	Therefore the third contribution is bounded by
	\[
	C_{\beta,s_1,s_2}(1+\log\Lambda_1)\lambda^{2-s_1-s_2}
	\Lambda_1^{4-\beta-2s_1-s_2}L_\lambda^2
	\]
	in the case $\beta+s_1\le2$, and by an even smaller quantity when $\beta+s_1>2$. Now choose $s_1>0$ sufficiently small and $s_2>\frac12$ sufficiently close to $\frac12$ so that
	\[
	2-s_1-s_2-\alpha_1(4-\beta-2s_1-s_2)>
	\frac32-2\alpha_1.
	\]
	Then the third contribution is also absorbed by $C_\beta\lambda^{3/2}\Lambda_1^2L_\lambda^2$. This proves \eqref{eq:shell-summability-local}.}
By \eqref{eq:def-shell-A-B},
	\[
	B_\ell
	=
	\sum_{(A,B)\in\mathcal C_{\rm shell}}\rho_\ell^{A,B},
	\]
	so
	\[
	B_\ell^*B_\ell
	\le
	3\sum_{(A,B)\in\mathcal C_{\rm shell}}
	(\rho_\ell^{A,B})^*\rho_\ell^{A,B}.
	\]
	The same bound for $B_\ell B_\ell^*$ follows after taking adjoints and reindexing
	$\ell\mapsto -\ell$.  Since $\mathcal C_{\rm shell}$ is finite,
	\eqref{eq:B-shell-summability-local} follows from \eqref{eq:shell-summability-local}.
\end{proof}


\begin{proof}[Proof of Theorem \ref{prop:shell-vanishing-local}]
	
	By Lemma~\ref{lem:shell-decomposition-local},
	\[
	\mathcal E_{\rm shell}(\Gamma_\lambda)
	=
	\frac{\lambda^2}{4}
	\sum_{\ell\neq 0}\hvb(\ell)\,
	\Tr_{\mathfrak F}\Big(
	A_\ell B_{-\ell}+B_\ell A_{-\ell}+B_\ell B_{-\ell}
	\Big)\Gamma_\lambda
	+
	R_{\rm shell}^{(1)}(\Gamma_\lambda).
	\]
	We estimate the three quadratic terms separately. 	Using \[
	A_\ell A_\ell^*\le \cN^2 \1_{|\ell|\lesssim \Lambda},
	\] and Proposition~\ref{prop:number-moments-log-safe} we obtain
	\begin{equation}
		\label{eq:shell-lowlow-local}
		\sum_{\ell\neq 0}\hvb(\ell)\,
		\lambda^2\Tr(A_\ell A_\ell^*\Gamma_\lambda)
		\le
		{\rmk{C_\beta(1+\log\Lambda)\Lambda^{2-\beta}}}\,\lambda^2\Tr(\cN^2\Gamma_\lambda)
		\le
		{\rmk{C_\beta(1+\log\Lambda)\Lambda^{2-\beta}L_\lambda^2.}}
	\end{equation}
	
For the cross term $A_\ell B_{-\ell}$,
	by Cauchy--Schwarz in the state $\Gamma_\lambda$,
	\[
	\sum_{\ell\neq0}\hvb(\ell)\lambda^2
	\bigl|
	\Tr(A_\ell B_{-\ell}\Gamma_\lambda)
	\bigr|
	\le
	\Bigg(
	\sum_{\ell\neq0}\hvb(\ell)\lambda^2
	\Tr(A_\ell A_\ell^*\Gamma_\lambda)
	\Bigg)^{1/2}
	\Bigg(
	\sum_{\ell\neq0}\hvb(\ell)\lambda^2
	\Tr(B_{-\ell}^*B_{-\ell}\Gamma_\lambda)
	\Bigg)^{1/2}.
	\]
	By \eqref{eq:shell-lowlow-local} and \eqref{eq:B-shell-summability-local}, the right-hand
	side is bounded by
	\begin{equation}
		\label{eq:shell-cross-rate-1-local}
		{\rmk{C_\beta\Bigl((1+\log\Lambda)\Lambda^{2-\beta}L_\lambda^2\,(\Lambda^{-\eta}L_\lambda^2
		+
		\lambda^{3/2}\Lambda_1^2L_\lambda^2)\Bigr)^{1/2}.}}
	\end{equation}
	For the cross term $B_\ell A_{-\ell}$,
	the same argument gives the same bound.

For the shell--shell term $B_\ell B_{-\ell}$,
	again by Cauchy--Schwarz,
	\[
	\sum_{\ell\neq0}\hvb(\ell)\lambda^2
	\bigl|
	\Tr(B_\ell B_{-\ell}\Gamma_\lambda)
	\bigr|
	\le
	\Bigg(
	\sum_{\ell\neq0}\hvb(\ell)\lambda^2
	\Tr(B_\ell B_\ell^*\Gamma_\lambda)
	\Bigg)^{1/2}
	\Bigg(
	\sum_{\ell\neq0}\hvb(\ell)\lambda^2
	\Tr(B_{-\ell}^*B_{-\ell}\Gamma_\lambda)
	\Bigg)^{1/2},
	\]
	and the right-hand side is bounded by
	\begin{equation}
		\label{eq:shell-shell-rate-local}
		C(\Lambda^{-\eta}L_\lambda^2
		+
		\lambda^{3/2}\Lambda_1^2L_\lambda^2).
	\end{equation}
	Thus the result follows
	by \eqref{eq:shell-remainder-local}.
	
\end{proof}

\section{Vanishing of the tail contribution}
\label{sec:tail-vanishing-checked}

Recall the second-cutoff notation
\[
P:=\1_{\{h\le \Lambda^2\}},
\qquad
P_1:=\1_{\{\Lambda^2<h\le \Lambda_1^2\}},
\qquad
S:=P+P_1=\1_{\{h\le \Lambda_1^2\}},
\qquad
Q_1:=1-S,
\]
with
\[
\Pi_1:=S^{\otimes 2},
\qquad
R_1:=1-\Pi_1.
\]
Throughout this section we assume
\[
\Lambda_1=\lambda^{-\alpha_1},
\qquad
\frac12<\alpha_1<\frac34,
\qquad
{\rmk{\frac32<\beta\le2,\qquad \eta\in(2-\beta,\beta-1).}}
\]

{\rmk{We also set
\[
\mathfrak r^{\rm tail}_{\beta,\lambda}
:=(1+|\log\lambda|)^3\Big(
\lambda^{\frac34-\alpha_1(3-\beta)}
+
\lambda^{\frac38-\frac{\alpha_1}{2}(3-\beta)-\frac\alpha2(2-\beta)}
+
\lambda^{\frac38-\frac{\alpha_1}{2}(3-\beta)+\frac{\alpha\eta}{2}}
+
\lambda^{\frac98-\frac{\alpha_1}{2}(5-\beta)}
\Big).
\]
For every fixed $\beta>\frac32$, one can choose $\alpha\in(0,1/4)$ sufficiently small and then choose $\alpha_1>\frac12$ sufficiently close to $\frac12$ so that $\mathfrak r^{\rm tail}_{\beta,\lambda}\to0$.}}

The goal of this section is to prove the following theorem, which will be used
directly in Theorem~\ref{thm:lower-bound-classical-final}. The conclusion is the quantitative lower bound
and in particular $\mathcal E_{\rm tail}(\Gamma_\lambda)\ge -o(1)$ as $\lambda\downarrow0$.

\begin{theorem}[Tail contribution in the lower bound]
	\label{cor:tail-lower-bound}
	{\rmk{The tail remainder satisfies}}
	\begin{equation}
		\label{eq:tail-lower-rate}
		\mathcal E_{\rm tail}(\Gamma_\lambda)\ge -{\rmk{C_\beta\,\mathfrak r^{\rm tail}_{\beta,\lambda}}}.
	\end{equation}
\end{theorem}

The main aim of this section is devoted to establishing \eqref{eq:tail-lower-rate}.

We first give the preliminary $Q_1$-tail inputs needed for the
explicit scalar correction in Proposition~\ref{prop:second-cutoff-decomp}. The final outcome is the
quantitative estimate in Lemma~\ref{lem:R1-correction}, stated in the exact form used later in
the free energy lower bound.

\begin{lemma}[Free high-frequency tail]
	\label{lem:free-tail-Q1}
	There exist constants $c,C>0$ such that
	\[
	\lambda N_{0,Q_1}
	\le
	C\log\frac{1}{1-e^{-c\lambda\Lambda_1^2}}\leq C e^{-c\lambda^{1-2\alpha_1}}, \qquad\alpha_1>1/2.
	\]
	In particular,
	for all sufficiently small $\lambda$, and therefore for every $m\ge 0$,
	\[
	(1+|\log\lambda|)^m\,\lambda N_{0,Q_1}\longrightarrow 0.
	\]
\end{lemma}
We put the proof of this result in Appendix \ref{app:exchange-bound}.

\begin{lemma}[High-frequency number fluctuation at the second cutoff]
	\label{lem:Q1-number-fluctuation}
	It holds that
	\[
	\lambda^2
	\Tr_{\mathfrak F(H)}\!\Big[
	(\cN_{Q_1}-N_{0,Q_1})^2\Gamma_\lambda
	\Big]
	\le
	C\Big[
	(1+|\log\lambda|^2)\Lambda_1^{-1}
	+
	\lambda^{3/2}\Lambda_1^2(1+|\log\lambda|^2)
	\Big].
	\]
\end{lemma}

\begin{proof}
	This is the $\Lambda_1$-scale version of Proposition~\ref{prop:high-frequency-fluctuations} above.
\end{proof}

\begin{lemma}[Quantitative bound on the explicit $R_1$-correction]
	\label{lem:Q1-tail-quantitative}
	\label{lem:R1-correction}
	It holds that,
	\[
	\lambda^2\Tr_{H^{\otimes_s2}}\!\big(Q_1^{\otimes2}\Gamma_\lambda^{(2)}\big)
	\le \lambda^2\Tr_{\mathfrak F(H)}(\cN_{Q_1}^2\Gamma_\lambda) \leq
	C
	\lambda^{3/2-2\alpha_1}(1+|\log\lambda|^2),
	\]
	\begin{equation}
		\label{eq:R1-correction-rate-special}
		\frac{\lambda^2}{2(2\pi)^2}
		\Tr_{H^{\otimes_s2}}(R_1\Gamma_\lambda^{(2)})
		\le
		C
		\lambda^{3/4-\alpha_1}(1+|\log\lambda|^2).
	\end{equation}
\end{lemma}

\begin{proof}
	We first have
	\[
	\lambda^2\Tr\!\big(Q_1^{\otimes2}\Gamma_\lambda^{(2)}\big)
	\le
	\lambda^2\Tr(\cN_{Q_1}^2\Gamma_\lambda).
	\]
	Using
	\[
	\cN_{Q_1}^2
	\le
	2(\cN_{Q_1}-N_{0,Q_1})^2+2N_{0,Q_1}^2,
	\]
	we obtain
	\[
	\lambda^2\Tr(\cN_{Q_1}^2\Gamma_\lambda)
	\le
	2\lambda^2\Tr_{\mathfrak F(H)}\!\big((\cN_{Q_1}-N_{0,Q_1})^2\Gamma_\lambda\big)
	+
	2(\lambda N_{0,Q_1})^2.
	\]
	Applying Lemma~\ref{lem:Q1-number-fluctuation} and using $\alpha_1>1/2$ imply the first result.
	
	Since
	\[
	R_1=1-S^{\otimes2}\le Q_1\otimes 1+1\otimes Q_1,
	\]
	we get
	\[
	\Tr(R_1\Gamma_\lambda^{(2)})
	\le
	2\Tr(\cN_{Q_1}\cN\Gamma_\lambda).
	\]
	Hence, by Cauchy--Schwarz,
	\[
	\lambda^2\Tr(R_1\Gamma_\lambda^{(2)})
	\le
	2
	\Bigl(\lambda^2\Tr(\cN_{Q_1}^2\Gamma_\lambda)\Bigr)^{1/2}
	\Bigl(\lambda^2\Tr(\cN^2\Gamma_\lambda)\Bigr)^{1/2}.
	\]
	By Proposition~\ref{prop:number-moments-log} and the first result we obtain
	\begin{equation*}
		\lambda^2\Tr_{H^{\otimes_s2}}(R_1\Gamma_\lambda^{(2)})
		\lesssim
		(1+|\log\lambda|^2)\lambda^{3/4-{\alpha_1}}
		+
		(1+|\log\lambda|)\lambda N_{0,Q_1},
	\end{equation*}
	which implies the result by Lemma \ref{lem:free-tail-Q1}.
\end{proof}

 By Lemma~\ref{lem:R1-correction}, it remains to control
\[
\frac{\lambda^2}{2}
\Tr_{H^{\otimes 2}}
\Bigl((\Pi_1VR_1+R_1V\Pi_1)\Gamma_\lambda^{(2)}\Bigr).
\]
We begin by rewriting this mixed term in a form adapted to the lower-bound argument. Let
\[
A_k:=\dG(Se_kS),
\qquad
B_k:=\dG(Se_kQ_1),
\qquad
C_k:=\dG(Q_1e_kS).
\]
Then
\[
C_k=B_{-k}^*,
\qquad
A_k^*=A_{-k}.
\]

\begin{lemma}[Exact decomposition of the mixed off-diagonal part]
	\label{lem:exact-mixed-decomp}
	For the full Bessel operator
	\[
	V=\sum_{k\in\mathbb Z^2}\hvb(k)(e_k\otimes e_{-k}),
	\]
	one has
	\begin{equation}
		\frac{\lambda^2}{2}\mathbb W(\Pi_1VR_1+R_1V\Pi_1)
		=
		\frac{\lambda^2}{4}
		\sum_{k\in\mathbb Z^2}\hvb(k)
		\Big(A_kB_{-k}+A_{-k}B_k+B_kB_{-k}+C_kA_{-k}+C_{-k}A_k+C_kC_{-k}\Big).
		\label{eq:exact-mixed-decomp}
	\end{equation}
 Since $SQ_1=Q_1S=0$, no one-body correction term appears.
\end{lemma}

\begin{proof}
	We write
	\[
	R_1=S\otimes Q_1+Q_1\otimes S+Q_1\otimes Q_1.
	\]
	Hence
	\begin{align*}
		\Pi_1VR_1+R_1V\Pi_1
		=&(S\otimes S)V(S\otimes Q_1)+(S\otimes S)V(Q_1\otimes S)+(S\otimes S)V(Q_1\otimes Q_1)
		\\
		&\quad +(Q_1\otimes S)V(S\otimes S)+(S\otimes Q_1)V(S\otimes S)+(Q_1\otimes Q_1)V(S\otimes S)
		\\=&\sum_k\hvb(k)
		\Big(Se_kS\otimes Se_{-k}Q_1+Se_kQ_1\otimes Se_{-k}S+Se_kQ_1\otimes Se_{-k}Q_1
		\\
		&+Q_1e_kS\otimes Se_{-k}S+Se_kS\otimes Q_1e_{-k}S+Q_1e_kS\otimes Q_1e_{-k}S\Big).
	\end{align*}
	Here we used
	\(
	V=\sum_k\hvb(k)(e_k\otimes e_{-k}).
	\)
	Now use the standard identity
	\[
	\mathbb W(X\otimes Y)=\frac12\bigl(\dG(X)\dG(Y)-\dG(XY)\bigr).
	\]
	Since $SQ_1=Q_1S=0$, the one-body correction vanishes in each of the six terms above. For instance,
	\[
	(Se_kS)(Se_{-k}Q_1)e_p=S_pQ_{1,p}S_{p-k}e_p=0,
	\]
	so
	\[
	\mathbb W(Se_kS\otimes Se_{-k}Q_1)=\frac12 A_kB_{-k}.
	\]
	Similarly,
	\[
	\mathbb W(Se_kQ_1\otimes Se_{-k}S)=\frac12 B_kA_{-k}=\frac12 A_{-k}B_k,
	\qquad
	\mathbb W(Se_kQ_1\otimes Se_{-k}Q_1)=\frac12 B_kB_{-k},
	\]
	\[
	\mathbb W(Q_1e_kS\otimes Se_{-k}S)=\frac12 C_kA_{-k},
	\qquad
	\mathbb W(Se_kS\otimes Q_1e_{-k}S)=\frac12 A_kC_{-k}=\frac12 C_{-k}A_k,
	\]
	\[
	\mathbb W(Q_1e_kS\otimes Q_1e_{-k}S)=\frac12 C_kC_{-k}.
	\]
	Here $B_kA_{-k}=A_{-k}B_k$ follows from $[B_k, A_{-k}]=0$ and similar for $A_kC_{-k}$. Summing over $k$ yields \eqref{eq:exact-mixed-decomp}.
\end{proof}

The decomposition \eqref{eq:exact-mixed-decomp} separates the mixed contribution into two qualitatively different pieces. The terms containing one $A$ and one $B$ or $C$ still involve one low-frequency leg and one far-tail leg; they will be controlled first by reducing them to weighted mixed-channel estimates. The remaining pair $B_kB_{-k}+C_kC_{-k}$ contains two $Q_1$ legs and cannot be treated as the mixed channel, so it will later be kept as a single off-diagonal two-body block between $S\otimes S$ and $Q_1\otimes Q_1$.

\subsection{The shell--mixed cross terms}

This subsection is devoted to the following estimate.

\begin{lemma}[Control of the shell--mixed cross terms]
	\label{lem:cross-terms-controlled}
	{\rmk{It holds that}}
	\begin{align}
		&\sum_{k\in\mathbb Z^2}\hvb(k)\lambda^2
		\Big|\Tr(A_kB_{-k}\Gamma_\lambda)\Big|
		+
		\sum_{k\in\mathbb Z^2}\hvb(k)\lambda^2
		\Big|\Tr(A_{-k}B_k\Gamma_\lambda)\Big|
		\notag
		\\
		&\qquad+
		\sum_{k\in\mathbb Z^2}\hvb(k)\lambda^2
		\Big|\Tr(C_kA_{-k}\Gamma_\lambda)\Big|
		+
		\sum_{k\in\mathbb Z^2}\hvb(k)\lambda^2
		\Big|\Tr(C_{-k}A_k\Gamma_\lambda)\Big|\notag
		\\ &\leq {\rmk{C_\beta\,\mathfrak r^{\rm tail}_{\beta,\lambda}.}}
		\label{eq:cross-terms-vanish}
	\end{align}
\end{lemma}

The basic point is that every operator $\dG(Le_\ell Q_1), L\in\{P,P_1,S\}$ carries one low-frequency leg and one far-tail leg. {\rmk{After summing in $\ell$, this structure produces a one-body operator bounded by $C_\beta(1+\log\Lambda_1)\Lambda_1^{2-\beta}Q_1$, so the problem reduces to controlling $\cN_{Q_1}$ together with the global number operator.}} Cauchy--Schwarz, Proposition~\ref{prop:number-moments-log}, and Lemma~\ref{lem:R1-correction} then yield the desired weighted estimate.

\begin{proposition}[Weighted mixed-channel estimate]
	\label{prop:weighted-mixed}
	Let $L\in\{P,P_1,S\}$ and define
	\[
	\rho_\ell^{L,Q_1}:=\dG(Le_\ell Q_1),
	\qquad \ell\in\mathbb Z^2.
	\]
	Then
	\begin{equation}
		\sum_{\ell\in\mathbb Z^2}\hvb(\ell)\,
		\lambda^2\Tr_{\mathfrak F}
		\Bigl((\rho_\ell^{L,Q_1})^*\rho_\ell^{L,Q_1}\Gamma_\lambda\Bigr)\leq {\rmk{C_\beta (1+|\log \lambda|^3)\lambda^{3/4-\alpha_1(3-\beta)}}} .
		\label{eq:weighted-mixed-vanishes}
	\end{equation}
\end{proposition}

\begin{proof}
	Fix $\ell\in\mathbb Z^2$ and set
	\(
	g:=Le_\ell Q_1.
	\)
	On the $n$-particle sector $\mathfrak H^{\otimes_s n}$ we have
	\[
	\dG(g)=\sum_{j=1}^n g_j,
	\]
	where $g_j$ acts as $g$ on the $j$-th variable. Hence, for every $\Psi$ in that sector,
	\[
	\|\dG(g)\Psi\|^2
	=
	\Big\|\sum_{j=1}^n g_j\Psi\Big\|^2
	\le n\sum_{j=1}^n\|g_j\Psi\|^2
	=
	\langle \Psi,\,n\,\dG(g^*g)\Psi\rangle.
	\]
	This yields on Fock space
	\[
	\dG(g)^*\dG(g)\le \cN\,\dG(g^*g).
	\]
	Applying this to $g=Le_\ell Q_1$, we obtain
	\begin{align}\label{rho}
		\sum_{\ell\in\mathbb Z^2}\hvb(\ell)
		(\rho_\ell^{L,Q_1})^*\rho_\ell^{L,Q_1}
		\le \cN\,\dG(K_L),
	\end{align}
	where
	\[
	K_L:=\sum_{\ell\in\mathbb Z^2}\hvb(\ell)
	Q_1e_{-\ell}Le_\ell Q_1.
	\]

	In the Fourier basis,
	\[
	e_\ell e_p=(2\pi)^{-1}e_{p+\ell},
	\qquad
	e_{-\ell}e_p=(2\pi)^{-1}e_{p-\ell}.
	\]
	Hence
	\[
	K_Le_p
	=
	\frac{Q_{1,p}}{(2\pi)^2}\Big(\sum_{\ell\in\mathbb Z^2}\hvb(\ell)L_{p+\ell}\Big)e_p
	=
	\frac{Q_{1,p}}{(2\pi)^2}\Big(\sum_{q: L_q=1}\hvb(q-p)\Big)e_p.
	\]
	
{
\color{red}
Since $L\le S$, the support of $L$ is contained in a ball of radius $C\Lambda_1$. For fixed $p$, the set
\[
E_p:=\{q-p:\ L_q=1\}
\]
is therefore a translate of a subset of $\{r\in\mathbb Z^2:\ |r|\le C\Lambda_1\}$. Because $\hvb(r)=\langle r\rangle^{-\beta}$ is positive and radially decreasing, the sum over $E_p$ is bounded by the corresponding centered ball:
\[
\sup_{p\in\mathbb Z^2}\sum_{q:L_q=1}\hvb(q-p)
\le C_\beta\sum_{|r|\le C\Lambda_1}\langle r\rangle^{-\beta}
\le C_\beta(1+\log\Lambda_1)\Lambda_1^{2-\beta}.
\]
 Hence
\[
0\le K_L\le C_\beta(1+\log\Lambda_1)\Lambda_1^{2-\beta}Q_1,
\qquad
\dG(K_L)\le C_\beta(1+\log\Lambda_1)\Lambda_1^{2-\beta}\cN_{Q_1}.
\]
}
	Plugging this into \eqref{rho} yields
	\[
	\sum_{\ell\in\mathbb Z^2}\hvb(\ell)
	(\rho_\ell^{L,Q_1})^*\rho_\ell^{L,Q_1}
	\le {\rmk{C_\beta(1+\log\Lambda_1)\Lambda_1^{2-\beta}}}\cN \cN_{Q_1}.
	\]

	Taking the expectation in $\Gamma_\lambda$ and multiplying by $\lambda^2$ gives
	\[
	\sum_{\ell\in\mathbb Z^2}\hvb(\ell)
	\lambda^2\Tr_{\mathfrak F}\Bigl((\rho_\ell^{L,Q_1})^*\rho_\ell^{L,Q_1}\Gamma_\lambda\Bigr)
	\le {\rmk{C_\beta(1+\log\Lambda_1)\Lambda_1^{2-\beta}}}\lambda^2\Tr_{\mathfrak F}(\cN \cN_{Q_1}\Gamma_\lambda).
	\]
	Then Cauchy--Schwarz implies
	\[
	\lambda^2\Tr_{\mathfrak F}(\cN \cN_{Q_1}\Gamma_\lambda)
	\le
	\Bigl(\lambda^2\Tr_{\mathfrak F}(\cN^2\Gamma_\lambda)\Bigr)^{1/2}
	\Bigl(\lambda^2\Tr_{\mathfrak F}(\cN_{Q_1}^2\Gamma_\lambda)\Bigr)^{1/2}.
	\]
	Using Proposition~\ref{prop:number-moments-log} and the first estimate in Lemma~\ref{lem:R1-correction}, we obtain
	\[
	\lambda^2\Tr_{\mathfrak F}(\cN \cN_{Q_1}\Gamma_\lambda)
	\le C(1+|\log\lambda|^2)\lambda^{3/4-\alpha_1}.
	\]
	{\rmk{Since $(1+\log\Lambda_1)\Lambda_1^{2-\beta}\le C_\beta(1+|\log\lambda|)\lambda^{-\alpha_1(2-\beta)}$, this proves \eqref{eq:weighted-mixed-vanishes}.}}
\end{proof}

\begin{proof}[Proof of Lemma \ref{lem:cross-terms-controlled}]
	{\color{red}Write
	\[
	A_k=A_k^{\rm low}+A_k^{\rm sh},
	\qquad
	A_k^{\rm low}:=\dG(Pe_kP),
	\qquad
	A_k^{\rm sh}:=\dG(Se_kS)-\dG(Pe_kP).
	\]
	The operator $A_k^{\rm sh}$ is exactly the shell observable denoted by $B_k$ in Section~\ref{sec:shell-vanishing-rewritten}. We treat the low and shell pieces separately.}
	
	{\color{red}\emph{Step 1: the low--mixed part.}}
	{\color{red}By Cauchy--Schwarz in the state $\Gamma_\lambda$,
	\[
	\sum_k\hvb(k)\lambda^2
	\big|\Tr_{\mathfrak F}(A_k^{\rm low}B_{-k}\Gamma_\lambda)\big|
	\le
	\Bigg(\sum_k\hvb(k)\lambda^2\Tr\big((A_k^{\rm low})(A_k^{\rm low})^*\Gamma_\lambda\big)\Bigg)^{1/2}
	\Bigg(\sum_k\hvb(k)\lambda^2\Tr\big(B_{-k}^*B_{-k}\Gamma_\lambda\big)\Bigg)^{1/2}.
	\]
	Since $Pe_kP\neq0$ implies $|k|\lesssim\Lambda$, the low--low estimate \eqref{eq:shell-lowlow-local} gives
	\[
	\sum_k\hvb(k)\lambda^2\Tr\big((A_k^{\rm low})(A_k^{\rm low})^*\Gamma_\lambda\big)
	\le
	C_\beta(1+\log\Lambda)\Lambda^{2-\beta}L_\lambda^2.
	\]
	On the other hand, Proposition~\ref{prop:weighted-mixed} with $L=S$ yields
	\[
	\sum_k\hvb(k)\lambda^2\Tr\big(B_{-k}^*B_{-k}\Gamma_\lambda\big)
	\le
	C_\beta(1+|\log\lambda|^3)\lambda^{3/4-\alpha_1(3-\beta)}.
	\]
	Hence
	\[
	\sum_k\hvb(k)\lambda^2
	\big|\Tr(A_k^{\rm low}B_{-k}\Gamma_\lambda)\big|
	\le
	C_\beta(1+|\log\lambda|^3)\lambda^{\frac38-\frac{\alpha_1}{2}(3-\beta)-\frac\alpha2(2-\beta)}.
	\]}
	
	{\color{red}\emph{Step 2: the shell--mixed part.}}
	{\color{red}Using Corollary~\ref{cor:shell-summability-local} for $A_k^{\rm sh}$ and Proposition~\ref{prop:weighted-mixed} for $B_{-k}$, we obtain
	\begin{align*}
	\sum_k\hvb(k)\lambda^2
	\big|\Tr(A_k^{\rm sh}B_{-k}\Gamma_\lambda)\big|
	&\le
	\Bigg(\sum_k\hvb(k)\lambda^2\Tr\big(A_k^{\rm sh}(A_k^{\rm sh})^*\Gamma_\lambda\big)\Bigg)^{1/2}
	\Bigg(\sum_k\hvb(k)\lambda^2\Tr\big(B_{-k}^*B_{-k}\Gamma_\lambda\big)\Bigg)^{1/2}
	\\
	&\le
	C_\beta(1+|\log\lambda|^3)
	\Big(
	\lambda^{\frac38-\frac{\alpha_1}{2}(3-\beta)+\frac{\alpha\eta}{2}}
	+
	\lambda^{\frac98-\frac{\alpha_1}{2}(5-\beta)}
	\Big).
	\end{align*}}

	{\color{red}\emph{Step 3: the remaining three sums.}}
	{\color{red}Exactly the same argument applies to
	\[
	\sum_k\hvb(k)\lambda^2\big|\Tr(A_{-k}B_k\Gamma_\lambda)\big|.
	\]
	For the last two sums we use
	\[
	C_k=B_{-k}^*,
	\qquad
	C_{-k}=B_k^*,
	\]
	and repeat the same Cauchy--Schwarz estimate. Collecting the three exponents obtained above gives precisely the error $\mathfrak r^{\rm tail}_{\beta,\lambda}$ defined at the beginning of this section. This proves \eqref{eq:cross-terms-vanish}.}
\end{proof}

\subsection{The high--high mixed pair}

Define the following high--high mixed pair
\[
\mathcal T_{\rm hh}
:=
\frac{\lambda^2}{4}
\sum_{k\in\mathbb Z^2}\hvb(k)
\Tr\Bigl((B_kB_{-k}+C_kC_{-k})\Gamma_\lambda\Bigr).
\]
A direct repetition of Proposition~\ref{prop:weighted-mixed} is no longer sufficient here, because after summing over $k$ one encounters the tail kernel
\[
\sum_{q\in Q_1}\hvb(q-p),
\]
{\rmk{which diverges. We therefore keep $\mathcal T_{\rm hh}$ as a single off-diagonal two-body block. The point is that this block has operator norm of order $C_\beta(1+\log\Lambda_1)\Lambda_1^{2-\beta}$, and this loss can then be paired with the small $Q_1\otimes Q_1$ mass furnished by Lemma~\ref{lem:R1-correction}.}}

\begin{proposition}[Vanishing of the high--high mixed pair]
	\label{prop:high-high-mixed-vanishes}
	It holds that
	\[
	|\mathcal T_{\rm hh}|\leq {\rmk{C_\beta(1+|\log \lambda|^3)\lambda^{3/4-\alpha_1(3-\beta)}}}.
	\]
\end{proposition}
\begin{lemma}[The high--high mixed pair as a two-body off-diagonal block]
	\label{lem:Thh-two-body}
	Define the two-body operator
	\[
	\mathbb T_{\rm hh}
	:=
	\sum_{k\in\mathbb Z^2}\hvb(k)
	\bigl(Se_kQ_1\otimes Se_{-k}Q_1\bigr)
	\]
	on $H^{\otimes 2}$. Then
	\[
	\mathbb T_{\rm hh}=(S\otimes S)\mathbb T_{\rm hh}(Q_1\otimes Q_1),
	\qquad
	\mathbb T_{\rm hh}^*=(Q_1\otimes Q_1)\mathbb T_{\rm hh}^*(S\otimes S).
	\]
	Moreover,
	\begin{equation}
		\mathcal T_{\rm hh}
		=
		\frac{\lambda^2}{2}\Tr_{H^{\otimes 2}}
		\Bigl((\mathbb T_{\rm hh}+\mathbb T_{\rm hh}^*)\Gamma_\lambda^{(2)}\Bigr).
		\label{eq:Thh-two-body}
	\end{equation}
\end{lemma}

\begin{proof}
	By definition,
	\[
	\mathbb T_{\rm hh}
	=
	\sum_k\hvb(k)(Se_kQ_1\otimes Se_{-k}Q_1),
	\]
	so its range is contained in $S\otimes S$ and its domain in $Q_1\otimes Q_1$, proving the first pair of identities.
	
	For the second statement, use the standard two-body quantization identity
	\[
	\Tr_{\mathfrak F}(\mathbb W(V)\Gamma)=\Tr_{H^{\otimes 2}}(V\Gamma^{(2)}).
	\]
	Now
	\[
	\mathbb W(\mathbb T_{\rm hh})
	=
	\frac1{2}\sum_k\hvb(k)B_kB_{-k}.
	\]
	Here we used
	\[
	(Se_kQ_1)(Se_{-k}Q_1)=0=(Se_{-k}Q_1)(Se_kQ_1)
	\]
	by the orthogonality $Q_1S=0$, and no one-body correction appears. Likewise,
	\[
	\mathbb	W(\mathbb T_{\rm hh}^*)
	=
	\frac1{2}\sum_k\hvb(k)C_kC_{-k}.
	\]
	Taking the expectation in $\Gamma_\lambda$ and multiplying by $\lambda^2/2$ yields \eqref{eq:Thh-two-body}.
\end{proof}

\begin{lemma}[Operator norm bound for the high--high block]
	\label{lem:Thh-norm}
	With $\mathbb T_{\rm hh}$ as in Lemma~\ref{lem:Thh-two-body}, one has
	\begin{equation}
		\|\mathbb T_{\rm hh}\|_{\mathrm{op}}\le {\rmk{C_\beta(1+\log \Lambda_1)\Lambda_1^{2-\beta}}}.
		\label{eq:Thh-norm}
	\end{equation}
\end{lemma}

\begin{proof}
	We argue fiberwise with respect to the conserved total momentum.
	
	\medskip
	\noindent
	\textbf{Step 1. Total-momentum fibers.}
	Let
	\[
	\mathcal D:=\mathrm{span}\{e_u\otimes e_v:\ u,v\in\mathbb Z^2\}\subset H^{\otimes 2}.
	\]
	For every basis vector,
	\[
	\mathbb T_{\rm hh}(e_u\otimes e_v)
	=
	\frac1{(2\pi)^2}\sum_{k\in\mathbb Z^2}\hvb(k)
	Q_{1,u}Q_{1,v}S_{u+k}S_{v-k}\,e_{u+k}\otimes e_{v-k}.
	\]
	Since $S$ has finite Fourier support, only finitely many $k$ contribute for fixed $u,v$, so $\mathbb T_{\rm hh}$ is well defined on $\mathcal D$.
	
	For each $\ell\in\mathbb Z^2$, define
	\[
	\mathcal H_\ell:=\overline{\mathrm{span}}\{e_p\otimes e_{\ell-p}:\ p\in\mathbb Z^2\}.
	\]
	Then
	\[
	H^{\otimes 2}=\bigoplus_{\ell\in\mathbb Z^2}\mathcal H_\ell
	\]
	orthogonally, and $\mathbb T_{\rm hh}$ preserves each fiber because
	\[
	e_u\otimes e_{\ell-u}\longmapsto e_{u+k}\otimes e_{\ell-u-k}.
	\]
	
	\medskip
	\noindent
	\textbf{Step 2. Formula on each fiber.}
	Identify $\mathcal H_\ell$ with $\ell^2(\mathbb Z^2)$ through
	\[
	e_p\otimes e_{\ell-p}\longleftrightarrow \delta_p.
	\]
	The restriction of $\mathbb T_{\rm hh}$ to $\mathcal H_\ell$ is then the operator
	\[
	(T_\ell^0f)(p)
	=
	\frac1{(2\pi)^2}\mathbf 1_{\mathcal R_\ell}(p)
	\sum_{u\in\mathcal Q_\ell}\hvb(p-u)f(u),
	\]
	where
	\[
	\mathcal R_\ell:=\{p:\ S_pS_{\ell-p}=1\},
	\qquad
	\mathcal Q_\ell:=\{u:\ Q_{1,u}Q_{1,\ell-u}=1\}.
	\]
	Thus $T_\ell^0$ is a truncated convolution operator from $\ell^2(\mathcal Q_\ell)$ to $\ell^2(\mathcal R_\ell)$.
	
	\medskip
	\noindent
	\textbf{Step 3. Schur bound on $T_\ell^0(T_\ell^0)^*$.}
	Since $\mathcal R_\ell$ is finite, the positive operator
	\[
	M_\ell:=T_\ell^0(T_\ell^0)^*
	\]
	has kernel
	\[
	M_\ell(p,p')
	=
	\frac1{(2\pi)^4}\sum_{u\in\mathcal Q_\ell}\hvb(p-u)\hvb(p'-u).
	\]
{
\color{red}
Using $\hvb(r)=\langle r\rangle^{-\beta}$ and Lemma~\ref{lem:sum}, we obtain the uniform convolution bound
\[
0\le M_\ell(p,p')\le C_\beta A_\beta(p-p'),
\qquad
A_\beta(d):=\sum_{r\in\mathbb Z^2}\langle r\rangle^{-\beta}\langle r+d\rangle^{-\beta}
\le C_\beta(1+\log(2+|d|))\langle d\rangle^{-(2\beta-2)}.
\]
Indeed, for $\beta<2$ this follows from Lemma~\ref{lem:sum}(i) with $d=2$ and $l=m=\beta$, while at the endpoint $\beta=2$ we use Lemma~\ref{lem:sum}(ii). Since every $p,p'\in\mathcal R_\ell$ satisfy $|p|,|p'|\lesssim \Lambda_1$, we have $|p'-p|\lesssim \Lambda_1$. Therefore
\[
\sum_{p'\in\mathcal R_\ell}M_\ell(p,p')
\le C_\beta\sum_{|d|\lesssim \Lambda_1}(1+\log(2+|d|))\langle d\rangle^{-(2\beta-2)}
\le C_\beta(1+\log\Lambda_1)^2\Lambda_1^{4-2\beta},
\]
and the same bound holds for the column sums. By Schur's test,
\[
\|M_\ell\|_{\mathrm{op}}\le C_\beta(1+\log\Lambda_1)^2\Lambda_1^{4-2\beta}.
\]
Hence
\[
\|T_\ell^0\|_{\mathrm{op}}\le C_\beta(1+\log\Lambda_1)\Lambda_1^{2-\beta}.
\]
}
	
	\medskip
	\noindent
	\textbf{Step 4. Reconstruction on $H^{\otimes 2}$.}
	Each $T_\ell^0$ extends boundedly to $T_\ell:\mathcal H_\ell\to\mathcal H_\ell$, and
	\[
	\sup_{\ell\in\mathbb Z^2}\|T_\ell\|_{\mathrm{op}}\le {\rmk{C_\beta(1+\log\Lambda_1)\Lambda_1^{2-\beta}}}.
	\]
	Thus the orthogonal direct sum
	\[
	\widetilde T:=\bigoplus_{\ell\in\mathbb Z^2}T_\ell
	\]
	is a bounded operator on $H^{\otimes 2}$ with
	\[
	\|\widetilde T\|_{\mathrm{op}}=\sup_{\ell}\|T_\ell\|_{\mathrm{op}}\le {\rmk{C_\beta(1+\log\Lambda_1)\Lambda_1^{2-\beta}}}.
	\]
	By construction, $\widetilde T$ coincides with $\mathbb T_{\rm hh}$ on the dense subspace $\mathcal D$, hence it is the closure of $\mathbb T_{\rm hh}$. This proves \eqref{eq:Thh-norm}.
\end{proof}

\begin{lemma}[Block Cauchy--Schwarz for positive trace-class operators]
	\label{lem:block-cs}
	Let $\Gamma\ge 0$ be trace-class on a Hilbert space $\mathfrak H$, let $\Pi_1,\Pi_2$ be orthogonal projections, and let $T=\Pi_1T\Pi_2$ be bounded. Then
	\[
	\big|\Tr_{\mathfrak H}(T\Gamma)\big|
	\le \|T\|_{\mathrm{op}}
	\Big(\Tr_{\mathfrak H}(\Pi_1\Gamma)\Big)^{1/2}
	\Big(\Tr_{\mathfrak H}(\Pi_2\Gamma)\Big)^{1/2}.
	\]
\end{lemma}

\begin{proof}
	Since $\Gamma\ge 0$, write $\Gamma=\Gamma^{1/2}\Gamma^{1/2}$. Then
	\[
	\Tr(T\Gamma)=\Tr(\Gamma^{1/2}\Pi_1T\Pi_2\Gamma^{1/2}).
	\]
	Applying Hilbert--Schmidt Cauchy--Schwarz gives
	\[
	|\Tr(T\Gamma)|
	\le \|\Gamma^{1/2}\Pi_1\|_{\mathrm{HS}}\,\|T\|_{\mathrm{op}}\,\|\Pi_2\Gamma^{1/2}\|_{\mathrm{HS}}.
	\]
	Finally,
	\[
	\|\Gamma^{1/2}\Pi_1\|_{\mathrm{HS}}^2=\Tr(\Pi_1\Gamma),
	\qquad
	\|\Pi_2\Gamma^{1/2}\|_{\mathrm{HS}}^2=\Tr(\Pi_2\Gamma).
	\]
\end{proof}

\begin{proof}[Proof of Proposition \ref{prop:high-high-mixed-vanishes}]
	By Lemma~\ref{lem:Thh-two-body},
	\[
	\mathcal T_{\rm hh}
	=
	\frac{\lambda^2}{2}\Tr_{H^{\otimes 2}}
	\Bigl((\mathbb T_{\rm hh}+\mathbb T_{\rm hh}^*)\Gamma_\lambda^{(2)}\Bigr).
	\]
	Since
	\[
	\mathbb T_{\rm hh}=(S\otimes S)\mathbb T_{\rm hh}(Q_1\otimes Q_1),
	\qquad
	\mathbb T_{\rm hh}^*=(Q_1\otimes Q_1)\mathbb T_{\rm hh}^*(S\otimes S),
	\]
	Lemma~\ref{lem:block-cs} yields
	\begin{align*}
		\lambda^2\big|\Tr_{H^{\otimes 2}}(\mathbb T_{\rm hh}\Gamma_\lambda^{(2)})\big|
		&\le \|\mathbb T_{\rm hh}\|_{\mathrm{op}}
		\Bigl(\lambda^2\Tr_{H^{\otimes 2}}((S\otimes S)\Gamma_\lambda^{(2)})\Bigr)^{1/2}
		\\
		&\qquad\times
		\Bigl(\lambda^2\Tr_{H^{\otimes 2}}((Q_1\otimes Q_1)\Gamma_\lambda^{(2)})\Bigr)^{1/2}.
	\end{align*}
	The same estimate holds for $\mathbb T_{\rm hh}^*$. Therefore,
	\[
	|\mathcal T_{\rm hh}|
	\le C\|\mathbb T_{\rm hh}\|_{\mathrm{op}}
	\Bigl(\lambda^2\Tr_{H^{\otimes 2}}((S\otimes S)\Gamma_\lambda^{(2)})\Bigr)^{1/2}
	\Bigl(\lambda^2\Tr_{H^{\otimes 2}}((Q_1\otimes Q_1)\Gamma_\lambda^{(2)})\Bigr)^{1/2}.
	\]
	By Proposition~\ref{prop:number-moments-log} we obtain
	\[
	\lambda^2\Tr((S\otimes S)\Gamma_\lambda^{(2)})
	\le \lambda^2 \Tr(\cN^2\Gamma_\lambda)
	\lesssim 1+|\log\lambda|^2.
	\]
{
\color{red}
Therefore, by Lemma~\ref{lem:Thh-norm} and the first estimate in Lemma~\ref{lem:R1-correction},
\[
|\mathcal T_{\rm hh}|
\le
C_\beta(1+\log\Lambda_1)\Lambda_1^{2-\beta}(1+|\log\lambda|)
\Bigl(\lambda^2\Tr_{H^{\otimes 2}}(Q_1^{\otimes 2}\Gamma_\lambda^{(2)})\Bigr)^{1/2}
\le C_\beta(1+|\log\lambda|^3)\lambda^{3/4-\alpha_1(3-\beta)}.
\]
This proves the proposition.
}
\end{proof}


\begin{proof}[Proof of Theorem \ref{cor:tail-lower-bound}]
	By Lemma~\ref{lem:exact-mixed-decomp}, Lemma~\ref{lem:cross-terms-controlled}, and Proposition~\ref{prop:high-high-mixed-vanishes}, we obtain
	\[
	\frac{\lambda^2}{2}
	\Tr_{H^{\otimes 2}}
	\Bigl((\Pi_1VR_1+R_1V\Pi_1)\Gamma_\lambda^{(2)}\Bigr)
	\ge -{\rmk{C_\beta\,\mathfrak r^{\rm tail}_{\beta,\lambda}.}}
	\]
	Combining this with \eqref{eq:R1-correction-rate-special} yields \eqref{eq:tail-lower-rate}.
\end{proof}



\section{Quantum free energy}
\label{sec:free-energy-density}\label{sec:free-energy}

The estimates obtained in Sections~\ref{sec:dec}--\ref{sec:tail-vanishing-checked} are now sufficient to close the variational argument.
We collect in this section the three final steps of the proof:  the matching lower bound extracted from the vanishing of the high-frequency
remainders, the upper bound on the relative
free energy, and the resulting convergence of the free energy.

\subsection{Low-frequency localization and free-energy lower bound}\label{sec:low-freq}\label{sec:lower}
In this subsection we derive the lower bound for the renormalized quantum free energy by reducing the problem to a finite-dimensional classical variational principle on the low-frequency space $PH$. The basic idea is to localize the Gibbs state, compare the localized interaction with the Hartree functional $D_P$ through lower symbols and a quantitative de Finetti expansion, and keep the high-frequency contribution in an explicit remainder form. We then combine monotonicity of the relative entropy and the Berezin--Lieb inequality with the estimates from Sections~\ref{sec:dec}--\ref{sec:tail-vanishing-checked} to show that these remainders are negligible, which yields a lower bound in terms of the finite-dimensional classical partition function.

We use the standard Fock-space factorization $\F(H)\simeq \F(PH)\otimes \F(QH)$ with the unitary $U$ satisfying
\[
Ua^\ast(f)U^\ast = a^\ast(Pf)\otimes 1 + 1\otimes a^\ast(Qf),\qquad f\in H,
\]
and similarly for $a(f)$. For a state $\Gamma$ on $\F(H)$ we define its localization (partial trace)
\[
\Gamma_P:=\Tr_{\F(QH)}[U\Gamma U^\ast],
\]
which is a state on $\F(PH)$. Then, for every $k\ge 1$, its $k$-particle density matrix satisfies
\begin{equation}\label{eq:GammaP-k-body}
	(\Gamma_P)^{(k)}=P^{\otimes k}\Gamma^{(k)}P^{\otimes k}.
\end{equation}
For $k\in\Z^2$ let $M_{e_k}$ denote the multiplication operator $(M_{e_k}f)(x)=e_k(x)f(x)$ on $H$ and define
\[
e_{k,P}:=PM_{e_k}P\quad\text{on }PH,\qquad \rho_{k,P}:=\dG(e_{k,P}).
\]
Recall from Section~\ref{sec:dec} that
\[
V_{\neq0,P}=P^{\otimes2}V_{\neq0}P^{\otimes2},
\qquad
V_{\neq0}=\sum_{k\neq0}\hvb(k)\,(M_{e_k}\otimes M_{e_{-k}}).
\]
Set $\mathcal K:=\{p\in\Z^2:\ h(p)\le \Lambda^2\}$. Since $\mathbb W(\cdot)$ was fixed in Section~\ref{sec:setting}, we have
\begin{align*}
	\mathbb W(V_{\neq 0,P})
	&=\frac1{2(2\pi)^2}\sum_{k\neq 0}\hvb(k)
	\sum_{p,q\in \mathcal K}
	\1_{\{p+k\in \mathcal K\}}\1_{\{q-k\in \mathcal K\}}
	a_{p+k}^* a_{q-k}^* a_q a_p\\
	&=\frac1{2}\sum_{k\neq 0}\hvb(k)
	\Bigl(\rho_{k,P}\rho_{-k,P}-\d\Gamma(e_{k,P}e_{-k,P})\Bigr),
\end{align*}
We also set
\begin{equation}\label{eq:WreP-def}
	W_P^{\mathrm{re}}
	:=
	\lambda^2\mathbb W(V_{\neq 0,P})
	+
	\frac{\lambda^2}{2(2\pi)^2}(\cN_P-N_{0,P})^2.
\end{equation}

For every bounded symmetric two-body operator $V$ on $(PH)^{\otimes2}$ and every state $\Gamma_P$ on $\mathfrak F(PH)$,
\begin{equation}
	\Tr_{\mathfrak F(PH)}\bigl(\mathbb W(V)\Gamma_P\bigr)
	=\Tr_{(PH)^{\otimes_s 2}}\bigl(V\Gamma_P^{(2)}\bigr).
\end{equation}

In this subsection we keep the exact decomposition
\begin{equation}\label{eq:exact-localization}
	\Tr_{\mathfrak F(H)}\!\big[W^{\mathrm{re}}\Gamma\big]
	=
	\Tr_{\mathfrak F(PH)}\!\big[W^{\mathrm{re}}_P \Gamma_P\big]
	+
	HF_{\neq 0}(\Gamma)
	+
	HF_0(\Gamma),
\end{equation}
with $HF_{\neq0}$ and $HF_0$ defined in \eqref{eq:HF-final} and \eqref{eq:HF-zero}
(respectively, in the notation introduced above).

Write $K:=\Tr(P)=\dim(PH)$. For a state $\Gamma$ on $\mathfrak F(H)$, we define the lower symbol of its localization $\Gamma_P$ on $PH$ at scale $\lambda$ by
\[
\d\mu^\lambda_{P,\Gamma}(u)
:=
(\lambda\pi)^{-K}
\big\langle W(u/\sqrt\lambda),\Gamma_P\,W(u/\sqrt\lambda)\big\rangle\,\d u,
\qquad u\in PH,
\]
where
\begin{equation}\label{eq:coherent state}
	W(u):= \exp(a^{*}(u)-a(u))  |0\rangle = e^{-\|u\|^2/2} \exp\left(a^{*}(u)\right)  |0\rangle = e^{-\|u\|^2/2} \bigoplus_{n=0}^\infty \frac{u^{\otimes n}} {\sqrt{n!}}
\end{equation}
is the coherent state on $\gF (P\gH)$, with $|0\rangle$ the vacuum in $\gF (P\gH)$ and $\|u\|$ the $L^2$-norm of $u\in P\gH$.
We define the finite-dimensional Hartree functional
\begin{equation}\label{eq:DP-def-new}
	D_P(u)
	:=
	\frac{1}{2}\sum_{k\neq0}\hvb(k)\,
	\big|\langle u,e_{k,P}u\rangle\big|^2
	+
	\frac{1}{2(2\pi)^2}
	\Big(\|u\|^2-(2\pi)^2\rho_{0,P}\Big)^2,
	\qquad
	\rho_{0,P}:=\frac{1}{(2\pi)^2}\Tr(Ph^{-1}).
\end{equation}

We first compare the localized quantum interaction
\[
\Tr_{\mathfrak F(PH)}\!\bigl[W^{\mathrm{re}}_P\Gamma_P\bigr]
\]
with
\[
\int_{PH} D_P(u)\,\d\mu^\lambda_{P,\Gamma}(u),
\]
up to an error that is explicit and small in the regime
\[
\Lambda=\lambda^{-\alpha},
\qquad
0<\alpha<\frac14.
\]

Recall the following quantitative version of the quantum de Finetti theorem \cite[Lemma 6.2
and Remark 6.4]{LewNamRou-15}.

\begin{theorem}[Quantitative quantum de Finetti]\label{thm:quant deF}
		For all $k\in\N$, we have
		\begin{equation}\label{Peq:Chiribella}
			\int_{P\gH}|u^{\otimes k}\rangle\langle u^{\otimes k}|\,\d\mu^\lambda_{P,\Gamma}(u)
			=
			k!\lambda^k\Gamma_P^{(k)}
			+
			k!\lambda^k\sum_{\ell=0}^{k-1}\binom{k}{\ell}\,
			\Gamma_P^{(\ell)}\otimes_s \mathbf 1_{(P\gH)^{\otimes_s (k-\ell)}}.
		\end{equation}
		Thus, with $d=\Tr[P]$,
		\begin{equation}\label{eq:quantitative}
			\Tr \left| k!\lambda^k\Gamma_P^{(k)}-\int_{P\gH}|u^{\otimes k}\rangle\langle u^{\otimes k}|\,\d\mu^\lambda_{P,\Gamma}(u) \right|
			\le \lambda^k \sum_{\ell=0}^{k-1}\binom{k}{\ell}^2 \frac{(k-\ell+d-1)!}{(d-1)!}\Tr \left[ \cN^{\ell}\Gamma_P\right].
		\end{equation}
\end{theorem}

\begin{lemma}[Explicit $k\neq0$ de Finetti expansion against $V_{\neq0,P}$]
	\label{lem:nonzero-explicit-corrected}
	Let $\mu^{\lambda}_{P,\Gamma}$ be the lower symbol on $PH$ at scale $\lambda$. Then
	\begin{equation}
		\Big|\lambda^2\,\Tr_{\Fock(PH)}\bigl(\mathbb W(V_{\neq 0,P})\Gamma_P\bigr)
		-
		\frac1{2}\sum_{k\neq 0}\hvb(k)
		\int_{PH}|\langle u,e_{k,P}u\rangle|^2\,\d\mu^{\lambda}_{P,\Gamma}(u)\Big|\le {\rmk{C_\beta\lambda^2\Lambda^{4-\beta}(1+\log\Lambda)}}
		.
		\label{eq:213}
	\end{equation}
\end{lemma}

\begin{proof}
	Apply Theorem~\ref{thm:quant deF} with $k=2$:
	\[
	\int_{PH}|u^{\otimes 2}\rangle\langle u^{\otimes 2}|\,\d\mu^{\lambda}_{P,\Gamma}(u)
	=2\lambda^2\Gamma_P^{(2)}+4\lambda^2(\Gamma_P^{(1)}\hs P)+2\lambda^2 P^{\hs 2}.
	\]
	Tracing against $V_{\neq 0,P}$ gives
	\begin{align}
		\int_{PH}\langle u^{\otimes 2},V_{\neq 0,P}u^{\otimes 2}\rangle\,\d\mu^{\lambda}_{P,\Gamma}(u)
		&=2\lambda^2\Tr\bigl(V_{\neq 0,P}\Gamma_P^{(2)}\bigr)
		+4\lambda^2\Tr\bigl(V_{\neq 0,P}(\Gamma_P^{(1)}\hs P)\bigr)
		\notag\\
		&\qquad +2\lambda^2\Tr\bigl(V_{\neq 0,P}P^{\hs 2}\bigr).
		\label{eq:217}
	\end{align}
	Next,
	\[
	\langle u^{\otimes 2},V_{\neq 0,P}u^{\otimes 2}\rangle
	=\sum_{k\neq 0}\hvb(k)
	\langle u,e_{k,P}u\rangle\langle u,e_{-k,P}u\rangle
	=\sum_{k\neq 0}\hvb(k)|\langle u,e_{k,P}u\rangle|^2,
	\]
	because $e_{-k,P}=e_{k,P}^*$. This gives the classical term in \eqref{eq:213}.
	
	We also have
	\begin{equation}
		\Tr_{(PH)^{\hs 2}}\bigl(V_{\neq 0,P}(\Gamma_P^{(1)}\hs P)\bigr)=0,
		\label{eq:214}
	\end{equation}
	since \[
	\Tr_{PH}(e_{-k,P})=0\qquad (k\neq 0).
	\]
	
	Moreover, we have
	\begin{equation}
		\Tr_{(PH)^{\hs 2}}\bigl(V_{\neq 0,P}P^{\hs 2}\bigr)
		=
		\frac1{2}\sum_{k\neq 0}\hvb(k)\Tr_{PH}(e_{k,P}e_{-k,P})
		=
		\frac1{2}\sum_{k\neq 0}\hvb(k)\|Pe_kP\|_{\mathrm{HS}}^2,
		\label{eq:215}
	\end{equation}
	and in particular
	\[
	\Tr_{(PH)^{\hs 2}}(V_{\neq 0,P}P^{\hs 2})
	\lesssim \Lambda^2\sum_{|k|\le 2\Lambda}\hvb(k)
	\le {\rmk{C_\beta\Lambda^{4-\beta}(1+\log\Lambda).}}
	\]
	\rmk{Indeed, in dimension $2$ one has $\sum_{|k|\le 2\Lambda}\langle k\rangle^{-\beta}\le C_\beta\Lambda^{2-\beta}(1+\log\Lambda)$ for $\frac32<\beta\le2$.}
	Substituting \eqref{eq:214} and \eqref{eq:215} into \eqref{eq:217} gives \eqref{eq:213}.
\end{proof}


\begin{proposition}[Finite-dimensional interaction comparison]
	\label{prop:localized-interaction-vs-DP}
	For every state $\Gamma$ on $\mathfrak F(H)$,
	\begin{equation}\label{eq:localized-interaction-vs-DP}
		\Tr_{\mathfrak F(PH)}\!\big[W_P^{\mathrm{re}}\Gamma_P\big]
		=
		\int_{PH} D_P(u)\,\d\mu^\lambda_{P,\Gamma}(u)
		+
		\mathcal E_{P,\lambda}(\Gamma),
	\end{equation}
	where
	\begin{align}
		|\mathcal E_{P,\lambda}(\Gamma)|
		&\le
		{\rmk{C_\beta\lambda^2 \Lambda^{4-\beta}(1+\log\Lambda)}}
		+
		C\lambda \Lambda^2
		\left(
		\log\Lambda+\int_{PH}\|u\|^2\,\d\mu^\lambda_{P,\Gamma}(u)
		\right).
		\label{eq:localized-interaction-vs-DP-error}
	\end{align}
\end{proposition}

\begin{proof}
	Split
	\[
	W_P^{\mathrm{re}}
	=
	\lambda^2\mathbb W(V_{\neq0,P})
	+
	\frac{\lambda^2}{2(2\pi)^2}(\cN_P-N_{0,P})^2.
	\]
	For the $k\neq0$ part, we have Lemma \ref{lem:nonzero-explicit-corrected}.
	For the $k=0$ part, applying Theorem~\ref{thm:quant deF} yields
	\begin{equation}
		\int_{PH}\|u\|^2\,\d\mu^{\lambda}_{P,\Gamma}(u)
		=\lambda\Tr_{\Fock(PH)}(\cN_P\Gamma_P)+\lambda K,
		\label{eq:221a}
	\end{equation}
	and
	\begin{align}
		\int_{PH}\|u\|^4\,\d\mu^{\lambda}_{P,\Gamma}(u)
		&=\lambda^2\Tr\bigl(\cN_P(\cN_P-1)\Gamma_P\bigr)
		+2\lambda^2(K+1)\Tr(\cN_P\Gamma_P)+\lambda^2 K(K+1).
		\label{eq:221b}
	\end{align}
		Set
		\[
		a_{P,\lambda}:=\lambda K+\lambda N_{0,P}=(2\pi)^2\rho_{0,P}+\delta_{P,\lambda},
		\]
		where
		\[
		\delta_{P,\lambda}:=\lambda K+\lambda N_{0,P}-(2\pi)^2\rho_{0,P}.
		\]
		Expanding the square and using \eqref{eq:221a}--\eqref{eq:221b}, one gets the exact identity
		\begin{equation}
			\lambda^2\Tr\bigl((\cN_P-N_{0,P})^2\Gamma_P\bigr)
			=
			\int_{PH}(\|u\|^2-a_{P,\lambda})^2\,\d\mu^{\lambda}_{P,\Gamma}(u)
			-\lambda\int_{PH}\|u\|^2\,\d\mu^{\lambda}_{P,\Gamma}(u).
			\label{eq:221d}
		\end{equation}
		For $h(p)\le \Lambda^2$ and $\lambda \Lambda^2\le1$, we have
	\[
	\Big|\frac{\lambda}{e^{\lambda h(p)}-1}
	-
	\frac1{h(p)}\Big|\lesssim\lambda
	\]
	uniformly in $p\in \mathcal K$. Summing over $p\in\mathcal K$ yields
	\[
	\left|
	\lambda N_{0,P}-\Tr(Ph^{-1})
	\right|
	\le C\lambda K \le C\lambda \Lambda^2.
	\]
	Adding the extra $\lambda K$ term gives
	\begin{equation}\label{eq:center-shift-bound}
		|\delta_{P,\lambda}|\le C\lambda \Lambda^2.
	\end{equation}
	Consequently,
	\begin{align}
		\Big|
		\lambda^2\Tr_{\mathfrak F(PH)}\!\big[(\cN_P-N_{0,P})^2\Gamma_P\big]
		-
		\int_{PH}\Big(\|u\|^2-(2\pi)^2\rho_{0,P}\Big)^2\,\d\mu^\lambda_{P,\Gamma}(u)
		\Big|
		\nonumber\\
		\le
		C\lambda \Lambda^2
		\left(
		\log\Lambda+\int_{PH}\|u\|^2\,\d\mu^\lambda_{P,\Gamma}(u)
		\right).
		\label{eq:number-square-classical-bound}
	\end{align}
	Combining the estimates for $k\neq 0$ and $k=0$ gives
	\eqref{eq:localized-interaction-vs-DP}--\eqref{eq:localized-interaction-vs-DP-error}.
\end{proof}

\begin{lemma}[Pointwise comparison of the free lower symbol and the Gaussian free field]
	\label{lem:pointwise-comparison}
	Let $\mu^\lambda_{P,0}$ denote the lower symbol of the free localized Gibbs state $(\Gamma_0)_P$,
	and let $\mu_{0,P}$ denote the Gaussian measure on $PH$ with covariance $h^{-1}P$. If
	$\lambda \Lambda^2\le 1$, then
	\begin{equation}\label{eq:mu-pointwise-two-sided}
		e^{-C\lambda \Lambda^4(1+\|u\|^2)}\,\d\mu_{0,P}(u)
		\le
		\d\mu^\lambda_{P,0}(u)
		\le
		e^{C\lambda \Lambda^4(1+\|u\|^2)}\,\d\mu_{0,P}(u).
	\end{equation}
	Consequently, for every probability measure $\nu$ on $PH$,
	\begin{equation}\label{eq:entropy-comparison}
		\mathcal{H}_{\mathrm{cl}}(\nu,\mu^\lambda_{P,0})
		\ge
		\mathcal{H}_{\mathrm{cl}}(\nu,\mu_{0,P})
		-
		C\lambda \Lambda^4
		\int_{PH}(1+\|u\|^2)\,\d\nu(u).
	\end{equation}
\end{lemma}

\begin{proof} The proof follows by the same argument as in \cite[(8.25)]{NZZ25}.
	Diagonalize $PhP=\sum_{j=1}^K \lambda_j |e_j\rangle\langle e_j|$.
	In the coordinates $\alpha_j=\langle e_j,u\rangle$, the free lower symbol factorizes:
	\[
	\d\mu^\lambda_{P,0}(u)
	=
	\prod_{j=1}^K
	\Big[
	(\lambda\pi)^{-1}(1-e^{-\lambda\lambda_j})
	\exp\Big(-\frac{1-e^{-\lambda\lambda_j}}{\lambda}|\alpha_j|^2\Big)
	\Big]
	\,\d\alpha_j .
	\]
	Since $\lambda_j\le \Lambda^2$ and $\lambda \Lambda^2\le1$, we have
	\[
	1-e^{-\lambda\lambda_j}
	=
	\lambda\lambda_j + O(\lambda^2\lambda_j^2),
	\]
	uniformly in $j$, and therefore
	\[
	\frac{1-e^{-\lambda\lambda_j}}{\lambda}
	=
	\lambda_j + O(\lambda\lambda_j^2).
	\]
	Using $\lambda_j\le \Lambda^2$ and summing over $j$ yields
	\eqref{eq:mu-pointwise-two-sided}. Then \eqref{eq:entropy-comparison} follows immediately from
	the definition of relative entropy.
\end{proof}

\begin{lemma}[Second moment of the Husimi field]
	\label{cor:husimi-second-moment}
	Let
	\(
	\mu^\lambda_{P,\lambda}:=\mu^\lambda_{P,\Gamma_\lambda}
	\)
	be the lower symbol of the localized interacting Gibbs state. Then
	\begin{equation}\label{eq:husimi-second-moment-log}
		\int_{PH}\|u\|^2\,\d\mu^\lambda_{P,\lambda}(u)
		\le C(1+|\log\lambda|).
	\end{equation}
\end{lemma}

\begin{proof}
	By Theorem~\ref{thm:quant deF},
	\[
	\int_{PH}\|u\|^2\,\d\mu^\lambda_{P,\lambda}(u)
	=
	\lambda\Tr_{\mathfrak F(PH)}(\cN_P\Gamma_{\lambda,P})
	+
	\lambda\Tr(P).
	\]
	Since localization preserves the expectation of observables supported on $PH$,
	\[
	\Tr_{\mathfrak F(PH)}(\cN_P\Gamma_{\lambda,P})
	=
	\Tr_{\mathfrak F(H)}(\cN_P\Gamma_\lambda)
	\le
	\Tr_{\mathfrak F(H)}(\cN\Gamma_\lambda).
	\]
	Hence, by Proposition~\ref{prop:number-moments-log},
	\[
	\int_{PH}\|u\|^2\,\d\mu^\lambda_{P,\lambda}(u)
	\le
	C(1+|\log\lambda|)+\lambda\Tr(P).
	\]
	Since
	\[
	\Tr(P)=\#\{p\in\mathbb Z^2:\ |p|^2+1\le\Lambda^2\}\le C\Lambda^2,
	\]
	we have $\lambda\Tr(P)\le C\lambda\Lambda^2=o(1)$ for $\Lambda=\lambda^{-\alpha}$ with $0<\alpha<1/4$. This proves \eqref{eq:husimi-second-moment-log}.
\end{proof}

\begin{proposition}[Free-energy lower bound modulo high-frequency remainders]
	\label{prop:free-energy-lower-modulo-HF}
	Let $\Gamma_\lambda$ be the interacting Gibbs state. Then
	\begin{align}
		-\log \frac{Z_\lambda^{\mathrm{re}}}{Z_0}
	\ge
		\mathcal{H}_{\mathrm{cl}}(\mu^\lambda_{P,\lambda},\mu_{0,P})
		+
		\int_{PH} D_P(u)\,\d\mu^\lambda_{P,\lambda}(u)
		+
		HF_{\neq0}(\Gamma_\lambda)
		+
		HF_0(\Gamma_\lambda)
		-
		R_{\lambda,\Lambda},
		\label{eq:free-energy-lower-modulo-HF}
	\end{align}
	where $\mu^\lambda_{P,\lambda}:=\mu^\lambda_{P,\Gamma_\lambda}$ and
	\begin{align}
		R_{\lambda,\Lambda}
		&\le
	C\lambda \Lambda^4|\log\lambda|.
		\label{eq:R-lambda-Lambda}
	\end{align}
%
\end{proposition}

\begin{proof}
	By the Gibbs variational principle,
	\[
	-\log \frac{Z_\lambda^{\mathrm{re}}}{Z_0}
	=
	\mathcal{H}(\Gamma_\lambda,\Gamma_0)
	+
	\Tr_{\mathfrak F(H)}\!\big[W^{\mathrm{re}}\Gamma_\lambda\big].
	\]
	Use the exact localization identity \eqref{eq:exact-localization} and monotonicity of relative entropy:
	\[
	\mathcal{H}(\Gamma_\lambda,\Gamma_0)\ge \mathcal{H}(\Gamma_{\lambda,P},(\Gamma_0)_P).
	\]
	Then apply the Berezin-Lieb inequality
	\[
	\mathcal{H}(\Gamma_{\lambda,P},(\Gamma_0)_P)
	\ge
	\mathcal{H}_{\mathrm{cl}}(\mu^\lambda_{P,\lambda},\mu^\lambda_{P,0}),
	\]
	followed by Lemma \ref{lem:pointwise-comparison}, to obtain
	\[
	\mathcal{H}(\Gamma_\lambda,\Gamma_0)
	\ge
	\mathcal{H}_{\mathrm{cl}}(\mu^\lambda_{P,\lambda},\mu_{0,P})
	-
	C\lambda \Lambda^4
	\int_{PH}(1+\|u\|^2)\,\d\mu^\lambda_{P,\lambda}(u).
	\]
	Next, Proposition \ref{prop:localized-interaction-vs-DP} yields
	\[
	\Tr_{\mathfrak F(PH)}\!\big[W_P^{\mathrm{re}}\Gamma_{\lambda,P}\big]
	=
	\int_{PH} D_P(u)\,\d\mu^\lambda_{P,\lambda}(u)
	+
	\mathcal E_{P,\lambda}(\Gamma_\lambda).
	\]
	Combining these with \eqref{eq:exact-localization} gives
	\eqref{eq:free-energy-lower-modulo-HF} and \eqref{eq:R-lambda-Lambda}.


\end{proof}

\begin{theorem}[Free-energy lower bound by the classical partition function]
	\label{thm:lower-bound-classical-final}
{\rmk{Fix $\frac32<\beta\le2$, choose $\eta\in(2-\beta,\beta-1)$, then choose $\alpha\in(0,1/4)$ sufficiently small and $\alpha_1>\frac12$ sufficiently close to $\frac12$ so that $\mathfrak r^{\rm tail}_{\beta,\lambda}\to0$. Under this choice one has}}
	\begin{equation}\label{eq:lower-bound-classical-final}
		-\log \frac{Z_\lambda^{\mathrm{re}}}{Z_0}
		\ge
		-\log \int_{PH} e^{-D_P(u)}\,\d\mu_{0,P}(u)
		-{\rmk{C_\beta\Big(\mathfrak r^{\rm tail}_{\beta,\lambda}+(\lambda^{\frac\alpha4}+\lambda^{\frac{\alpha(\eta-2+\beta)}{2}})(1+|\log\lambda|^3)\Big).}}
	\end{equation}
\end{theorem}

\begin{proof}
	
	By Proposition~\ref{prop:free-energy-lower-modulo-HF},  Theorem \ref{cor:HF0-vanishes},  Theorem \ref{prop:shell-vanishing-local}, Theorem \ref{cor:tail-lower-bound}
	\[
	-\log \frac{Z_\lambda^{\mathrm{re}}}{Z_0}
	\ge
	\mathcal{H}_{\mathrm{cl}}(\mu^\lambda_{P,\lambda},\mu_{0,P})
	+
	\int_{PH} D_P(u)\,\d\mu^\lambda_{P,\lambda}(u)-{\rmk{C_\beta\Big(\mathfrak r^{\rm tail}_{\beta,\lambda}+(\lambda^{\frac\alpha4}+\lambda^{\frac{\alpha(\eta-2+\beta)}{2}})(1+|\log\lambda|^3)\Big).}}
	\]
	Apply the classical variational principle
	\[
	\mathcal{H}_{\mathrm{cl}}(\nu,\mu_{0,P})+\int D_P\,\d\nu
	\ge
	-\log\int_{PH} e^{-D_P(u)}\,\d\mu_{0,P}(u)
	\]
	to the probability measure $\nu=\mu^\lambda_{P,\lambda}$ implies the result.
\end{proof}

\subsection{Free-energy upper bound}\label{sec:upper}

In this subsection we prove the matching upper bound by constructing a trial state that is interacting on the low-frequency sector and free on the high-frequency sector. The key point is that this ansatz is compatible with the localization of the Hamiltonian, so the high-frequency remainders can be estimated directly and shown to vanish. For the truncated low-frequency problem we then use coherent states and the Peierls--Bogoliubov inequality to compare the quantum partition function with the classical partition function associated with $D_P$.

	Let
\[
\Gamma^{\rm re}_{P,\lambda}
:=
\frac{\exp(-\lambda \dG(Ph)-W^{\rm re}_P)}
{\Tr_{\gF(PH)}\!\big[\exp(-\lambda \dG(Ph)-W^{\rm re}_P)\big]},
\]
where $W^{\rm re}_P$ is the projected renormalized interaction defined in
\eqref{eq:WreP-def}. Define the trial state on $\mathfrak F(H)$ by
\[
\widetilde\Gamma_\lambda
:=
U^*\bigl(\Gamma^{\rm re}_{P,\lambda}\otimes (\Gamma_0)_Q\bigr)U,
\]
with $U:\mathfrak F(H)\to \mathfrak F(PH)\otimes \mathfrak F(QH)$ the localization unitary.

\begin{lemma}[Vanishing of the high-frequency remainders for the trial state]
	\label{lem:upper}
It holds that
		\begin{equation}
			\label{eq:trial-HF-explicit}
			|HF_{\neq0}(\widetilde\Gamma_\lambda)|+|HF_0(\widetilde\Gamma_\lambda)|
			\le
			{\color{red}C_\beta(1+|\log\lambda|)^2\Lambda^{-\beta}.}
		\end{equation}
\end{lemma}

\begin{proof}
	We write
	\[
	U\widetilde\Gamma_\lambda U^*=\omega_P\otimes \omega_Q,
	\qquad
	\omega_P:=\Gamma^{\rm re}_{P,\lambda},
	\qquad
	\omega_Q:=(\Gamma_0)_Q.
	\]
		Both factors commute with their corresponding number operators. Moreover $\omega_Q$ is the free
		Gibbs state on $\gF(QH)$. Set
		\[
		\mathcal K:=\{p\in\mathbb Z^2:\ h(p)\le \Lambda^2\}.
		\]
		Then, in the Fourier basis and for $r,s\notin \mathcal K$,
		\[
		\langle a_r^*a_s\rangle_{\omega_Q}=n_r\delta_{r,s},
		\qquad
		n_r:=\frac1{e^{\lambda h(r)}-1}.
		\]
		For $p\in \mathcal K$ let
		\[
		m_p:=\langle a_p^*a_p\rangle_{\omega_P}.
		\]

	\medskip
	\noindent
	{\bf Step 1: the zero-mode remainder.}
	Recall from \eqref{eq:HF-zero} that
	\[
	HF_0(\Gamma)
	=
	\frac{\lambda^2}{2(2\pi)^2}
	\Tr\!\Big[\Big(
	2(\cN_P-N_{0,P})(\cN_Q-N_{0,Q})+(\cN_Q-N_{0,Q})^2
	\Big)\Gamma\Big].
	\]
	Hence, using the tensor-product structure of $\widetilde\Gamma_\lambda$,
	\[
	\Tr\!\big[(\cN_P-N_{0,P})(\cN_Q-N_{0,Q})\widetilde\Gamma_\lambda\big]
	=
	\Tr_{\gF(PH)}[(\cN_P-N_{0,P})\omega_P]\,
	\Tr_{\gF(QH)}[(\cN_Q-N_{0,Q})\omega_Q]=0.
	\]
	Therefore the mixed term
	vanishes exactly, and
	\[
	HF_0(\widetilde\Gamma_\lambda)
	=
	\frac{\lambda^2}{2(2\pi)^2}
	\Tr_{\gF(QH)}\!\big[(\cN_Q-N_{0,Q})^2\omega_Q\big].
	\]
	Since $\omega_Q$ is quasi-free,
	\[
	\Tr_{\gF(QH)}\!\big[(\cN_Q-N_{0,Q})^2\omega_Q\big]
	=
	\sum_{p\notin \mathcal K} n_p(1+n_p).
	\]
	Using
	\[
	\lambda^2 n_p(1+n_p)\le \frac{C}{h(p)^2},
	\]
	we infer
	\[
	|HF_0(\widetilde\Gamma_\lambda)|
	\le
	C\sum_{p\notin \mathcal K}\frac1{h(p)^2}
	\le
	C\Lambda^{-2}.
	\]

	\medskip
	\noindent
	{\bf Step 2: blockwise simplification of $HF_{\neq0}$.}
	Recall that
	\[
	HF_{\neq0}(\Gamma)
	=
	\frac{\lambda^2}{2}
	\Tr_{H^{\otimes2}}\!\big[(V_{\neq0}-V_{\neq0,P})\Gamma^{(2)}\big],
	\]
	and that the block decomposition of Section~\ref{sec:high0} leaves only the blocks
	\[
	B2,\quad B4,\quad B9
	\]
	for the product state $\widetilde\Gamma_\lambda$. Indeed, since
	\[
	U\widetilde\Gamma_\lambda U^*=\omega_P\otimes \omega_Q,
	\]
	and both $\omega_P$ and $\omega_Q$ commute with their corresponding number operators, every expectation vanishes unless each sector contains the same number of creation and annihilation operators. Moreover, $\omega_P$ is diagonal in the Fourier basis on $PH$ because the truncated Hamiltonian $\lambda \dG(Ph)+W_P^{\rm re}$ is translation invariant; thus
	\[
	\langle a_r^*a_s\rangle_{\omega_P}=m_r\delta_{r,s},\qquad r,s\in\mathcal K,
	\]
	while $\omega_Q$ satisfies
	\[
	\langle a_r^*a_s\rangle_{\omega_Q}=n_r\delta_{r,s},\qquad r,s\notin\mathcal K,
	\]
	and all anomalous expectations $\langle a_ra_s\rangle$, $\langle a_r^*a_s^*\rangle$ are zero. Therefore the blocks
	\[
	A2,\ A3,\ A5,\ A6,\ B1,\ B3,\ B5,\ B6,\ B7,\ B8
	\]
	vanish immediately by gauge invariance, since in one of the two factors they leave an unequal number of creators and annihilators. The pure low-frequency block $A1$ has already been removed by subtracting $V_{\neq 0,P}$. For the remaining mixed blocks $A4$ and $A7$, the monomial factors into diagonal two-point functions such as
	\[
	\langle a_{p+k}^*a_p\rangle_{\omega_P}\,\langle a_{q-k}^*a_q\rangle_{\omega_Q},
	\]
	which vanish because $k\neq 0$ and both one-body density matrices are diagonal in momentum. Hence only $B2$, $B4$, and $B9$ can contribute.
	For $B2$ one has
	\[
	p\in \mathcal K,\quad p+k\notin \mathcal K,\quad q\notin \mathcal K,\quad q-k\in \mathcal K,
	\]
	and the expectation factorizes as
	\[
	\langle a_{p+k}^*a_{q-k}^*a_q a_p\rangle_{\widetilde\Gamma_\lambda}
	=
	\langle a_{q-k}^*a_p\rangle_{\omega_P}
	\langle a_{p+k}^*a_q\rangle_{\omega_Q}.
	\]
	Since $\omega_Q$ is diagonal in Fourier space,
	\[
	\langle a_{p+k}^*a_q\rangle_{\omega_Q}=n_{p+k}\delta_{q,p+k},
	\]
	so the $q$-sum collapses and $B2$ contributes
	\[
	\sum_{\substack{p\in \mathcal K\\ p+k\notin\mathcal K}} m_p n_{p+k}.
	\]
	Similarly, for $B4$ one gets
	\[
	\sum_{\substack{q\in \mathcal K\\ q-k\notin \mathcal K}} m_q n_{q-k}.
	\]
	Finally, $B9$ lies entirely in the $Q$-sector and Wick's rule gives the contribution
	\[
	\sum_{\substack{p\notin \mathcal K\\ p+k\notin \mathcal K}} n_p n_{p+k}.
	\]
	Collecting the three surviving pieces, we obtain the exact identity
	\begin{equation}
		\label{eq:upper-HF-nonzero-explicit}
		HF_{\neq0}(\widetilde\Gamma_\lambda)
		=
		\frac{\lambda^2}{2}
		\sum_{k\neq0}\hvb(k)
		\left[
		2\sum_{\substack{p\in \mathcal K\\ p+k\notin \mathcal K}} m_p n_{p+k}
		+
		\sum_{\substack{p\notin \mathcal K\\ p+k\notin \mathcal K}} n_p n_{p+k}
		\right].
	\end{equation}

	\medskip
	\noindent
	{\bf Step 3: the pure $Q$-$Q$ term.}
	Using $\lambda n_p\le h(p)^{-1}$, we obtain
	\[
	\frac{\lambda^2}{2}
	\sum_{k\neq0}\hvb(k)
	\sum_{\substack{p\notin \mathcal K\\ p+k\notin \mathcal K}} n_p n_{p+k}
	\le
	C
	\sum_{k\neq0}\hvb(k)
	\sum_{\substack{p\notin \mathcal K\\ p+k\notin \mathcal K}}
	\frac1{h(p)h(p+k)}.
	\]
	Split the $k$-sum into $|k|\le \Lambda/2$ and $|k|>\Lambda/2$.
	If $|k|\le \Lambda/2$, then by the arithmetic--geometric mean inequality,
	\[
	\sum_{\substack{p\notin \mathcal K\\ p+k\notin \mathcal K}}
	\frac1{h(p)h(p+k)}
	\le
	\sum_{p\notin \mathcal K}\frac1{h(p)^2}
	\le C\Lambda^{-2},
	\]
	and therefore
	\[
	\sum_{0<|k|\le \Lambda/2}\hvb(k)
	\sum_{\substack{p\notin\mathcal K\\ p+k\notin \mathcal  K}}
	\frac1{h(p)h(p+k)}
	\le C\Lambda^{-2}\sum_{|k|\le \Lambda/2}\hvb(k)
	\le {\color{red}C_\beta(1+\log\Lambda)\Lambda^{-\beta}.}
	\]
	If $|k|>\Lambda/2$, then we use the convolution bound from Lemma~\ref{lem:sum}:
	\[
	\sum_{p\in\mathbb Z^2}\frac1{h(p)h(p+k)}
	\le C\frac{\log(2+|k|)}{1+|k|^2}.
	\]
	Hence
	\[
	\sum_{|k|>\Lambda/2}\hvb(k)
	\sum_{\substack{p\notin \mathcal K\\ p+k\notin \mathcal K}}
	\frac1{h(p)h(p+k)}
	\le
	C\sum_{|k|>\Lambda/2}
	\frac{\log(2+|k|)}{\langle k\rangle^{2+\beta}}
	\le {\color{red}C_\beta(1+\log\Lambda)\Lambda^{-\beta}.}
	\]
	Altogether,
	\begin{equation}
		\label{eq:trial-HF-QQ-explicit}
			\frac{\lambda^2}{2}
		\sum_{k\neq0}\hvb(k)
		\sum_{\substack{p\notin \mathcal K\\ p+k\notin \mathcal K}} n_p n_{p+k}
		\le {\color{red}C_\beta(1+\log\Lambda)\Lambda^{-\beta}.}
	\end{equation}

	\medskip
	\noindent
	{\bf Step 4: the mixed $P$/$Q$ term.}
	Define, for $p\in \mathcal K$,
	\[
	T_p^{(\Lambda)}:=
	\sum_{r\notin \mathcal K}\hvb(r-p)\,\frac1{h(r)}.
	\]
	Then by $\lambda n_r\le h(r)^{-1}$,
	\[
	\lambda^2
	\sum_{k\neq0}\hvb(k)
	\sum_{\substack{p\in \mathcal K\\ p+k\notin \mathcal K}} m_p n_{p+k}
	\le
	\lambda\sum_{p\in \mathcal K}m_p\, T_p^{(\Lambda)}.
	\]
We claim
	\[
	{\color{red}\sup_{p\in \mathcal K}T_p^{(\Lambda)}\le C_\beta\frac{1+\log\Lambda}{\Lambda^\beta}.}
	\]
	Indeed, fix $p\in\mathcal K$, so that $|p|\le \Lambda$. Split
	\[
	T_p^{(\Lambda)}
	=
	\sum_{\substack{r\notin\mathcal K \\ |r|\le 2\Lambda}}
	\hvb(r-p)\frac1{h(r)}
	+
	\sum_{|r|>2\Lambda}
	\hvb(r-p)\frac1{h(r)}
	=:T_{p,\mathrm{near}}^{(\Lambda)}+T_{p,\mathrm{far}}^{(\Lambda)}.
	\]
	For the near part, $r\notin\mathcal K$ implies $h(r)>\Lambda^2$, hence $h(r)^{-1}\le \Lambda^{-2}$, and
	$|r-p|\le |r|+|p|\le 3\Lambda$. Therefore
	\[
	T_{p,\mathrm{near}}^{(\Lambda)}
	\le
	\frac1{\Lambda^2}\sum_{|s|\le 3\Lambda}\hvb(s)
	\le
	{\color{red}C_\beta(1+\log\Lambda)\Lambda^{-\beta}.}
	\]
	For the far part, if $|r|>2\Lambda$ then $|r-p|\ge |r|-|p|\ge |r|/2$, so
	\[
	\hvb(r-p)\frac1{h(r)}
	\le
	\frac{C_\beta}{\langle r\rangle^{2+\beta}}.
	\]
	Hence
	\[
	T_{p,\mathrm{far}}^{(\Lambda)}
	\le
	C_\beta\sum_{|r|>2\Lambda}\langle r\rangle^{-2-\beta}
	\le
	{\color{red}C_\beta\Lambda^{-\beta}.}
	\]
	Combining the two bounds proves the claimed estimate. Therefore
	\[
	\lambda^2
	\sum_{k\neq0}\hvb(k)
	\sum_{\substack{p\in \mathcal K\\ p+k\notin \mathcal K}} m_p n_{p+k}
	\le
	C_\beta\Bigl(\lambda\sum_{p\in \mathcal K}m_p\Bigr)\frac{1+\log\Lambda}{\Lambda^\beta}.
	\]
	Applying the same logarithmic number bound as in Proposition~\ref{prop:number-moments-log} to the truncated Gibbs state on $\gF(PH)$, we have
	\[
	\lambda\sum_{p\in \mathcal K}m_p
	=
	\lambda\,\Tr_{\gF(PH)}(\cN_P\Gamma^{\rm re}_{P,\lambda})
	\le
	C(1+|\log\lambda|).
	\]
	Hence
	\begin{equation}
		\label{eq:trial-HF-PQ-explicit}
		\lambda^2
		\sum_{k\neq0}\hvb(k)
		\sum_{\substack{p\in \mathcal K\\ p+k\notin \mathcal K}} m_p n_{p+k}
		\le
		{\color{red}C_\beta(1+|\log\lambda|)(1+\log\Lambda)\Lambda^{-\beta}.}
	\end{equation}
	Since $\log\Lambda=\alpha|\log\lambda|$, combining \eqref{eq:upper-HF-nonzero-explicit}, \eqref{eq:trial-HF-QQ-explicit},
	and \eqref{eq:trial-HF-PQ-explicit}, we conclude that
	\[
	|HF_{\neq0}(\widetilde\Gamma_\lambda)|
	\le
	{\color{red}C_\beta(1+|\log\lambda|)^2\Lambda^{-\beta}.}
	\]
	Together with the estimate for $HF_0(\widetilde\Gamma_\lambda)$ proved in Step~1, this yields
	\eqref{eq:trial-HF-explicit}.
\end{proof}

\begin{proposition}[Free-energy upper bound via the trial state $\widetilde\Gamma_\lambda$]
	\label{prop:free-energy-upper-final}	{\color{red}Fix $\frac32<\beta\le2$.} It holds that
	\begin{equation}\label{eq:upper-bound-classical-final}
		-\log\frac{Z^{\rm re}_\lambda}{Z_0}
		\le
		-\log \int_{PH} e^{-D_P(u)}\,\d\mu_{0,P}(u)+{\color{red}C_\beta(1+|\log\lambda|)^2(\lambda^{\alpha\beta}+\lambda^{1-4\alpha})}
		.
	\end{equation}
\end{proposition}

\begin{proof}
	By the Gibbs variational principle,
	\[
	-
	\log\frac{Z^{\rm re}_\lambda}{Z_0}
	\le
	\mathcal{H}(\widetilde\Gamma_\lambda,\Gamma_0)
	+
	\Tr_{\gF(H)}(W^{\rm re}\widetilde\Gamma_\lambda).
	\]
	Since the free Gibbs state factorizes under the Fock-space decomposition,
	\[
	\Gamma_0
	=
	U^*\bigl((\Gamma_0)_P\otimes (\Gamma_0)_Q\bigr)U,
	\]
	and since relative entropy is additive for tensor products, we get
	\[
	\mathcal{H}(\widetilde\Gamma_\lambda,\Gamma_0)
	=
	\mathcal{H}(\Gamma^{\rm re}_{P,\lambda},(\Gamma_0)_P).
	\]
	Next, applying the exact localization identity \eqref{eq:exact-localization} to the state
	$\widetilde\Gamma_\lambda$, we have
	\[
	\Tr_{\gF(H)}(W^{\rm re}\widetilde\Gamma_\lambda)
	=
	\Tr_{\gF(PH)}(W^{\rm re}_P\Gamma^{\rm re}_{P,\lambda})
	+
	HF_{\neq 0}(\widetilde\Gamma_\lambda)
	+
	HF_0(\widetilde\Gamma_\lambda).
	\]
	Therefore,
	\[
	-
	\log\frac{Z^{\rm re}_\lambda}{Z_0}
	\le
	\mathcal{H}(\Gamma^{\rm re}_{P,\lambda},(\Gamma_0)_P)
	+
	\Tr_{\gF(PH)}(W^{\rm re}_P\Gamma^{\rm re}_{P,\lambda})
	+
	|HF_{\neq 0}(\widetilde\Gamma_\lambda)|
	+
	|HF_0(\widetilde\Gamma_\lambda)|.
	\]
		By Lemma~\ref{lem:upper},
		\[
		|HF_{\neq 0}(\widetilde\Gamma_\lambda)|+|HF_0(\widetilde\Gamma_\lambda)|
		\le {\color{red}C_\beta(1+|\log\lambda|)^2\Lambda^{-\beta}.}
		\]
	Since $\Gamma^{\rm re}_{P,\lambda}$ is the Gibbs minimizer on $\gF(PH)$ for the truncated Hamiltonian
	$\lambda \dG(Ph)+W^{\rm re}_P$, we have the exact identity
	\[
	\mathcal{H}(\Gamma^{\rm re}_{P,\lambda},(\Gamma_0)_P)
	+
	\Tr_{\gF(PH)}(W^{\rm re}_P\Gamma^{\rm re}_{P,\lambda})
	=
	-
	\log
	\frac{\Tr_{\gF(PH)} e^{-\lambda \dG(Ph)-W^{\rm re}_P}}
	{\Tr_{\gF(PH)} e^{-\lambda \dG(Ph)}}.
	\]
	Hence
	\begin{equation}
		\label{eq:free-energy-upper-reduced-partition}
		-
		\log\frac{Z^{\rm re}_\lambda}{Z_0}
		\le
		-
		\log
		\frac{\Tr_{\gF(PH)} e^{-\lambda \dG(Ph)-W^{\rm re}_P}}
		{\Tr_{\gF(PH)} e^{-\lambda \dG(Ph)}}
		+
		{\color{red}C_\beta(1+|\log\lambda|)^2\Lambda^{-\beta}.}
	\end{equation}
	
	\medskip
	\noindent
	{\bf Coherent-state lower bound for the truncated partition function.}
	Set
	\[
	A_{P,\lambda}:=\lambda \dG(Ph)+W^{\rm re}_P
	\qquad\text{on }\gF(PH),
	\]
	and let
	\[
	K:=\dim(PH)=\Tr(P).
	\]
	Using the coherent-state resolution of the identity on $\gF(PH)$,
	\[
	\1_{\gF(PH)}
	=
	(\lambda\pi)^{-K}
	\int_{PH}
	|W(u/\sqrt\lambda)\rangle\langle W(u/\sqrt\lambda)|\,\d u,
	\]
	we obtain
	\[
	\Tr_{\gF(PH)}(e^{-A_{P,\lambda}})
	=
	(\lambda\pi)^{-K}
	\int_{PH}
	\langle W(u/\sqrt\lambda),e^{-A_{P,\lambda}}W(u/\sqrt\lambda)\rangle\,\d u.
	\]
	By the Peierls--Bogoliubov inequality
	\[
	\langle \psi,e^{-A}\psi\rangle\ge e^{-\langle \psi,A\psi\rangle}
	\qquad (\|\psi\|=1),
	\]
	it follows that
	\[
	\Tr_{\gF(PH)}(e^{-A_{P,\lambda}})
	\ge
	(\lambda\pi)^{-K}
	\int_{PH}
	\exp\!\Big(
	-\langle W(u/\sqrt\lambda),A_{P,\lambda}W(u/\sqrt\lambda)\rangle
	\Big)\,\d u.
	\]
	For the coherent vector $W(u/\sqrt\lambda)$ one has
	\[
	\lambda\langle W(u/\sqrt\lambda),\dG(Ph)W(u/\sqrt\lambda)\rangle
	=
	\langle u,hu\rangle,
	\]
	and, by the explicit form of $W_P^{\rm re}$,
	\[
	\langle W(u/\sqrt\lambda),W_P^{\rm re}W(u/\sqrt\lambda)\rangle
	=
	\frac1{2}
	\sum_{k\neq0}\hvb(k)|\langle u,e_{k,P}u\rangle|^2
	+
	\frac{1}{2(2\pi)^2}
	\Big[(\|u\|^2-\lambda N_{0,P})^2+\lambda\|u\|^2\Big].
	\]
	Comparing the zero-mode centers exactly as in the proof of
	Proposition~\ref{prop:localized-interaction-vs-DP}, and using \eqref{eq:center-shift-bound}
	together with $(2\pi)^2\rho_{0,P}\lesssim \log\Lambda\le \Lambda^2$, we get
	\[
	\Big|
	\langle W(u/\sqrt\lambda),W_P^{\rm re}W(u/\sqrt\lambda)\rangle-D_P(u)
	\Big|
	\le
	C\lambda\Lambda^2(1+\|u\|^2).
	\]
	Therefore
	\[
	\Tr_{\gF(PH)}(e^{-A_{P,\lambda}})
	\ge
	e^{-C\lambda\Lambda^2}
	(\lambda\pi)^{-K}
	\int_{PH}
	e^{-\langle u,hu\rangle-D_P(u)-C\lambda\Lambda^2\|u\|^2}\,\d u.
	\]
	
	\medskip
	\noindent
	{\bf Division by the free partition function.}
	Let $\lambda_1,\dots,\lambda_K$ be the eigenvalues of $h_P:=PhP$ on $PH$. Then
	\[
	\Tr_{\gF(PH)}(e^{-\lambda \dG(Ph)})
	=
	\prod_{j=1}^K \frac1{1-e^{-\lambda\lambda_j}}.
	\]
	On the other hand, in the $h_P$-eigenbasis,
	\[
	\d\mu_{0,P}(u)
	=
	\pi^{-K}\prod_{j=1}^K \lambda_j e^{-\lambda_j |\alpha_j|^2}\,\d\alpha_j,
	\]
	so that
	\[
	(\lambda\pi)^{-K}e^{-\langle u,hu\rangle}\,\d u
	=
	\Big(\prod_{j=1}^K \frac1{\lambda\lambda_j}\Big)\,\d\mu_{0,P}(u).
	\]
	Hence
	\[
	\frac{\Tr_{\gF(PH)}(e^{-\lambda \dG(Ph)-W_P^{\rm re}})}
	{\Tr_{\gF(PH)}(e^{-\lambda \dG(Ph)})}
	\ge
	e^{-C\lambda\Lambda^2}
	\Big(\prod_{j=1}^K\frac{1-e^{-\lambda\lambda_j}}{\lambda\lambda_j}\Big)
	\int_{PH}
	e^{-D_P(u)-C\lambda\Lambda^2\|u\|^2}\,\d\mu_{0,P}(u).
	\]
	By \cite[Section 10]{LewNamRou-21} (see also \cite[(8.48)]{NZZ25}) we have
	\[\prod_{j=1}^K\frac{1-e^{-\lambda\lambda_j}}{\lambda\lambda_j}\geq 1-C\lambda\Lambda^4.\]
	Moreover, the finite-dimensional classical partition
	function is stable under the added mass term. Writing
	\[
	\int_{PH} e^{-D_P(u)-C\lambda\Lambda^2\|u\|^2}\,\d\mu_{0,P}(u)
	=
	Z_P\int_{PH} e^{-C\lambda\Lambda^2\|u\|^2}\,\d\nu_P(u),
	\qquad
	Z_P:=\int_{PH} e^{-D_P(u)}\,\d\mu_{0,P}(u),
	\]
	and using Jensen's inequality together with
	\[
	\int_{PH}\|u\|^2\,\d\nu_P(u)
	\le Z_P^{-1}\int_{PH}\|u\|^2\,\d\mu_{0,P}(u)
	\le C(1+\log\Lambda),
	\]
	we obtain
	\[
	-C\lambda\Lambda^2(1+\log\Lambda)
	\le
	\log\int_{PH} e^{-C\lambda\Lambda^2\|u\|^2}\,\d\nu_P(u)
	\le 0.
	\]
	Hence
	\[
	\Big|
	\log\int_{PH} e^{-D_P(u)-C\lambda\Lambda^2\|u\|^2}\,\d\mu_{0,P}(u)
	-
	\log\int_{PH} e^{-D_P(u)}\,\d\mu_{0,P}(u)
	\Big|
	\le C\lambda\Lambda^2(1+\log\Lambda).
	\]
	Therefore,
	\[
	-
	\log
	\frac{\Tr_{\gF(PH)}(e^{-\lambda \dG(Ph)-W_P^{\rm re}})}
	{\Tr_{\gF(PH)}(e^{-\lambda \dG(Ph)})}
	\le
	-
	\log\int_{PH} e^{-D_P(u)}\,\d\mu_{0,P}(u)
	+
	C\lambda\Lambda^4
	+
	C\lambda\Lambda^2(1+\log\Lambda).
	\]
	Combining this with \eqref{eq:free-energy-upper-reduced-partition} yields
 \eqref{eq:upper-bound-classical-final}.
\end{proof}

\begin{lemma}\label{lem:class}It holds that
	\begin{equation}
		\label{eq:HS-DP-to-D-L1}
		D_P(P\cdot)\longrightarrow D
		\qquad\text{in }L^1(\mu_0).
	\end{equation}
	and the following hold for every fixed $k\ge1$.
	One has
	\begin{equation}
		\label{eq:HS-DP-D-positive}
		D_P(u)\ge0\quad\text{for all }u\in PH,
		\qquad
		D(u)\ge0\quad\text{for }\mu_0\text{-a.e. }u.
	\end{equation}
	Consequently,
	\begin{equation}
		\label{eq:HS-gP-to-g-L1}
		g_P:=z_P^{-1}e^{-D_P(P\cdot)}
		\longrightarrow
		g:=z_{\rm cl}^{-1}e^{-D}
		\qquad\text{in }L^1(\mu_0),
	\end{equation}
	and
	\begin{equation}
		\label{eq:HS-nu-dominated-by-mu0}
		\d\nu(u)=g(u)\,\d\mu_0(u)\le z_{\rm cl}^{-1}\,\d\mu_0(u).
	\end{equation}
	\end{lemma}
	\begin{proof}
		We first have \eqref{eq:HS-DP-to-D-L1} by standard stochastic calculation (see e.g. \cite{NZZ25} and \cite{Bri22, OOT24}). More precisely, we have
		\begin{align}\label{con:rate}
			\|D_P(P\cdot)-D\|_{L^1(\mu_0)}\lesssim\Lambda^{-\kappa}\to0,
		\end{align}
		for $\kappa>0$ small enough.
		By definition of $D_P(u)$ from \eqref{eq:DP-def-new}, we have $D_P\ge0$. Since $D_P(P\cdot)\to D$ in $L^1(\mu_0)$, a subsequence converges to $D$ $\mu_0$-almost surely; hence $D\ge0$ $\mu_0$-almost surely as well. Because $x\mapsto e^{-x}$ is $1$-Lipschitz on $[0,\infty)$, we obtain
		\[
		\|e^{-D_P(P\cdot)}-e^{-D}\|_{L^1(\mu_0)}
		\le
		\|D_P(P\cdot)-D\|_{L^1(\mu_0)}
		\to0,
		\]
	which also implies that
		\[
		z_P:=\int e^{-D_P(Pu)}\,\d\mu_0(u)\longrightarrow z_{\rm cl}:=\int e^{-D(u)}\,\d\mu_0(u),
		\]
		this proves \eqref{eq:HS-gP-to-g-L1}. Since $0\le e^{-D}\le1$, we get \eqref{eq:HS-nu-dominated-by-mu0}.
	\end{proof}

\begin{theorem}[Convergence of the relative free energy]
	\label{thm:free-energy-convergence-final}
	It holds that
	\begin{equation}\label{eq:free-energy-convergence-final}
		-\log \frac{Z_\lambda^{\mathrm{re}}}{Z_0}
		\longrightarrow
		-\log z_{\mathrm{cl}}
		\qquad (\lambda\downarrow0),
	\end{equation}
	where
	\[
	z_{\mathrm{cl}}:=\int e^{-D(u)}\,\d\mu_0(u)
	\]
	{\color{red}is the partition function of the limiting fractional-Bessel Gibbs measure.}
\end{theorem}

\begin{proof}
Hence, one has the convergence of the finite-dimensional
	classical partition functions from Lemma \ref{lem:class}, which combined with Proposition~\ref{prop:free-energy-upper-final},
	Theorem~\ref{thm:lower-bound-classical-final} implies the result.
\end{proof}

\section{Convergence of density matrices}\label{sec:den}

This section turns the free-energy analysis into convergence of reduced density matrices. We proceed
in two stages. In Subsection~\ref{sec:con:density} we prove scalar convergence against finite
Fourier test functions by combining the weighted convergence of the low-frequency lower symbols with
the quantitative de Finetti formula. In the second part we upgrade this to Hilbert--Schmidt
convergence of the full rescaled reduced density matrices: we compare the interacting state with a
tensor product of its low-frequency localization and the free high-frequency Gibbs state, show that
the high-frequency contribution disappears after the $\lambda^k$ scaling, and then pass to the
classical limit through the finite-dimensional Gibbs measures.

\subsection{Weak convergence of the density matrix}\label{sec:con:density}

The density-matrix convergence is obtained by transferring the problem from quantum correlations to the lower symbols introduced in the free-energy analysis. The first step is to prove that the lower symbols of the localized Gibbs states converge to the finite-dimensional Gibbs measures, not only in total variation but also with the polynomial weights needed for moments of $|\langle \phi,u\rangle|^{2n}$. This uses the entropy bounds from Subsections~\ref{sec:lower}--\ref{sec:upper} together with Pinsker's inequality and the logarithmic moment estimates established earlier. We then combine this weighted convergence with the quantitative de Finetti formula and the classical finite-dimensional convergence  to pass from lower-symbol moments to the quantum reduced density matrices, which yields Theorem~\ref{thm:density-matrix-convergence-final}.

\begin{proposition}[Weighted convergence of the lower symbol]
	\label{prop:weighted-lower-symbol}
For the parameter as in Section \ref{sec:free-energy-density}	it holds that
	\begin{equation}
		\label{eq:l1-verified-weighted}
		\|\mu^\lambda_{P,\lambda}-\nu_P\|_{L^1(PH)}^2 \le
		{\color{red}C_\beta\Bigl((1+|\log\lambda|)^3(\lambda^{\frac\alpha4}+\lambda^{\frac{\alpha(\eta-2+\beta)}{2}})+\mathfrak r^{\rm tail}_{\beta,\lambda}\Bigr).}
	\end{equation}
	Consequently, for every $\phi\in L^2(\T^2)$ with finitely supported Fourier transform and every
	$n\ge1$,
	\begin{equation}
		\label{eq:weighted-l1-lower-symbol}
		\int_{PH} |\langle \phi,u\rangle|^{2n}
		\,|\d\mu^\lambda_{P,\lambda}-\d\nu_P|(u)
		\le
		{\color{red}C_{\phi,n}\,(1+|\log\lambda|)^n\Bigl((1+|\log\lambda|)^3(\lambda^{\frac\alpha4}+\lambda^{\frac{\alpha(\eta-2+\beta)}{2}})+\mathfrak r^{\rm tail}_{\beta,\lambda}\Bigr)^{1/4}.}
	\end{equation}
	In particular, the left side of \eqref{eq:weighted-l1-lower-symbol} tends to $0$.
\end{proposition}

\begin{proof}
	Set
	\[
	F_\lambda:=-\log\frac{Z_\lambda^{\rm re}}{Z_0},\qquad  Z_P=\int e^{-D_P(u)}\d \mu_{0,P}(u).
	\]
	By \eqref{eq:lower-bound-classical-final},
	\[
	F_\lambda
	\ge
	\mathcal{H}_{\rm cl}(\mu^\lambda_{P,\lambda},\mu_{0,P})
	+
	\int_{PH} D_P(u)\,\d\mu^\lambda_{P,\lambda}(u)
	-
{\color{red}C_\beta\Bigl((1+|\log\lambda|)^3(\lambda^{\frac\alpha4}+\lambda^{\frac{\alpha(\eta-2+\beta)}{2}})+\mathfrak r^{\rm tail}_{\beta,\lambda}\Bigr).}
	\]
	Using the identity
	\[
	\mathcal{H}_{\rm cl}(\mu^\lambda_{P,\lambda},\nu_P)
	=
	\mathcal{H}_{\rm cl}(\mu^\lambda_{P,\lambda},\mu_{0,P})
	+
	\int_{PH} D_P(u)\,\d\mu^\lambda_{P,\lambda}(u)
	+
	\log Z_P,
	\]
	we infer
	\[
	\mathcal{H}_{\rm cl}(\mu^\lambda_{P,\lambda},\nu_P)
	\le
F_\lambda+\log Z_P+{\color{red}C_\beta\Bigl((1+|\log\lambda|)^3(\lambda^{\frac\alpha4}+\lambda^{\frac{\alpha(\eta-2+\beta)}{2}})+\mathfrak r^{\rm tail}_{\beta,\lambda}\Bigr).}
	\]
	On the other hand, \eqref{eq:upper-bound-classical-final} gives
	\[
	F_\lambda\le -\log Z_P+{\color{red}C_\beta}(1+|\log\lambda|)^2\Bigl(
	\lambda^{1-4\alpha}
	+
	{\color{red}\lambda^{\alpha\beta}}
	\Bigr).
	\]
	Therefore
	\[
	\mathcal{H}_{\rm cl}(\mu^\lambda_{P,\lambda},\nu_P)
	\le
{\color{red}C_\beta\Bigl((1+|\log\lambda|)^3(\lambda^{\frac\alpha4}+\lambda^{\frac{\alpha(\eta-2+\beta)}{2}})+\mathfrak r^{\rm tail}_{\beta,\lambda}\Bigr).}
	\]
	By Pinsker's inequality,
	\[
	\|\mu^\lambda_{P,\lambda}-\nu_P\|_{L^1(PH)}^2
	\le
	2\,\mathcal{H}_{\rm cl}(\mu^\lambda_{P,\lambda},\nu_P),
	\]
	which proves \eqref{eq:l1-verified-weighted}.
	
	It remains to prove \eqref{eq:weighted-l1-lower-symbol}. Fix $\phi$ with finitely supported Fourier
	transform and $n\ge1$. Since
	\[
	P=1_{\{h\le \Lambda^2\}},\qquad \Lambda=\lambda^{-\alpha}\to\infty,
	\]
	we have $P\phi=\phi$ for all sufficiently small $\lambda$.
	
	\medskip
	\noindent
	{\bf Step 1: a $4n$-moment bound for $\mu^\lambda_{P,\lambda}$.}
	Let $\Gamma_{\lambda,P}$ be the localization of $\Gamma_\lambda$ on $\gF(PH)$. Applying the
	quantitative de Finetti identity on $\gF(PH)$ at order $2n$ and testing against
	$\phi^{\otimes 2n}$, we get
	\[
	\int_{PH} |\langle \phi,u\rangle|^{4n}\,\d\mu^\lambda_{P,\lambda}(u)
	\le
	C_{\phi,n}\sum_{\ell=0}^{2n}
	\lambda^\ell
	\big\langle \phi^{\otimes \ell},(\Gamma_{\lambda,P})^{(\ell)}\phi^{\otimes \ell}\big\rangle
	\le
	C_{\phi,n}\sum_{\ell=0}^{2n}\|\phi\|_{L^2}^{2\ell}\,\lambda^\ell\Tr(\cN^\ell\Gamma_\lambda).
	\]
	Hence Proposition~\ref{prop:number-moments-log} yields
	\[
	\int_{PH} |\langle \phi,u\rangle|^{4n}\,\d\mu^\lambda_{P,\lambda}(u)
	\le
	C_{\phi,n}(1+|\log\lambda|)^{2n}.
	\]
	
	\medskip
	\noindent
	{\bf Step 2: a uniform $4n$-moment bound for $\nu_P$.}
	By definition,
	\[
	D_P(u)=\frac1{2}\sum_{k\neq0}\hvb(k)|\langle u,e_{k,P}u\rangle|^2
	+\frac{1}{2(2\pi)^2}\Big(\|u\|^2-(2\pi)^2\rho_{0,P}\Big)^2\ge0.
	\]
	Therefore $e^{-D_P(u)}\le1$, and since $Z_P\to z_{\rm cl}>0$ by Lemma \ref{lem:class}, one has
	$Z_P\ge c>0$ for all sufficiently small $\lambda$. It follows that
	\[
	\int_{PH} |\langle \phi,u\rangle|^{4n}\,\d\nu_P(u)
	\le
	c^{-1}\int_{PH} |\langle \phi,u\rangle|^{4n}\,\d\mu_{0,P}(u)
	\le C_{\phi,n},
	\]
	because $P\phi=\phi$ and the law of $\langle\phi,u\rangle$ under $\mu_{0,P}$ is the centered
	complex Gaussian with variance $\langle \phi,h^{-1}\phi\rangle$.
	
	\medskip
	\noindent
	{\bf Step 3: weighted $L^1$ interpolation.}
	Set $f(u):=|\langle \phi,u\rangle|^{2n}$. By Cauchy--Schwarz with respect to the measure
	$|\d\mu^\lambda_{P,\lambda}-\d\nu_P|$,
	\begin{align*}
		\int_{PH} f(u)\,|\d\mu^\lambda_{P,\lambda}-\d\nu_P|(u)
		&\le
		\Big(
		\int_{PH} f(u)^2\,|\d\mu^\lambda_{P,\lambda}-\d\nu_P|(u)
		\Big)^{1/2}
		\|\mu^\lambda_{P,\lambda}-\nu_P\|_{L^1}^{1/2}
		\\&\le
		\Big(
		\int_{PH} f(u)^2\,\d\mu^\lambda_{P,\lambda}(u)
		+
		\int_{PH} f(u)^2\,\d\nu_P(u)
		\Big)^{1/2}
		\|\mu^\lambda_{P,\lambda}-\nu_P\|_{L^1}^{1/2}.
	\end{align*}
	Thus the result follows.
\end{proof}

\begin{theorem}[Scalar moment convergence against finite Fourier test functions]
	\label{thm:density-matrix-convergence-final}
	Let
	\[
	\d\nu(u):=z_{\mathrm{cl}}^{-1}e^{-D(u)}\,\d\mu_0(u)
	\]
	{\color{red}be the limiting fractional-Bessel Gibbs measure.} Then, for every $n\ge1$ and every
	$\phi\in L^2(\T^2)$ with finite Fourier support,
	\begin{equation}\label{eq:density-matrix-convergence-final}
		n!\lambda^n
		\big\langle \phi^{\otimes n},\Gamma_\lambda^{(n)}\phi^{\otimes n}\big\rangle
		\longrightarrow
		\int |\langle \phi,u\rangle|^{2n}\,\d\nu(u)
		\qquad (\lambda\downarrow0).
	\end{equation}
\end{theorem}

\begin{proof}
	Fix $n\ge1$ and $\phi\in L^2(\T^2)$ with finite Fourier support. Since
	\[
	P=1_{\{h\le \Lambda^2\}},\qquad \Lambda=\lambda^{-\alpha}\to\infty,
	\]
	we have $P\phi=\phi$ for all sufficiently small $\lambda$.
	
	By Proposition~\ref{prop:weighted-lower-symbol},
	\[
	\int_{PH} |\langle \phi,u\rangle|^{2n}
	|\d\mu^\lambda_{P,\lambda}-\d\nu_P|(u)
	\longrightarrow0.
	\]
	Moreover,
	\[
	\int_{PH} |\langle \phi,u\rangle|^{2n}\,\d\nu_P(u)
	\longrightarrow
	\int |\langle \phi,u\rangle|^{2n}\,\d\nu(u)
	\qquad \text{by Lemma \ref{lem:class}.}
	\]
	Hence
	\begin{equation}
		\label{eq:lower-symbol-moment-limit}
		\int_{PH} |\langle \phi,u\rangle|^{2n}\,\d\mu^\lambda_{P,\lambda}(u)
		\longrightarrow
		\int |\langle \phi,u\rangle|^{2n}\,\d\nu(u).
	\end{equation}
	Applying the quantitative de Finetti identity in Theorem~\ref{thm:quant deF}
	to $\Gamma_{\lambda,P}$ and using $P\phi=\phi$, we obtain
	\[
	\int_{PH} |\langle \phi,u\rangle|^{2n}\,\d\mu^\lambda_{P,\lambda}(u)
	=
	n!\lambda^n
	\big\langle \phi^{\otimes n},(\Gamma_{\lambda,P})^{(n)}\phi^{\otimes n}\big\rangle
	+R_{\lambda,n}(\phi)
	=
	n!\lambda^n
	\big\langle \phi^{\otimes n},\Gamma_\lambda^{(n)}\phi^{\otimes n}\big\rangle
	+R_{\lambda,n}(\phi),
	\]
	where by Proposition \ref{prop:number-moments-log}
	\begin{align*}
		|R_{\lambda,n}(\phi)|
		\le
		C_{n,\phi}
		\sum_{\ell=0}^{n-1}
		\lambda^{n-\ell}
		\lambda^\ell
		\big\langle \phi^{\otimes \ell},(\Gamma_{\lambda,P})^{(\ell)}\phi^{\otimes \ell}\big\rangle
		\\\le
		C_{n,\phi}\sum_{\ell=0}^{n-1}
		\lambda^{n-\ell}(1+|\log\lambda|)^\ell
		\le
		C_{n,\phi}\,\lambda(1+|\log\lambda|)^{n-1}
		\longrightarrow0.
	\end{align*}
	Combining this with \eqref{eq:lower-symbol-moment-limit} yields
	\eqref{eq:density-matrix-convergence-final}.
\end{proof}

\subsection{Hilbert--Schmidt convergence of the reduced density matrices}
\label{sec:HS-convergence}

In this section we strengthen Theorem~\ref{thm:density-matrix-convergence-final} from scalar moment convergence along finite Fourier test functions to operator convergence in the Hilbert--Schmidt topology.

The main point is that the localized state $\Gamma_{\lambda,P}$ only records the
low-frequency marginal of the full Gibbs state $\Gamma_\lambda$ and therefore
cannot be compared directly with $\Gamma_\lambda$ as a state on the same Fock space.
To recover the full reduced density matrices, one has to control both the high-frequency
sector and the mixed $P/Q$ contributions.

For this reason we introduce two auxiliary product states with different roles.
The first one,
\[
\widetilde\Gamma_\lambda := U^*(\Gamma^{\rm re}_{P,\lambda}\otimes \Gamma_{0,Q})U,
\]
is the variational trial state from the free-energy upper bound. Since its free energy
is already known to be asymptotically optimal, it is the natural object for the
state-level entropy comparison with $\Gamma_\lambda$.

The second one,
\[
\widehat\Gamma_\lambda := U^*(\Gamma_{\lambda,P}\otimes \Gamma_{0,Q})U,
\]
has the correct low-frequency marginal and an explicit free high-frequency factor.
Its advantage is that the reduced density matrices admit an exact binomial expansion
into low-frequency interacting pieces and high-frequency free pieces, and the latter
can be shown to vanish after the $\lambda^k$ scaling.

Thus the proof proceeds in four steps: first compare $\Gamma_\lambda$ with
$\widetilde\Gamma_\lambda$ in trace norm through the free-energy comparison;
next replace $\widetilde\Gamma_\lambda$ by $\widehat\Gamma_\lambda$ using the
closeness of $\Gamma^{\rm re}_{P,\lambda}$ and $\Gamma_{\lambda,P}$;
then remove the free high-frequency sector by exploiting the explicit decay of
$\lambda^m\Gamma_{0,Q}^{(m)}$; finally identify the remaining low-frequency block
with the finite-dimensional classical correlation operator via the quantitative
de Finetti theorem and the convergence of the lower symbols.

Throughout the section we write
\[
\|A\|_{\mathfrak S^1}:=\Tr|A|,
\qquad
\|A\|_{\mathfrak S^2}:=\bigl(\Tr(A^*A)\bigr)^{1/2},
\]
and we keep the notation
\[
P=\1_{\{h\le \Lambda^2\}},
\qquad
Q=1-P,
\qquad
\Lambda=\lambda^{-\alpha},
\qquad
0<\alpha<\frac14.
\]
We also set
\[
\Gamma_{\lambda,P}:=(\Gamma_\lambda)_P,
\qquad
\Gamma_{0,Q}:=(\Gamma_0)_Q.
\]
Finally, we introduce the quantitative error
\begin{equation}
	\label{eq:HS-main-epsilon-lambda}
	\varepsilon_\lambda
	:=
	{\color{red}(1+|\log\lambda|)^3(\lambda^{\frac\alpha4}+\lambda^{\frac{\alpha(\eta-2+\beta)}{2}})+\mathfrak r^{\rm tail}_{\beta,\lambda}},
\end{equation}
which tends to $0$ for the admissible choices of $\alpha$ and $\alpha_1$.

The following is the main result of this section.

\begin{theorem}
	\label{thm:HS-convergence-main} Let
	\[
	\d\nu(u)=z_{\rm cl}^{-1}e^{-D(u)}\,\d\mu_0(u)
	\]
	{\color{red}be the limiting fractional-Bessel Gibbs measure.}
	It holds that 	for every fixed $k\ge1$
	\begin{equation}
		\label{eq:HS-convergence-main}
		k!\lambda^k\Gamma_\lambda^{(k)}
		\longrightarrow
		\int |u^{\otimes k}\rangle\langle u^{\otimes k}|\,\d\nu(u)
		\qquad\text{in }\mathfrak S^2(\gH^{\otimes_s k})
		\qquad (\lambda\downarrow0).
	\end{equation}
\end{theorem}

{\color{red}The next result is the fractional-Bessel analogue of \cite[Lemma~5.3]{LewNamRou-21}. It packages all the classical facts needed at the end of the proof: positivity of the renormalized interaction, definition and domination of the limiting $k$-point correlation operator, and convergence of the classical cutoffs.}

\begin{proposition}\label{prop:HS-classical-cutoff-limit}
	For every fixed $k \ge 1$, the following hold.
	
	\begin{enumerate}
		\item The $k$-point correlation operator
		\begin{equation}
			\gamma^{(k)}_\nu := \int |u^{\otimes k}\rangle \langle u^{\otimes k}|\, d\nu(u)
		\end{equation}
		extends uniquely to a positive operator in $\mathfrak S^2(H^{\otimes_s k})$ and satisfies
		\begin{equation}
			0 \le \gamma^{(k)}_\nu \le z_{\mathrm{cl}}^{-1}\gamma^{(k)}_{\mu_0}
			= z_{\mathrm{cl}}^{-1} k! (h^{-1})^{\otimes_s k}.
		\end{equation}
		Similarly,
		\begin{equation}
			0 \le \gamma^{(k)}_{\nu_P} \le c_0^{-1} k! (h^{-1})^{\otimes_s k},
			\qquad c_0 := \inf_{\lambda \ll 1} z_P > 0.
		\end{equation}
		
		\item The classical cutoffs converge in Hilbert--Schmidt norm:
		\begin{equation}
			\|\gamma^{(k)}_{\nu_P} - \gamma^{(k)}_\nu\|_{\mathfrak S^2(H^{\otimes_s k})} \longrightarrow 0
			\qquad (\lambda \downarrow 0).
		\end{equation}
	\end{enumerate}
\end{proposition}

\begin{proof}
	By Wick's formula,
	\[
	\gamma^{(k)}_{\mu_0}
	:= \int |u^{\otimes k}\rangle \langle u^{\otimes k}|\, d\mu_0(u)
	= k!(h^{-1})^{\otimes_s k} \in \mathfrak S^2(H^{\otimes_s k}),
	\]
	because $\operatorname{Tr}(h^{-2})<\infty$.
	
	Since $\d\nu = g\, \d\mu_0$ with $g = z_{\mathrm{cl}}^{-1} e^{-D}$ and $D \ge 0$, we have
	\[
	0 \le g \le z_{\mathrm{cl}}^{-1}.
	\]
	Hence, for every positive finite-rank operator $A$ on $H^{\otimes_s k}$,
	\[
	\operatorname{Tr}(A \gamma^{(k)}_\nu)
	= \int \langle u^{\otimes k}, A u^{\otimes k}\rangle\, \d\nu(u)
	\le z_{\mathrm{cl}}^{-1}
	\int \langle u^{\otimes k}, A u^{\otimes k}\rangle\, \d\mu_0(u)
	= z_{\mathrm{cl}}^{-1}\operatorname{Tr}(A\gamma^{(k)}_{\mu_0}).
	\]
	Therefore $\gamma^{(k)}_\nu$ extends uniquely to a positive operator satisfying
	\[
	0 \le \gamma^{(k)}_\nu \le z_{\mathrm{cl}}^{-1}\gamma^{(k)}_{\mu_0}.
	\]
	Since the right-hand side belongs to $\mathfrak S^2(H^{\otimes_s k})$, the same is true for
	$\gamma^{(k)}_\nu$.
	
	Similarly, using the ambient-space representation
	\[
	\gamma^{(k)}_{\nu_P}
	= \int |(Pu)^{\otimes k}\rangle \langle (Pu)^{\otimes k}|\, g_P(u)\, \d\mu_0(u),
	\qquad g_P = z_P^{-1} e^{-D_P(P\cdot)},
	\]
	and Lemma~\ref{lem:class}, we have $0 \le g_P \le c_0^{-1}$. Thus
	\[
	0 \le \gamma^{(k)}_{\nu_P} \le c_0^{-1}\gamma^{(k)}_{\mu_0}
	= c_0^{-1} k!(h^{-1})^{\otimes_s k},
	\]
	and therefore $\gamma^{(k)}_{\nu_P} \in \mathfrak S^2(H^{\otimes_s k})$ as well.
	
	It remains to prove the Hilbert--Schmidt convergence.
We split
\begin{align}
	\gamma_{\nu_P}^{(k)}-\gamma_\nu^{(k)}
	&=
	\bigl(\gamma_{\nu_P}^{(k)}-P^{\otimes_s k}\gamma_\nu^{(k)}P^{\otimes_s k}\bigr)
	+
	\bigl(P^{\otimes_s k}\gamma_\nu^{(k)}P^{\otimes_s k}-\gamma_\nu^{(k)}\bigr).
	\label{eq:HS-classical-split}
\end{align}
For the first term, by duality in $\mathfrak S^2$, if $\|A\|_{\mathfrak S^2}\le1$ then
\[
\Tr\Bigl(A\bigl(\gamma_{\nu_P}^{(k)}-P^{\otimes_s k}\gamma_\nu^{(k)}P^{\otimes_s k}\bigr)\Bigr)
=
\int \langle (Pu)^{\otimes k},A(Pu)^{\otimes k}\rangle\,(g_P-g)\,\d\mu_0.
\]
Since
\[
|\langle (Pu)^{\otimes k},A(Pu)^{\otimes k}\rangle|
\le \|Pu\|^{2k},
\]
Cauchy--Schwarz gives
\begin{align*}
	&\|\gamma_{\nu_P}^{(k)}-P^{\otimes_s k}\gamma_\nu^{(k)}P^{\otimes_s k}\|_{\mathfrak S^2}\le
	\Bigl(
	\int \|Pu\|^{4k}\,(g_P+g)\,\d\mu_0
	\Bigr)^{1/2}
	\|g_P-g\|_{L^1(\mu_0)}^{1/2}.
\end{align*}
Now $g_P\le c_0^{-1}$ and $g\le z_{\rm cl}^{-1}$, so the Gaussian moment bound gives
\[
\int \|Pu\|^{4k}(g_P+g)\,\d\mu_0
\le C_k\bigl(\Tr(Ph^{-1}P)\bigr)^{2k}
\le C_k(1+|\log\lambda|)^{2k}.
\]
Therefore
\begin{equation}
	\label{eq:HS-classical-first-piece}
	\|\gamma_{\nu_P}^{(k)}-P^{\otimes_s k}\gamma_\nu^{(k)}P^{\otimes_s k}\|_{\mathfrak S^2}
	\le
	C_k(1+|\log\lambda|)^k\|g_P-g\|_{L^1(\mu_0)}^{1/2}.
\end{equation}
For the second term in \eqref{eq:HS-classical-split}, we have $\gamma_\nu^{(k)}\in\mathfrak{S}^2$, which implies
\[
\|P^{\otimes_s k}\gamma_\nu^{(k)}P^{\otimes_s k}-\gamma_\nu^{(k)}\|_{\mathfrak S^2}\to0.
\]
The result then follows by \eqref{con:rate}.
\end{proof}


Recall from Subsection~\ref{sec:upper} the low-frequency interacting Gibbs state
\[
\Gamma_{P,\lambda}^{\rm re}
=
\frac{\exp(-\lambda \dG(Ph)-W_P^{\rm re})}
{\Tr_{\gF(PH)}[\exp(-\lambda \dG(Ph)-W_P^{\rm re})]}
\]
and the corresponding trial state
\[
\widetilde\Gamma_\lambda
=
U^*\bigl(\Gamma_{P,\lambda}^{\rm re}\otimes \Gamma_{0,Q}\bigr)U.
\]

\begin{proposition}
	\label{prop:HS-trial-vs-true-state}
	With $\varepsilon_\lambda$ defined in \eqref{eq:HS-main-epsilon-lambda}, one has
	\begin{equation}
		\label{eq:HS-trial-vs-true-relative-entropy}
		\mathcal H(\widetilde\Gamma_\lambda,\Gamma_\lambda)
		\le C\varepsilon_\lambda,
	\end{equation}
	and consequently
	\begin{equation}
		\label{eq:HS-trial-vs-true-trace}
		\|\widetilde\Gamma_\lambda-\Gamma_\lambda\|_{\mathfrak S^1}
		\le C\varepsilon_\lambda^{1/2}.
	\end{equation}
	In particular,
	\begin{equation}
		\label{eq:HS-low-localized-trace}
		\|\Gamma_{P,\lambda}^{\rm re}-\Gamma_{\lambda,P}\|_{\mathfrak S^1}
		\le C\varepsilon_\lambda^{1/2}.
	\end{equation}
\end{proposition}

\begin{proof}
	Set
	\[
	\mathcal F_\lambda(\Gamma):=\mathcal H(\Gamma,\Gamma_0)+\Tr(W^{\rm re}\Gamma).
	\]
	By the variational principle \eqref{eq:variational},
	\[
	\mathcal F_\lambda(\Gamma_\lambda)=-\log\frac{Z_\lambda^{\rm re}}{Z_0}.
	\]
	Moreover, for every state $\Gamma$ one has
	\begin{equation}
		\label{eq:HS-free-energy-identity}
		\mathcal F_\lambda(\Gamma)-\mathcal F_\lambda(\Gamma_\lambda)
		=
		\mathcal H(\Gamma,\Gamma_\lambda).
	\end{equation}
	Indeed, since
	\[
	\Gamma_0=Z_0^{-1}e^{-\lambda \dG(h)},
	\qquad
	\Gamma_\lambda=(Z_\lambda^{\rm re})^{-1}e^{-\lambda \dG(h)-W^{\rm re}},
	\]
	we have
	\[
	\log \Gamma_\lambda
	=
	\log \Gamma_0-W^{\rm re}-\log\frac{Z_\lambda^{\rm re}}{Z_0}.
	\]
	Therefore
	\begin{align*}
		\mathcal H(\Gamma,\Gamma_\lambda)
		&=\Tr\bigl[\Gamma(\log\Gamma-\log\Gamma_\lambda)\bigr]
		\\
		&=\Tr\bigl[\Gamma(\log\Gamma-\log\Gamma_0)\bigr]
		+\Tr(W^{\rm re}\Gamma)
		+\log\frac{Z_\lambda^{\rm re}}{Z_0}
		\\
		&=\mathcal F_\lambda(\Gamma)+\log\frac{Z_\lambda^{\rm re}}{Z_0}
		=\mathcal F_\lambda(\Gamma)-\mathcal F_\lambda(\Gamma_\lambda),
	\end{align*}
	which proves \eqref{eq:HS-free-energy-identity}.
	
	Inspecting the proof of Proposition~\ref{prop:free-energy-upper-final} with the specific trial state $\widetilde\Gamma_\lambda$ shows more precisely that
	\[
	\mathcal F_\lambda(\widetilde\Gamma_\lambda)
	\le
	-\log \int_{PH}e^{-D_P(u)}\,\d\mu_{0,P}(u)
	+C\varepsilon_\lambda,
	\]
	whereas Theorem~\ref{thm:lower-bound-classical-final} gives
	\[
	\mathcal F_\lambda(\Gamma_\lambda)
	=
	-\log\frac{Z_\lambda^{\rm re}}{Z_0}
	\ge
	-\log \int_{PH}e^{-D_P(u)}\,\d\mu_{0,P}(u)
	-C\varepsilon_\lambda.
	\]
	Combining the two bounds with \eqref{eq:HS-free-energy-identity} for
	\(
	\Gamma=\widetilde\Gamma_\lambda
	\)
	proves \eqref{eq:HS-trial-vs-true-relative-entropy}.
	
	Pinsker's inequality implies
	\[
	\|\widetilde\Gamma_\lambda-\Gamma_\lambda\|_{\mathfrak S^1}
	\le
	\sqrt{2\mathcal H(\widetilde\Gamma_\lambda,\Gamma_\lambda)}
	\le C\varepsilon_\lambda^{1/2},
	\]
	which is \eqref{eq:HS-trial-vs-true-trace}. Finally, taking the partial trace over $\gF(QH)$ and using that the partial trace is trace-norm contractive, we obtain
	\[
	\|\Gamma_{P,\lambda}^{\rm re}-\Gamma_{\lambda,P}\|_{\mathfrak S^1}
	=
	\Bigl\|
	\Tr_{\gF(QH)}
	\bigl(
	U\widetilde\Gamma_\lambda U^*-U\Gamma_\lambda U^*
	\bigr)
	\Bigr\|_{\mathfrak S^1}
	\le
	\|\widetilde\Gamma_\lambda-\Gamma_\lambda\|_{\mathfrak S^1}
	\le C\varepsilon_\lambda^{1/2},
	\]
	which proves \eqref{eq:HS-low-localized-trace}.
\end{proof}

For the reduced density matrices it is more convenient to compare with the mixed product state
\[
\widehat\Gamma_\lambda
:=
U^*\bigl(\Gamma_{\lambda,P}\otimes\Gamma_{0,Q}\bigr)U.
\]

\begin{corollary}
	\label{cor:HS-true-vs-hat-state}
	One has
	\begin{equation}
		\label{eq:HS-true-vs-hat-trace}
		\|\Gamma_\lambda-\widehat\Gamma_\lambda\|_{\mathfrak S^1}
		\le C\varepsilon_\lambda^{1/2}.
	\end{equation}
	Moreover,
	\begin{equation}
		\label{eq:HS-high-localized-trace}
		\|(\Gamma_\lambda)_Q-\Gamma_{0,Q}\|_{\mathfrak S^1}
		\le C\varepsilon_\lambda^{1/2}
	\end{equation}
	and
	\begin{equation}
		\label{eq:HS-nearly-product-state}
		\Bigl\|
		\Gamma_\lambda-U^*\bigl(\Gamma_{\lambda,P}\otimes(\Gamma_\lambda)_Q\bigr)U
		\Bigr\|_{\mathfrak S^1}
		\le C\varepsilon_\lambda^{1/2}.
	\end{equation}
\end{corollary}

\begin{proof}
	Since
	\[
	U\widehat\Gamma_\lambda U^*-U\widetilde\Gamma_\lambda U^*
	=
	(\Gamma_{\lambda,P}-\Gamma_{P,\lambda}^{\rm re})\otimes \Gamma_{0,Q},
	\]
	Proposition~\ref{prop:HS-trial-vs-true-state} gives
	\[
	\|\widehat\Gamma_\lambda-\widetilde\Gamma_\lambda\|_{\mathfrak S^1}
	=
	\|\Gamma_{\lambda,P}-\Gamma_{P,\lambda}^{\rm re}\|_{\mathfrak S^1}
	\le C\varepsilon_\lambda^{1/2}.
	\]
	Together with \eqref{eq:HS-trial-vs-true-trace}, this proves \eqref{eq:HS-true-vs-hat-trace}.
	
	Taking the partial trace over $\gF(PH)$ yields
	\[
	\|(\Gamma_\lambda)_Q-\Gamma_{0,Q}\|_{\mathfrak S^1}
	\le
	\|\Gamma_\lambda-\widehat\Gamma_\lambda\|_{\mathfrak S^1}
	\le C\varepsilon_\lambda^{1/2},
	\]
	which is \eqref{eq:HS-high-localized-trace}. Finally,
	\begin{align*}
		&\Bigl\|
		\Gamma_\lambda-U^*\bigl(\Gamma_{\lambda,P}\otimes(\Gamma_\lambda)_Q\bigr)U
		\Bigr\|_{\mathfrak S^1}
		\\
		&\qquad\le
		\|\Gamma_\lambda-\widehat\Gamma_\lambda\|_{\mathfrak S^1}
		+
		\Bigl\|
		\Gamma_{\lambda,P}\otimes\Gamma_{0,Q}
		-
		\Gamma_{\lambda,P}\otimes(\Gamma_\lambda)_Q
		\Bigr\|_{\mathfrak S^1}
		\\
		&\qquad=
		\|\Gamma_\lambda-\widehat\Gamma_\lambda\|_{\mathfrak S^1}
		+
		\|\Gamma_{0,Q}-(\Gamma_\lambda)_Q\|_{\mathfrak S^1}
		\\
		&\qquad\le C\varepsilon_\lambda^{1/2}.
	\end{align*}
	This proves \eqref{eq:HS-nearly-product-state}.
\end{proof}


We use the following purely Fock-space estimate; see \cite[Lemma~11.7]{LewNamRou-21}.

\begin{lemma}[A trace-class estimate for reduced density matrices]
	\label{lem:HS-states-to-rdm-trace}
	Let $q,q'>1$ be conjugate exponents, and let $\Gamma,\Gamma'$ be two states on bosonic Fock space commuting with the number operator. Then for every $k\ge1$,
	\begin{equation}
		\label{eq:HS-states-to-rdm-trace}
		\|\Gamma^{(k)}-\Gamma'^{(k)}\|_{\mathfrak S^1}
		\le
		\|\Gamma-\Gamma'\|_{\mathfrak S^1}^{1/q'}
		\Bigl(\Tr[\cN^{qk}(\Gamma+\Gamma')]\Bigr)^{1/q}.
	\end{equation}
\end{lemma}

\begin{proposition}
	\label{prop:HS-state-to-rdm-scaled}
	For every fixed $k\ge1$,
	\begin{equation}
		\label{eq:HS-state-to-rdm-scaled-trace}
		\|\lambda^k(\Gamma_\lambda^{(k)}-\widehat\Gamma_\lambda^{(k)})\|_{\mathfrak S^1}
		\le
		C_k(1+|\log\lambda|)^k\varepsilon_\lambda^{1/4}.
	\end{equation}
	Consequently,
	\begin{equation}
		\label{eq:HS-state-to-rdm-scaled-HS}
		\|\lambda^k(\Gamma_\lambda^{(k)}-\widehat\Gamma_\lambda^{(k)})\|_{\mathfrak S^2}
		\le
		C_k(1+|\log\lambda|)^k\varepsilon_\lambda^{1/4}.
	\end{equation}
	In particular, since
	\begin{equation}
		\label{eq:HS-r-k-def}
		r_{\lambda,k}:=(1+|\log\lambda|)^k\varepsilon_\lambda^{1/4}\to0.
	\end{equation}
\end{proposition}

\begin{proof}
	Apply Lemma~\ref{lem:HS-states-to-rdm-trace} with $q=q'=2$, $\Gamma=\Gamma_\lambda$, and $\Gamma'=\widehat\Gamma_\lambda$. This gives
	\begin{equation}
		\label{eq:HS-state-to-rdm-proof-step}
		\|\lambda^k(\Gamma_\lambda^{(k)}-\widehat\Gamma_\lambda^{(k)})\|_{\mathfrak S^1}
		\le
		\|\Gamma_\lambda-\widehat\Gamma_\lambda\|_{\mathfrak S^1}^{1/2}
		\Bigl(
		\lambda^{2k}\Tr[\cN^{2k}(\Gamma_\lambda+\widehat\Gamma_\lambda)]
		\Bigr)^{1/2}.
	\end{equation}
	By Corollary~\ref{cor:HS-true-vs-hat-state},
	\[
	\|\Gamma_\lambda-\widehat\Gamma_\lambda\|_{\mathfrak S^1}^{1/2}
	\le C\varepsilon_\lambda^{1/4}.
	\]
	For the second factor, Proposition~\ref{prop:number-moments-log} with $t=0$ yields
	\[
	\lambda^{2k}\Tr(\cN^{2k}\Gamma_\lambda)
	\le C_k(1+|\log\lambda|)^{2k}.
	\]
	Moreover,
	\[
	U\widehat\Gamma_\lambda U^*=\Gamma_{\lambda,P}\otimes\Gamma_{0,Q},
	\qquad
	\cN=\cN_P+\cN_Q,
	\qquad
	\cN^{2k}\le 2^{2k-1}(\cN_P^{2k}+\cN_Q^{2k}).
	\]
	Hence
	\begin{align*}
		{\color{red}\lambda^{2k}\Tr(\cN^{2k}\widehat\Gamma_\lambda)}
		&\le
		{\color{red}C_k\bigl(\lambda^{2k}\Tr(\cN_P^{2k}\Gamma_{\lambda,P})+\lambda^{2k}\Tr(\cN_Q^{2k}\Gamma_{0,Q})\bigr)}
		\\
		&\le
		{\color{red}C_k\bigl(\lambda^{2k}\Tr(\cN^{2k}\Gamma_\lambda)+\lambda^{2k}\Tr(\cN^{2k}\Gamma_0)\bigr)}
		\\
		&\le {\color{red}C_k(1+|\log\lambda|)^{2k},}
	\end{align*}
{\color{red}by Proposition~\ref{prop:number-moments-log}.}
	Inserting these bounds into \eqref{eq:HS-state-to-rdm-proof-step} proves \eqref{eq:HS-state-to-rdm-scaled-trace}. The Hilbert--Schmidt bound \eqref{eq:HS-state-to-rdm-scaled-HS} follows from
	\(
	\|A\|_{\mathfrak S^2}\le \|A\|_{\mathfrak S^1}
	\).
\end{proof}


\begin{lemma}
	\label{lem:HS-tensor-product-rdm-expansion}
	Let $\Gamma_P$ be a state on $\gF(PH)$ and $\Gamma_Q$ a state on $\gF(QH)$. Define
	\[
	\Gamma:=U^*(\Gamma_P\otimes\Gamma_Q)U.
	\]
	Then for every $k\ge1$,
	\begin{equation}
		\label{eq:HS-tensor-product-rdm-expansion}
		\Gamma^{(k)}
		=
		\sum_{\ell=0}^k
		\binom{k}{\ell}
		\Gamma_P^{(\ell)}\otimes_s \Gamma_Q^{(k-\ell)}.
	\end{equation}
\end{lemma}

\begin{proof}
	The identity follows from the splitting formula
	\[
	Ua^*(f)U^*=a^*(Pf)\otimes 1+1\otimes a^*(Qf)
	\qquad (f\in\gH),
	\]
	and the analogous formula for annihilation operators. Expanding the $k$ creation and $k$ annihilation operators and grouping together the terms with exactly $\ell$ variables in the $P$-sector gives \eqref{eq:HS-tensor-product-rdm-expansion}. Pairing with arbitrary rank-one test operators on $\gH^{\otimes_s k}$ yields the operator identity.
\end{proof}

Applied to
\[
\widehat\Gamma_\lambda=U^*(\Gamma_{\lambda,P}\otimes\Gamma_{0,Q})U,
\]
Lemma~\ref{lem:HS-tensor-product-rdm-expansion} yields
\begin{equation}
	\label{eq:HS-hat-rdm-expansion}
	\widehat\Gamma_\lambda^{(k)}
	=
	\sum_{\ell=0}^k
	\binom{k}{\ell}
	\Gamma_{\lambda,P}^{(\ell)}\otimes_s \Gamma_{0,Q}^{(k-\ell)}.
\end{equation}

\begin{proposition}[Reduction to the low-frequency block, with rate]
	\label{prop:HS-hat-to-low-frequency}
	For every fixed $k\ge1$,
	\begin{equation}
		\label{eq:HS-hat-to-low-frequency}
		\|k!\lambda^k\widehat\Gamma_\lambda^{(k)}-k!\lambda^k\Gamma_{\lambda,P}^{(k)}\|_{\mathfrak S^2}
		\le C_k\lambda^\alpha(1+|\log\lambda|)^k.
	\end{equation}
	Consequently,
	\begin{equation}
		\label{eq:HS-true-to-low-frequency}
		\|k!\lambda^k\Gamma_\lambda^{(k)}-k!\lambda^k\Gamma_{\lambda,P}^{(k)}\|_{\mathfrak S^2}
		\le C_k\Bigl(r_{\lambda,k}+\lambda^\alpha(1+|\log\lambda|)^k\Bigr).
	\end{equation}
\end{proposition}

\begin{proof}
	Multiply \eqref{eq:HS-hat-rdm-expansion} by $\lambda^k$. Every term except the one with $\ell=k$ has the form
	\[
	\lambda^\ell\Gamma_{\lambda,P}^{(\ell)}\otimes_s \lambda^m\Gamma_{0,Q}^{(m)},
	\qquad m:=k-\ell\ge1.
	\]
	Using
	\[
	\|A\otimes_s B\|_{\mathfrak S^2}
	\le C_{k,\ell}\|A\|_{\mathfrak S^1}\|B\|_{\mathfrak S^2},
	\]
	we obtain
	\[
	\|\lambda^\ell\Gamma_{\lambda,P}^{(\ell)}\otimes_s \lambda^m\Gamma_{0,Q}^{(m)}\|_{\mathfrak S^2}
	\le
	C_{k,\ell}
	\|\lambda^\ell\Gamma_{\lambda,P}^{(\ell)}\|_{\mathfrak S^1}
	\|\lambda^m\Gamma_{0,Q}^{(m)}\|_{\mathfrak S^2}.
	\]
	Since $\Gamma_{\lambda,P}^{(\ell)}\ge0$,
	\[
	\|\lambda^\ell\Gamma_{\lambda,P}^{(\ell)}\|_{\mathfrak S^1}
	=
	\lambda^\ell\Tr(\Gamma_{\lambda,P}^{(\ell)})
	\le
	\lambda^\ell\Tr(\cN^\ell\Gamma_\lambda)
	\le
	C_\ell(1+|\log\lambda|)^\ell
	\]
	by Proposition~\ref{prop:number-moments-log}. On the other hand, 	since $\Gamma_0$ is quasi-free,
	\[
	\Gamma_0^{(m)}=\gamma_0^{\otimes m},
	\qquad
	\gamma_0=(e^{\lambda h}-1)^{-1}.
	\]
	Since $Q$ commutes with $h$, the localized free state on $\gF(QH)$ is again quasi-free, with
	one-body density matrix $Q\gamma_0Q$. Hence
	\[
	\Gamma_{0,Q}^{(m)}=(Q\gamma_0Q)^{\otimes m}.
	\]
	Therefore
	\[
	\|\lambda^m\Gamma_{0,Q}^{(m)}\|_{\mathfrak S^2}
	\le
	C_m\|\lambda Q\gamma_0Q\|_{\mathfrak S^2}^{\,m}.
	\]
	Moreover,
	\[
	0\le \lambda\gamma_0=\frac{\lambda}{e^{\lambda h}-1}\le h^{-1},
	\]
	hence
	\[
	\|\lambda Q\gamma_0Q\|_{\mathfrak S^2}
	\le
	\|Qh^{-1}\|_{\mathfrak S^2}.
	\]
	Finally,
	\[
	\|Qh^{-1}\|_{\mathfrak S^2}^2
	=
	\sum_{h(p)>\Lambda^2}h(p)^{-2}
	\le C\Lambda^{-2}
	\]
	in dimension two. Thus
	\[
	\|\lambda^m\Gamma_{0,Q}^{(m)}\|_{\mathfrak S^2}
	\le
	C_m\Lambda^{-m}
	\le
	C_m\Lambda^{-1},
	\]
	which proves
	\(
	\|\lambda^m\Gamma_{0,Q}^{(m)}\|_{\mathfrak S^2}\le C_m\lambda^\alpha
	\)
	for every fixed $m\ge1$. Therefore all terms with at least one high-frequency free factor are bounded by $C_k\lambda^\alpha(1+|\log\lambda|)^k$, and summing over the finitely many values of $\ell$ proves \eqref{eq:HS-hat-to-low-frequency}.
	
	Combining \eqref{eq:HS-hat-to-low-frequency} with Proposition~\ref{prop:HS-state-to-rdm-scaled} yields \eqref{eq:HS-true-to-low-frequency}.
\end{proof}


For every $k\ge1$ we define
\[
M_{\lambda,P}^{(k)}
:=
\int_{PH}|u^{\otimes k}\rangle\langle u^{\otimes k}|\,\d\mu^\lambda_{P,\lambda}(u)
\]
and
\[
\gamma_{\nu_P}^{(k)}
:=
\int_{PH}|u^{\otimes k}\rangle\langle u^{\otimes k}|\,\d\nu_P(u).
\]
Thus $M_{\lambda,P}^{(k)}$ is the $k$-th lower-symbol moment operator, whereas $\gamma_{\nu_P}^{(k)}$ is the finite-dimensional classical $k$-point correlation operator.

\begin{proposition}
	\label{prop:HS-deFinetti-remainder-vanish}
	For every fixed $k\ge1$,
	\begin{equation}
		\label{eq:HS-deFinetti-remainder-trace}
		\|k!\lambda^k\Gamma_{\lambda,P}^{(k)}-M_{\lambda,P}^{(k)}\|_{\mathfrak S^1}
		\le
		C_k\sum_{\ell=0}^{k-1}(\lambda K)^{k-\ell}(1+|\log\lambda|)^\ell.
	\end{equation}
	Hence also
	\begin{equation}
		\label{eq:HS-deFinetti-remainder-HS}
		\|k!\lambda^k\Gamma_{\lambda,P}^{(k)}-M_{\lambda,P}^{(k)}\|_{\mathfrak S^2}
		\le
		C_k\sum_{\ell=0}^{k-1}(\lambda K)^{k-\ell}(1+|\log\lambda|)^\ell.
	\end{equation}
	In particular, the right-hand side tends to $0$.
\end{proposition}

\begin{proof}
	Apply the quantitative de Finetti estimate \eqref{eq:quantitative} from Theorem~\ref{thm:quant deF} with $\Gamma=\Gamma_{\lambda,P}$. This yields
	\begin{align*}
		&\|k!\lambda^k\Gamma_{\lambda,P}^{(k)}-M_{\lambda,P}^{(k)}\|_{\mathfrak S^1}
		\\
		&\qquad\le
		\lambda^k
		\sum_{\ell=0}^{k-1}
		\binom{k}{\ell}^2
		\frac{(k-\ell+K-1)!}{(K-1)!}
		\Tr(\cN^\ell\Gamma_{\lambda,P}).
	\end{align*}
	Since
	\[
	\frac{(k-\ell+K-1)!}{(K-1)!}\le C_k K^{k-\ell},
	\qquad
	\Tr(\cN^\ell\Gamma_{\lambda,P})\le \Tr(\cN^\ell\Gamma_\lambda),
	\]
	and Proposition~\ref{prop:number-moments-log} gives
	\(
	\lambda^\ell\Tr(\cN^\ell\Gamma_\lambda)\le C_\ell(1+|\log\lambda|)^\ell
	\),
	we obtain \eqref{eq:HS-deFinetti-remainder-trace}. The Hilbert--Schmidt bound follows from
	\(
	\|A\|_{\mathfrak S^2}\le \|A\|_{\mathfrak S^1}
	\).
	Finally,
	\[
	\lambda K\sim \lambda\Lambda^2=\lambda^{1-2\alpha}\to0
	\qquad (0<\alpha<1/4),
	\]
	so the right-hand side tends to $0$.
\end{proof}

\begin{lemma}
	\label{lem:HS-full-norm-moment-mu}
	For every fixed integer $m\ge1$,
	\begin{equation}
		\label{eq:HS-full-norm-moment-mu}
		\int_{PH}\|u\|^{2m}\,\d\mu^\lambda_{P,\lambda}(u)
		\le
		C_m(1+|\log\lambda|)^m.
	\end{equation}
\end{lemma}

\begin{proof}
	By definition of $M_{\lambda,P}^{(m)}$,
	\[
	\Tr M_{\lambda,P}^{(m)}
	=
	\int_{PH}\Tr\bigl(|u^{\otimes m}\rangle\langle u^{\otimes m}|\bigr)\,\d\mu^\lambda_{P,\lambda}(u)
	=
	\int_{PH}\|u\|^{2m}\,\d\mu^\lambda_{P,\lambda}(u).
	\]
	On the other hand, Theorem~\ref{thm:quant deF} implies
	\[
	M_{\lambda,P}^{(m)}
	=
	m!\lambda^m\Gamma_{\lambda,P}^{(m)}+R_{\lambda,P}^{(m)}
	\]
	with
	\[
	\|R_{\lambda,P}^{(m)}\|_{\mathfrak S^1}
	\le
	C_m
	\sum_{\ell=0}^{m-1}
	(\lambda K)^{m-\ell}\lambda^\ell\Tr(\cN^\ell\Gamma_\lambda)
	\le C_m.
	\]
	Therefore
	\[
	\int_{PH}\|u\|^{2m}\,\d\mu^\lambda_{P,\lambda}(u)
	\le
	m!\lambda^m\Tr(\Gamma_{\lambda,P}^{(m)})+C_m
	\le
	m!\lambda^m\Tr(\cN^m\Gamma_\lambda)+C_m.
	\]
	Applying Proposition~\ref{prop:number-moments-log} completes the proof.
\end{proof}

\begin{lemma}
	\label{lem:HS-full-norm-moment-nuP}
	For every fixed integer $m\ge1$,
	\begin{equation}
		\label{eq:HS-full-norm-moment-nuP}
		\int_{PH}\|u\|^{2m}\,\d\nu_P(u)
		\le
		C_m(1+|\log\lambda|)^m.
	\end{equation}
\end{lemma}

\begin{proof}
	By definition,
	\[
	\d\nu_P(u)=z_P^{-1}e^{-D_P(u)}\,\d\mu_{0,P}(u),
	\qquad
	z_P:=\int_{PH}e^{-D_P(u)}\,\d\mu_{0,P}(u).
	\]
	Since $D_P\ge0$, we have $0\le e^{-D_P}\le1$. Moreover, Lemma \ref{lem:class} implies
	\(
	z_P\to z_{\rm cl}>0
	\),
	so there exists $c_0>0$ such that $z_P\ge c_0$ for all sufficiently small $\lambda$. Hence
	\[
	\int_{PH}\|u\|^{2m}\,\d\nu_P(u)
	\le
	c_0^{-1}\int_{PH}\|u\|^{2m}\,\d\mu_{0,P}(u).
	\]
	The right-hand side is the $2m$-th moment of a centered complex Gaussian with covariance $Ph^{-1}P$. By Wick's formula,
	\[
	\int_{PH}\|u\|^{2m}\,\d\mu_{0,P}(u)
	\le
	{\color{red}C_m\bigl(\Tr(Ph^{-1}P)\bigr)^m
	\le C_m(1+\log\Lambda)^m
	\le C_m(1+|\log\lambda|)^m.}
	\]
{\color{red}This implies \eqref{eq:HS-full-norm-moment-nuP}.}
\end{proof}

\begin{lemma}
	\label{lem:HS-L1-to-moment-operator}
	Let $\mu$ and $\eta$ be two probability measures on a finite-dimensional complex Hilbert space $E$. For every $k\ge1$, define
	\[
	M_\mu^{(k)}:=\int_E |u^{\otimes k}\rangle\langle u^{\otimes k}|\,\d\mu(u),
	\qquad
	M_\eta^{(k)}:=\int_E |u^{\otimes k}\rangle\langle u^{\otimes k}|\,\d\eta(u).
	\]
	Then
	\begin{equation}
		\label{eq:HS-L1-to-moment-operator}
		\|M_\mu^{(k)}-M_\eta^{(k)}\|_{\mathfrak S^2}
		\le
		\Bigl(
		\int_E\|u\|^{4k}\,\d\mu(u)+\int_E\|u\|^{4k}\,\d\eta(u)
		\Bigr)^{1/2}
		\|\mu-\eta\|_{L^1(E)}^{1/2}.
	\end{equation}
\end{lemma}

\begin{proof}
	By duality for Hilbert--Schmidt operators,
	\[
	\|M_\mu^{(k)}-M_\eta^{(k)}\|_{\mathfrak S^2}
	=
	\sup_{\|A\|_{\mathfrak S^2}\le1}
	\Bigl|
	\Tr\bigl(A(M_\mu^{(k)}-M_\eta^{(k)})\bigr)
	\Bigr|.
	\]
	For such an $A$,
	\[
	\Tr\bigl(A(M_\mu^{(k)}-M_\eta^{(k)})\bigr)
	=
	\int_E \langle u^{\otimes k},Au^{\otimes k}\rangle\,\d(\mu-\eta)(u).
	\]
	Moreover,
	\[
	|\langle u^{\otimes k},Au^{\otimes k}\rangle|
	\le
	\|A\|_{\mathfrak S^2}\|u^{\otimes k}\|^2
	\le
	\|u\|^{2k}.
	\]
	Hence Cauchy--Schwarz with respect to the positive measure $|\d\mu-\d\eta|$ gives
	\begin{align*}
		&\Bigl|
		\Tr\bigl(A(M_\mu^{(k)}-M_\eta^{(k)})\bigr)
		\Bigr|
		\\
		&\qquad\le
		\Bigl(
		\int_E\|u\|^{4k}\,|\d\mu-\d\eta|(u)
		\Bigr)^{1/2}
		\|\mu-\eta\|_{L^1(E)}^{1/2}
		\\
		&\qquad\le
		\Bigl(
		\int_E\|u\|^{4k}\,\d\mu(u)+\int_E\|u\|^{4k}\,\d\eta(u)
		\Bigr)^{1/2}
		\|\mu-\eta\|_{L^1(E)}^{1/2}.
	\end{align*}
	Taking the supremum over $A$ proves \eqref{eq:HS-L1-to-moment-operator}.
\end{proof}

\begin{proposition}
	\label{prop:HS-low-frequency-to-fd-classical}
	For every fixed $k\ge1$,
	\begin{equation}
		\label{eq:HS-M-lambdaP-to-gamma-nuP}
		\|M_{\lambda,P}^{(k)}-\gamma_{\nu_P}^{(k)}\|_{\mathfrak S^2}
		\le C_k r_{\lambda,k}.
	\end{equation}
	Consequently,
	\begin{equation}
		\label{eq:HS-low-frequency-to-fd-classical-main}
		\|k!\lambda^k\Gamma_{\lambda,P}^{(k)}-\gamma_{\nu_P}^{(k)}\|_{\mathfrak S^2}
		\le C_k\Bigl(r_{\lambda,k}+\sum_{\ell=0}^{k-1}(\lambda K)^{k-\ell}(1+|\log\lambda|)^\ell\Bigr).
	\end{equation}
\end{proposition}

\begin{proof}
	Apply Lemma~\ref{lem:HS-L1-to-moment-operator} with
	\(
	\mu=\mu^\lambda_{P,\lambda}
	\)
	and
	\(
	\eta=\nu_P
	\).
	Using Lemmas~\ref{lem:HS-full-norm-moment-mu} and \ref{lem:HS-full-norm-moment-nuP}, we obtain
	\[
	\|M_{\lambda,P}^{(k)}-\gamma_{\nu_P}^{(k)}\|_{\mathfrak S^2}
	\le
	C_k(1+|\log\lambda|)^k
	\|\mu^\lambda_{P,\lambda}-\nu_P\|_{L^1(PH)}^{1/2}.
	\]
	By Proposition~\ref{prop:weighted-lower-symbol},
we obtain \eqref{eq:HS-M-lambdaP-to-gamma-nuP}.
	Combining \eqref{eq:HS-M-lambdaP-to-gamma-nuP} with Proposition~\ref{prop:HS-deFinetti-remainder-vanish} yields \eqref{eq:HS-low-frequency-to-fd-classical-main}.
\end{proof}



\begin{proof}[Proof of Theorem~\ref{thm:HS-convergence-main}]
	Fix $k\ge1$.
	By Proposition~\ref{prop:HS-hat-to-low-frequency},
 Proposition~\ref{prop:HS-low-frequency-to-fd-classical} and Proposition~\ref{prop:HS-classical-cutoff-limit} the result follows.
\end{proof}

\appendix

\section{\rmk{Useful estimates and the periodic fractional-Bessel kernel}}\label{sec:appendix-yukawa}

\subsection{\rmk{Positivity of the periodic fractional-Bessel kernel and existence of the Gibbs state}}
{\rmk{In this subsection we collect the basic structural facts about the periodic fractional-Bessel kernel that are used throughout the paper. We first prove that the kernel is strictly positive away from the diagonal by means of its subordinated heat-kernel representation on the torus. This positivity is then inserted into the sector-wise realization of the renormalized Hamiltonian to obtain a coercive lower bound of order $N^2$, which implies that the Hamiltonian is bounded from below and that the corresponding grand-canonical Gibbs state is well defined.}}

\begin{lemma}[Positive lower bound away from the diagonal]
	\label{lem:periodic-yukawa-positive}
	{\rmk{Let}}
	\[
	{\rmk{\vbeta=(1-\Delta)^{-\beta/2}}}
	\qquad\text{on }\T^2=[0,2\pi]^2,
	\qquad {\rmk{\frac32<\beta\le2.}}
	\]
	Then $\vbeta(x)>0$ for every $x\in\T^2\setminus\{0\}$, and there exists a constant $\vbstar>0$ such that
	\[
	\vbeta(x)\ge \vbstar
	\qquad\text{for all }x\in\T^2\setminus\{0\}.
	\]
\end{lemma}

\begin{proof}
	{\rmk{Write the periodic fractional-Bessel kernel by subordination as}}
	\[
	{\rmk{\vbeta(x)=\frac1{\Gamma(\beta/2)}\int_0^\infty t^{\beta/2-1}e^{-t}p_t(x)\,\d t,}}
	\]
	where $p_t$ is the heat kernel on $\T^2$:
	\[
	p_t(x)
	=
	\frac1{(2\pi)^2}\sum_{k\in\Z^2}e^{-t|k|^2}e^{\ii k\cdot x}
	=
	\sum_{n\in\Z^2}\frac{1}{4\pi t}
	\exp\!\left(-\frac{|x+2\pi n|^2}{4t}\right).
	\]
	Fix $x\in\T^2$, viewed as a representative in $[-\pi,\pi]^2$. Then $|x|^2\le 2\pi^2$. Since every term in the periodized heat kernel is nonnegative, we have
	\[
	p_t(x)
	\ge
	\frac{1}{4\pi t}\exp\!\left(-\frac{|x|^2}{4t}\right)
	\ge
	\frac{e^{-\pi^2/2}}{4\pi t}
	\qquad (1\le t\le 2).
	\]
	Hence
	\[
	\vbeta(x)
	\ge
	{\rmk{\frac1{\Gamma(\beta/2)}\int_1^2 t^{\beta/2-1}e^{-t}p_t(x)\,\d t}}
	\ge
	{\rmk{\frac{e^{-2-\pi^2/2}}{4\pi\Gamma(\beta/2)}\int_1^2 t^{\beta/2-2}\,\d t>0.}}
	\]
	This proves the claim.
\end{proof}

\begin{proposition}[Existence of the grand-canonical Gibbs state]\label{prop:def-Gibbs}
	For every fixed $\lambda>0$, the formal Hamiltonian
	\[
	\mathbb{H}_\lambda^{\mathrm{re}}
	=
	\lambda \dG(h)+W^{\mathrm{re}}
	\]
defined in \eqref{eq:Hre}	admits a natural realization as a self-adjoint operator on bosonic Fock space,
	which is bounded from below. Moreover,
	\[
	Z_\lambda^{\mathrm{re}}
	:=
	\Tr_{\mathfrak F}\big(e^{-\mathbb{H}_\lambda^{\mathrm{re}}}\big)<\infty.
	\]
	Hence the grand-canonical Gibbs state
	\[
	\Gamma_\lambda^{\mathrm{re}}
	:=
	\frac{e^{-\mathbb{H}_\lambda^{\mathrm{re}}}}{Z_\lambda^{\mathrm{re}}}
	\]
	is well defined.
\end{proposition}

\begin{proof}
	We divide the proof into four steps.
	
	\medskip
	
	\noindent\textbf{Step 1: sector-wise realization.}
	Recall that
	\[
	\mathfrak F=\bigoplus_{n\ge0}\mathfrak H_n,
	\qquad
	\mathfrak H_n=L^2_s((\T^2)^n).
	\]
 Therefore the restriction of the formal Hamiltonian to $\mathfrak H_n$ is
	\begin{equation}\label{eq:Hn-re}
		\mathbb{H}_{\lambda,n}^{\mathrm{re}}
		=
		\lambda\sum_{j=1}^n(-\Delta_j+1)
		+
		\lambda^2\sum_{1\le i<j\le n}\vbeta(x_i-x_j)
		+
		c_\lambda n + C_\lambda.
	\end{equation}
	
	We define $\mathbb{H}_{\lambda,n}^{\mathrm{re}}$ rigorously through its quadratic form
	\begin{align*}
		q_{\lambda,n}[\Psi]
		:={}&
		\lambda\sum_{j=1}^n\|\nabla_j\Psi\|_{L^2((\T^2)^n)}^2
		+
		((\lambda+c_\lambda)n+C_\lambda)\|\Psi\|_{L^2((\T^2)^n)}^2
		\\
		&\quad
		+
		\lambda^2\sum_{1\le i<j\le n}
		\int_{(\T^2)^n}\vbeta(x_i-x_j)|\Psi(X_n)|^2\,\d X_n,
	\end{align*}
	with form domain
	\[
	Q_{\lambda,n}
	=
	H^1_s((\T^2)^n)
	\cap
	\left\{
	\Psi\in L^2_s((\T^2)^n):
	\sum_{1\le i<j\le n}\int \vbeta(x_i-x_j)|\Psi|^2<\infty
	\right\}.
	\]
	Here $X_n=(x_1,\dots,x_n)$ and $\d X_n=\d x_1\cdots \d x_n$.
	
	Since the potential term is nonnegative, the form $q_{\lambda,n}$ is densely defined and bounded from below on each fixed sector. The set
	\[
	C^\infty\!\big((\T^2)^n\setminus\mathcal D_n\big)\cap \mathfrak H_n,
	\qquad
	\mathcal D_n:=\{X_n:\ x_i=x_j\ \text{for some }i\neq j\},
	\]
	is dense in $\mathfrak H_n$ because $\mathcal D_n$ has measure zero. The potential form
	\[
	\Psi\mapsto
	\sum_{1\le i<j\le n}\int \vbeta(x_i-x_j)|\Psi|^2
	\]
	is closed as a nonnegative multiplication form. Hence $q_{\lambda,n}$ is closed and bounded below, so by the Friedrichs construction it determines a unique self-adjoint operator on $\mathfrak H_n$, still denoted by $\mathbb{H}_{\lambda,n}^{\mathrm{re}}$.
	
	Finally, define
	\[
	\mathbb{H}_\lambda^{\mathrm{re}}
	:=
	\bigoplus_{n\ge0}\mathbb{H}_{\lambda,n}^{\mathrm{re}}
	\]
	on Fock space. This gives a precise meaning to the formal expression
	\[
	\mathbb{H}_\lambda^{\mathrm{re}}=\lambda \dG(h)+\Wbeta+c_\lambda \cN+C_\lambda.
	\]
	By Lemma~\ref{lem:periodic-yukawa-positive}, there exists $\vbstar>0$ such that
	\[
	\vbeta(x)\ge \vbstar
	\qquad\text{for all }x\in\T^2\setminus\{0\}.
	\]
	Consequently, for almost every $X_n\in(\T^2)^n$,
	\[
	\sum_{1\le i<j\le n}\vbeta(x_i-x_j)
	\ge
	\frac{\vbstar}{2}n(n-1).
	\]

	\medskip
	
	\noindent\textbf{Step 2: coercive lower bound in the particle number.}
	Using \eqref{eq:Hn-re} and the previous step, we obtain on $\mathfrak H_n$
	\begin{align*}
		\mathbb{H}_{\lambda,n}^{\mathrm{re}}
		&\ge
		\lambda\sum_{j=1}^n(-\Delta_j)
		+
		\lambda n
		+
		\frac{\lambda^2\vbstar}{2}n(n-1)
		+
		c_\lambda n
		+
		C_\lambda
		\\
		&=
		\lambda\sum_{j=1}^n(-\Delta_j)
		+
		a_\lambda n^2
		+
		b_\lambda n
		+
		C_\lambda,
	\end{align*}
	where
	\[
	a_\lambda:=\frac{\lambda^2\vbstar}{2}>0,
	\qquad
	b_\lambda:=\lambda+c_\lambda-a_\lambda.
	\]
	Now
	\[
	a_\lambda n^2+b_\lambda n
	\ge
	\frac{a_\lambda}{2}n^2-\frac{b_\lambda^2}{2a_\lambda}
	=
	\frac{\lambda^2\vbstar}{4}n^2-\frac{b_\lambda^2}{\lambda^2\vbstar}.
	\]
yields
	\[
	\mathbb{H}_{\lambda,n}^{\mathrm{re}}
	\ge
	\lambda\sum_{j=1}^n(-\Delta_j)
	+
	\frac{\lambda^2\vbstar}{4}n^2
	-
	C_\lambda',
	\]
	where we may take
	\[
	C_\lambda':=\frac{b_\lambda^2}{\lambda^2\vbstar}+|C_\lambda|.
	\]
	Therefore, on Fock space,
	\begin{equation}\label{eq:global-lower-bound}
		\mathbb{H}_\lambda^{\mathrm{re}}
		\ge
		\lambda \dG(-\Delta)
		+
		\frac{\lambda^2\vbstar}{4}\cN^2
		-
		C_\lambda'.
	\end{equation}

	\medskip
	
	\noindent\textbf{Step 3: trace-class bound.}
	For fixed $n\ge0$, set
	\[
	K_{\lambda,n}
	:=
	\lambda\sum_{j=1}^n(-\Delta_j)
	+
	\frac{\lambda^2\vbstar}{4}n^2
	-
	C_\lambda'.
	\]
	By \eqref{eq:global-lower-bound},
	\[
	\mathbb{H}_{\lambda,n}^{\mathrm{re}}\ge K_{\lambda,n}
	\qquad\text{on }\mathfrak H_n.
	\]
	Therefore
	\begin{align*}
		\Tr_{\mathfrak H_n}\big(e^{-\mathbb{H}_{\lambda,n}^{\mathrm{re}}}\big)
	\leq
		\Tr_{\mathfrak H_n}\big(e^{-K_{\lambda,n}}\big)
		=
		e^{C_\lambda'}e^{-\frac{\lambda^2\vbstar}{4}n^2}
		\Tr_{\mathfrak H_n}\left(
		e^{-\lambda\sum_{j=1}^n(-\Delta_j)}
		\right).
	\end{align*}
	
	Now set
	\[
	A_\lambda
	:=
	\Tr_{L^2(\T^2)}\big(e^{-\lambda(-\Delta)}\big)
	=
	\sum_{k\in\Z^2}e^{-\lambda|k|^2}
	<\infty.
	\]
	Since $\mathfrak H_n$ is a subspace of the full tensor product $\bigotimes^nL^2(\T^2)$ and the operator is positive,
	\begin{align*}
		\Tr_{\mathfrak H_n}\left(
		e^{-\lambda\sum_{j=1}^n(-\Delta_j)}
		\right)
		&\le
		\Tr_{\otimes^nL^2(\T^2)}\left(
		(e^{-\lambda(-\Delta)})^{\otimes n}
		\right)
		\\
		&=
		\big(\Tr_{L^2(\T^2)}e^{-\lambda(-\Delta)}\big)^n
		=
		A_\lambda^n.
	\end{align*}
	Hence
	\[
	\Tr_{\mathfrak H_n}\big(e^{-\mathbb{H}_{\lambda,n}^{\mathrm{re}}}\big)
	\le
	e^{C_\lambda'}e^{-\frac{\lambda^2\vbstar}{4}n^2}A_\lambda^n.
	\]
	Summing over all particle numbers yields
	\begin{align*}
		Z_\lambda^{\mathrm{re}}
		=
		\Tr_{\mathfrak F}\big(e^{-\mathbb{H}_\lambda^{\mathrm{re}}}\big)
		&=
		\sum_{n\ge0}\Tr_{\mathfrak H_n}\big(e^{-\mathbb{H}_{\lambda,n}^{\mathrm{re}}}\big)
		\\
		&\le
		e^{C_\lambda'}\sum_{n\ge0}e^{-\frac{\lambda^2\vbstar}{4}n^2}A_\lambda^n
		<\infty.
	\end{align*}
	Since the Gaussian factor $e^{-\frac{\lambda^2\vbstar}{4}n^2}$ dominates the exponential growth $A_\lambda^n$ as $n\to\infty$, we conclude that $e^{-\mathbb{H}_\lambda^{\mathrm{re}}}$ is trace class. Therefore the Gibbs state
	\[
	\Gamma_\lambda^{\mathrm{re}}
	=
	\frac{e^{-\mathbb{H}_\lambda^{\mathrm{re}}}}{\Tr_{\mathfrak F}(e^{-\mathbb{H}_\lambda^{\mathrm{re}}})}
	\]
	is well defined.
\end{proof}

\subsection{Several useful lemmas}\label{app:exchange-bound}

\paragraph{Convolution kernels.}
\begin{lemma}\label{lem:sum}
	(i) (\cite[Lemma 3.10]{ZZ15}) Let $0<l,m<d$ and $l+m-d>0$. Then
	\begin{equation}\label{eq:sum-general}
		\sum_{k_1\in \mathbb{Z}^d}\frac{1}{\langle k_1\rangle^{l}\langle k-k_1\rangle^{m}}
		\lesssim \frac{1}{\langle k\rangle ^{l+m-d}}.
	\end{equation}

	(ii) In dimension $d=2$,
	\begin{equation}\label{eq:sum-log}
		\sum_{p\in\Z^2}\frac1{(|p+k|^2+1)(|p|^2+1)}
		\lesssim \frac{\log(2+|k|)}{1+|k|^2}.
	\end{equation}
\end{lemma}
\begin{proof}
	Set
	\[
	A_2(k):=
	\sum_{p\in\Z^2}\frac1{(|p+k|^2+1)(|p|^2+1)}.
	\]
	If $|k|\le 2$, then
	\[
	A_2(k)
	\le
	\sum_{p\in\Z^2}(|p|^2+1)^{-2}
	\le C.
	\]
	Assume now that $|k|>2$, and set $R:=|k|/2$. Decompose
	\[
	\Z^2=\Omega_1\cup\Omega_2\cup\Omega_3\cup\Omega_4,
	\]
	where
	\[
	\Omega_1:=\{p\in\Z^2:\ |p|\le R\},
	\qquad
	\Omega_2:=\{p\in\Z^2:\ |p+k|\le R\},
	\]
	\[
	\Omega_3:=\{p\in\Z^2:\ |p|>R,\ |p+k|>R,\ |p|\le 2|k|\},
	\qquad
	\Omega_4:=\{p\in\Z^2:\ |p|>2|k|\}.
	\]
	These four sets cover $\Z^2$.
	
	On $\Omega_1$ we have $|p+k|\ge |k|-|p|\ge R$, and therefore
	\[
	\sum_{p\in\Omega_1}\frac1{(|p+k|^2+1)(|p|^2+1)}
	\le
	\frac{C}{1+|k|^2}\sum_{|p|\le R}\frac1{|p|^2+1}
	\le
	C\,\frac{\log(2+|k|)}{1+|k|^2}.
	\]
	By the change of variables $q=p+k$, the same estimate holds on $\Omega_2$:
	\[
	\sum_{p\in\Omega_2}\frac1{(|p+k|^2+1)(|p|^2+1)}
	\le
	C\,\frac{\log(2+|k|)}{1+|k|^2}.
	\]
	
	On $\Omega_3$ we still have $|p+k|>R$, hence
	\[
	\sum_{p\in\Omega_3}\frac1{(|p+k|^2+1)(|p|^2+1)}
	\le
	\frac{C}{1+|k|^2}\sum_{R<|p|\le 2|k|}\frac1{|p|^2+1}
	\le
	C\,\frac{\log(2+|k|)}{1+|k|^2}.
	\]
	
	Finally, if $p\in\Omega_4$, then $|p+k|\ge |p|-|k|\ge |p|/2$, and therefore
	\[
	\frac1{(|p+k|^2+1)(|p|^2+1)}
	\le
	\frac{C}{(1+|p|^2)^2}.
	\]
	Since
	\[
	\sum_{|p|>2|k|}(1+|p|^2)^{-2}
	\le
	\frac{C}{1+|k|^2}
	\qquad\text{in dimension }2,
	\]
	we obtain
	\[
	\sum_{p\in\Omega_4}\frac1{(|p+k|^2+1)(|p|^2+1)}
	\le
	\frac{C}{1+|k|^2}.
	\]
	Combining the four regions proves \eqref{eq:sum-log}.
\end{proof}

\begin{lemma}[Convolution kernels]\label{lem:convolution-kernels}
	For \(s>0\), define
	\[
	S_k(s):=\sum_{u\in\mathbb Z^2} h(u+k)^{-s}h(u)^{-s}.
	\]
	For $L\ge1$, let
	\[
	\chi_u^{(L)}:=\mathbf 1_{\{h(u)\le L^2\}},
	\qquad
	S_k^{(L)}(s):=\sum_{u\in\mathbb Z^2} \chi_u^{(L)}\,h(u+k)^{-s}h(u)^{-s}.
	\]
	Then the following hold.
	\begin{align}
		S_k^{(L)}(s)&\le C_sL^{2-2s}\langle k\rangle^{-2s},
		&&0<s<1,\quad L\ge1,
		\label{eq:SkL-bound}
		\\
		S_k(s)&\le C_s,
		&&s>\frac12.
		\label{eq:Sk-bound}
	\end{align}
	In particular,
	\begin{align}
		S_k^{(P)}(s)&\le C_s\Lambda^{2-2s}\langle k\rangle^{-2s},
		&&0<s<1,
		\label{eq:SkP-bound}
		\\
		S_k^{(S)}(s)&\le C_s\Lambda_1^{2-2s}\langle k\rangle^{-2s},
		&&0<s<1.
		\label{eq:SkS-bound}
	\end{align}
\end{lemma}
\begin{proof}
	Since $h(u)\asymp \langle u\rangle^2$, \eqref{eq:Sk-bound} follows from \eqref{eq:sum-general}
	with $d=2$ and $l=m=2s$.
	
	For \eqref{eq:SkL-bound}, note first that
	\[
	|k|\le |u|+|u+k|,
	\]
	so at least one of $|u|$ and $|u+k|$ is bounded below by $|k|/2$. Consequently,
	\[
	h(u+k)^{-s}h(u)^{-s}
	\le
	C_s\langle k\rangle^{-2s}\bigl(h(u)^{-s}+h(u+k)^{-s}\bigr).
	\]
	Summing over the truncated region gives
	\[
	S_k^{(L)}(s)
	\le
	C_s\langle k\rangle^{-2s}
	\sum_{u\in\Z^2}\chi_u^{(L)}\bigl(h(u)^{-s}+h(u+k)^{-s}\bigr).
	\]
	We estimate the two sums on the right separately.
	
	Since $\chi_u^{(L)}=1$ implies $h(u)\le L^2$ and hence $|u|\lesssim L$, the standard lattice
	bound yields
	\[
	\sum_{u\in\Z^2}\chi_u^{(L)} h(u)^{-s}
	\le
	C_s\sum_{|u|\lesssim L}\langle u\rangle^{-2s}
	\le C_sL^{2-2s},
	\qquad 0<s<1.
	\]
	For the shifted term, change variables $v=u+k$:
	\[
	\sum_{u\in\Z^2}\chi_u^{(L)} h(u+k)^{-s}
	=
	\sum_{v\in\Z^2}\chi_{v-k}^{(L)} h(v)^{-s}.
	\]
	The support of $\chi_{v-k}^{(L)}$ is contained in a lattice ball $B(k,CL)$ of radius $CL$.
	For $j\ge0$, let
	\[
	\mathcal A_j:=\{v\in\Z^2:\ 2^j\le \langle v\rangle<2^{j+1}\}.
	\]
	Then
	\[
	\#\bigl(B(k,CL)\cap \mathcal A_j\bigr)
	\le
	C\min\{L^2,2^{2j}\}.
	\]
	Therefore
	\begin{align*}
		\sum_{v\in\Z^2}\chi_{v-k}^{(L)} h(v)^{-s}
		&\le
		C_s\sum_{j\ge0}2^{-2sj}\,\#\bigl(B(k,CL)\cap \mathcal A_j\bigr)
		\\
		&\le
		C_s\sum_{j\ge0}2^{-2sj}\min\{L^2,2^{2j}\}
		\\
		&\le
		C_s\Bigl(\sum_{2^j\le L}2^{(2-2s)j}+L^2\sum_{2^j>L}2^{-2sj}\Bigr)
		\le C_sL^{2-2s}.
	\end{align*}
	Combining the last three displays proves \eqref{eq:SkL-bound}. The special cases
	\eqref{eq:SkP-bound} and \eqref{eq:SkS-bound} follow by taking $L=\Lambda$ and
	$L=\Lambda_1$, respectively.
\end{proof}

%

{\color{red}
\begin{lemma}[Shifted fractional-Bessel $k$-sums]\label{lem:shifted-Yukawa-sums}
	Let $\frac32<\beta\le2$, $0<s<1$, and $0<s_1\le s_2<1$. Then
	\begin{align}
		\sum_{k\neq 0} |\hvb(k)|^2 \langle k+\ell\rangle^{-2s}
		&\le C_{\beta,s}\,\langle \ell\rangle^{-2s}, \label{eq:shifted-v-sum-1}\\
		\sum_{k\neq 0} |\hvb(k)|^2
		\langle k+\ell\rangle^{-2s_1}\langle \ell-k\rangle^{-2s_2}
		&\le C_{\beta,s_1,s_2}\,\langle \ell\rangle^{-2s_1}. \label{eq:shifted-v-sum-2}
	\end{align}
\end{lemma}

\begin{proof}
	Since $|\hvb(k)|^2=\langle k\rangle^{-2\beta}$ and $2\beta>3$, the kernel $|\hvb|^2$ belongs to $\ell^1(\Z^2)$.
	
	For \eqref{eq:shifted-v-sum-1}, assume first $|\ell|\ge1$ and split
	\[
	\Z^2=R_0\cup R_-\cup R_\infty,
	\qquad
	R_0:=\{|k|\le |\ell|/2\},
	\qquad
	R_-:=\{|k+\ell|\le |\ell|/2\},
	\]
	with $R_\infty$ the complement. On $R_0$ one has $|k+\ell|\ge |\ell|/2$, hence
	\[
	\sum_{k\in R_0}\langle k\rangle^{-2\beta}\langle k+\ell\rangle^{-2s}
	\le C_{\beta,s}\langle \ell\rangle^{-2s}.
	\]
	On $R_-$ write $k=-\ell+j$ with $|j|\le |\ell|/2$. Then $|k|\ge |\ell|/2$, so
	\[
	\sum_{k\in R_-}\langle k\rangle^{-2\beta}\langle k+\ell\rangle^{-2s}
	\le C_\beta\langle \ell\rangle^{-2\beta}\sum_{|j|\le |\ell|/2}\langle j\rangle^{-2s}
	\le C_{\beta,s}\langle \ell\rangle^{-2s}.
	\]
	On $R_\infty$ one has both $|k|\gtrsim |\ell|$ and $|k+\ell|\gtrsim |\ell|$, hence
	\[
	\sum_{k\in R_\infty}\langle k\rangle^{-2\beta}\langle k+\ell\rangle^{-2s}
	\le C_{\beta,s}\langle \ell\rangle^{-2s}\sum_{|k|\gtrsim |\ell|}\langle k\rangle^{-2\beta}
	\le C_{\beta,s}\langle \ell\rangle^{-2s}.
	\]
	This proves \eqref{eq:shifted-v-sum-1} for $|\ell|\ge1$, and the case $|\ell|<1$ is immediate after enlarging the constant.
	Furthermore,  \eqref{eq:shifted-v-sum-2} follows form \eqref{eq:shifted-v-sum-1}.
\end{proof}
}

\begin{proof}[Proof of Lemma \ref{lem:free-tail-Q1}]
	By definition,
	\[
	\lambda N_{0,Q_1}
	=
	\lambda\sum_{h(p)>\Lambda_1^2}\frac{1}{e^{\lambda h(p)}-1}.
	\]
	Using the standard lattice-to-integral comparison in dimension two, we obtain
	\[
	\lambda N_{0,Q_1}
	\lesssim
	\int_{\Lambda_1^2}^{\infty}\frac{\lambda\,\d r}{e^{c\lambda r}-1}.
	\]
	With the change of variables $x=\lambda r$, this becomes
	\[
	\lambda N_{0,Q_1}
	\lesssim
	\int_{\lambda\Lambda_1^2}^{\infty}\frac{\d x}{e^{cx}-1}
	=
	\frac1c\log\frac{1}{1-e^{-c\lambda\Lambda_1^2}}.
	\]
	This proves the first claim.
	
	Now since ${\alpha_1}>1/2$, we have
	\[
	\lambda\Lambda_1^2=\lambda^{1-2{\alpha_1}}\to\infty.
	\]
	Hence, for all sufficiently small $\lambda$,
	\[
	e^{-c\lambda\Lambda_1^2}\le \frac12,
	\]
	and therefore
	\[
	\log\frac{1}{1-e^{-c\lambda\Lambda_1^2}}
	\le
	2e^{-c\lambda\Lambda_1^2}.
	\]
	Thus
	\[
	\lambda N_{0,Q_1}\le C e^{-c\lambda\Lambda_1^2}.
	\]
	Since exponential decay dominates every power of $|\log\lambda|$, we conclude that
	\[
	(1+|\log\lambda|)^m\,\lambda N_{0,Q_1}\to 0
	\qquad\text{for every }m\ge 0.
	\]
\end{proof}

\bigskip
\noindent(P. T. Nam) LMU Munich, Department of Mathematics, Theresienstrasse 39, 80333 Munich, Germany

\noindent Email address: \texttt{nam@math.lmu.de}

\medskip
\noindent(R. Zhu) Department of Mathematics, Beijing Institute of Technology, Beijing 100081, China

\noindent Email address: \texttt{zhurongchan@126.com}

\medskip
\noindent(X. Zhu) Academy of Mathematics and Systems Science, Chinese Academy of Sciences, Beijing 100190, China

\noindent Email address: \texttt{zhuxiangchan@126.com}

\end{document}